\documentclass[11pt]{article}
\usepackage[T1]{fontenc}
\usepackage[utf8]{inputenc}
\usepackage{lmodern}
\usepackage[margin=0.9in]{geometry}
\usepackage{amsmath,amssymb,amsthm,bm,graphicx,booktabs,array,microtype,xcolor}
\usepackage[numbers,sort&compress]{natbib}
\usepackage{xurl}
\usepackage[colorlinks=true,linkcolor=blue!45!black,citecolor=blue!45!black,urlcolor=blue!45!black]{hyperref}
\newcommand{\CodeRepositoryURL}{https://github.com/Fumin111994/Unbounded-degree-overhead-for-Alice-conditioned-quantum-Bell-certificates}
\newcommand{\CodeRevision}{v0.3.0}
\newcommand{\DataRecordURL}{https://zenodo.org/records/23356474}

\newcommand{\codepath}[1]{\href{\CodeRepositoryURL/blob/\CodeRevision/#1}{\nolinkurl{#1}}}
\newcommand{\datapath}[1]{\href{\DataRecordURL}{\nolinkurl{#1}}}

\newcommand{\D}{\mathsf D}
\newcommand{\Os}{\mathsf O}
\newcommand{\Mix}{\mathsf M}
\newcommand{\eps}{\epsilon}

\title{The cost of direction in quantum Bell certification}
\author{Fumin Wang$^{1,2}$\thanks{Correspondence: \href{mailto:fwang1991@xjtu.edu.cn}{fwang1991@xjtu.edu.cn}}\\[0.5em]
\small $^1$MED-X Institute, The First Affiliated Hospital of Xi'an Jiaotong University,\\
\small Xi'an 710061, China\\
\small $^2$Shaanxi Key Laboratory of Quantum Information and Quantum Optoelectronic Devices,\\
\small College of Physics, Xi'an Jiaotong University, Xi'an 710049, China}
\date{October 14, 2026}
\begin{document}
\maketitle
\begin{abstract}
Certifying a quantum bound can require more information than realizing the correlations that attain it. We show that this distinction becomes unbounded when a proof must condition on one party's measurement question. A fixed Bell scenario attained by one entangled qubit pair has a standard degree-two sum-of-squares certificate, yet the minimum degree in either conditioning direction grows as the inverse square root of the distance to a local endpoint. Choosing the better party does not remove this cost. The obstruction also survives approximate certification: bounding the quantum maximum to any fixed fractional accuracy relative to its Bell advantage requires the same diverging level. Thus a short standard proof cannot be transferred to a uniformly bounded conditioned level, even after relaxing exactness and optimising the direction. Mixing directions within a proof removes this family's obstruction, while a separate finite separation establishes its limits. Contact geometry explains the cost and connects it to a sharp precision law for randomness certification.
\end{abstract}

% MAIN_TEXT_START
A quantum experiment and a proof of its limitations use different resources. An entangled pair of qubits may attain a Bell bound with only a few measurement settings, while certifying that no quantum strategy can exceed that bound requires relations among unmeasured operator products. The Navascu\'es--Pironio--Ac\'in (NPA) hierarchy organises these relations by word length~\cite{NPA2008}. Its dual sum-of-squares (SOS) certificates express the gap to a proposed bound as a sum of positive terms. A low-degree certificate therefore supplies a finite algebraic explanation of a universal quantum limit. The question is whether that explanation remains short when another certification framework can use only a restricted class of positive terms.

This is a concrete issue in compiled nonlocal games, where cryptography replaces spatial separation~\cite{KLVY2023,NatarajanZhang2023}. The framework of Cui, Falor, Natarajan and Zhang (CFNZ) uses ``nice'' SOS certificates, whose positive terms condition on a single measurement question of a designated party~\cite{CuiFalorNatarajanZhang2025}. Its convergence guarantees arbitrarily accurate bounds, but does not supply a uniform level at which a short standard certificate can be recovered. A certificate's permitted information structure is therefore an additional input to a resource estimate. Other soundness approaches, including sequential moment hierarchies, use different structures~\cite{KulpeEtAl2025,KlepEtAl2025}.

We establish an obstruction to uniform certificate transfer. Already for four binary questions per party and an attaining qubit pair, standard degree two is exact while the required conditioned degree diverges in \emph{both} orientations. No rule that selects the better party avoids this growth. Crucially, the same cost remains when the target is relaxed by any fixed fraction of the quantum-minus-local advantage. This gives a quantitative answer for a user who needs an approximate bound: the conditioned hierarchy must resolve increasingly long words to keep its algebraic uncertainty commensurate with the physical effect, although a standard degree-two certificate already determines the bound exactly.

The distinction also identifies a remedy and its scope. Permitting different positive terms to use different conditioning parties makes the coupled family easy. This freedom is different from choosing one direction for the whole certificate (Fig.~\ref{fig:resources}). A separate exact finite separation shows that mixed proofs do not universally recover the standard degree. The asymptotic coupled construction is the main result; the mixed example delineates which change of proof resource can evade it.

\begin{figure}[t]
\centering
\includegraphics[width=\textwidth]{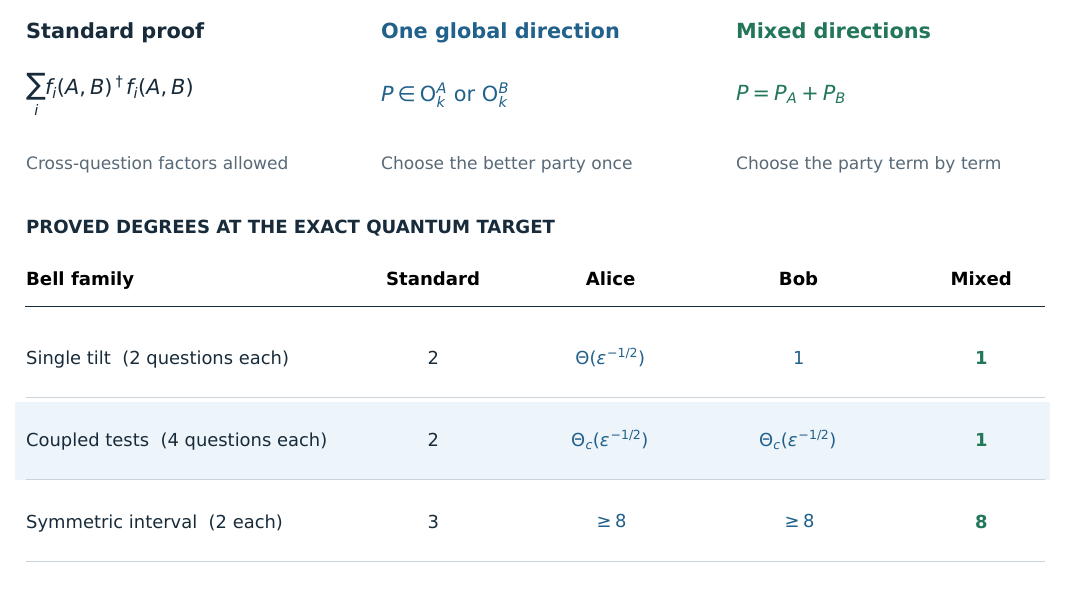}
\caption{\textbf{Proof structure is a resource beyond the physical system size.}
Standard SOS factors may combine questions from both parties. A globally conditioned proof fixes one party for every positive term; mixing permits that party to vary between terms. The highlighted coupled family is the central separation: standard degree two and mixed degree one coexist with diverging degree in either global direction, for a fixed four-question scenario. Here $\epsilon=2-\alpha\downarrow0$ and $0<c<1$ is fixed. The single-tilt row explains why choosing a direction can matter. The symmetric interval supplies a complementary finite separation at an attained, local quantum target; its individual degrees are only bounded below here. Standard degree counts total word length, and conditioned degree counts length on the unconditioned party.}
\label{fig:resources}
\end{figure}

\section*{A fixed quantum system with growing proof cost}
Begin with tilted CHSH,
\begin{equation}
 F_\alpha=\alpha A_0+A_0(B_0+B_1)+A_1(B_0-B_1),\qquad
 q(\alpha)=\sqrt{8+2\alpha^2},
 \label{eq:tilted}
\end{equation}
where $A_x,B_y$ have outcomes $\pm1$ and $0<\alpha<2$.
Its quantum bound $q$ has a standard degree-two certificate~\cite{Acin2012Tilted,BampsPironio2015}.
Our Alice-conditioned degree grows as $\Theta(\epsilon^{-1/2})$, with $\epsilon=2-\alpha$, whereas Bob-conditioned degree one suffices. A single tilted test therefore cannot establish a cost that survives the choice of party.

Couple two oppositely tilted tests:
\begin{align}
 T_\alpha&=\alpha B_2+A_2(B_2+B_3)+A_3(B_2-B_3),\nonumber\\
 G_{\alpha,c}&=F_\alpha+T_\alpha+c\epsilon A_0B_2,\qquad
 \beta_{\alpha,c}=2q(\alpha)+c\epsilon,
 \label{eq:coupled}
\end{align}
for fixed $0<c<1$. The bridge joins their correlator graphs (Fig.~\ref{fig:coupling}). One entangled qubit pair simultaneously attains both tilted bounds and unit bridge correlation. The two standard identities and the bridge identity give the exact quantum value $\beta_{\alpha,c}$ with minimum standard degree two. In contrast,
\begin{equation}
 d_A(G_{\alpha,c},\beta_{\alpha,c})
 =d_B(G_{\alpha,c},\beta_{\alpha,c})
 =\Theta_c(\epsilon^{-1/2}).
 \label{eq:double-cost}
\end{equation}
Each subcritical instance has finite exact certificates in both directions. What fails is a uniform bound on their degree as the parameter approaches the endpoint, despite fixed question number, bounded coefficients and a two-qubit attaining strategy.

The lower bound provides more than a failed search for certificates. At any fixed level and a suitable subcritical tilt, a complete collection of conditional moments passes all positivity and consistency tests while assigning a score above the true quantum maximum. We construct it by lifting a tilted-test witness together with a genuine quantum strategy for the opposite test. Positivity holds on every tested word, including words crossing the two setting pairs. Thus the relaxation itself permits the excess, independently of solver accuracy. A matching upper construction gives finite certificates at the same inverse-square-root order (Methods).

The bridge also has a controlled physical role. Allow each subtest its own bipartite quantum factor, including classical mixtures and departures from either subtest's optimum. The best score in this independent-subtest class, $\beta_{\rm ind}$, obeys
\begin{equation}
 \frac{\beta_{\alpha,c}-\beta_{\rm ind}}
 {\beta_{\alpha,c}-\beta_{\rm loc}}\longrightarrow c(2-c).
 \label{eq:physical-ratio}
\end{equation}
For $c=0.9$, independent subtest resources lose asymptotically $99\%$ of the available Bell advantage. This comparison quantifies the role of the shared realization within the specified resource class.

\begin{figure}[t]
\centering
\includegraphics[width=\textwidth]{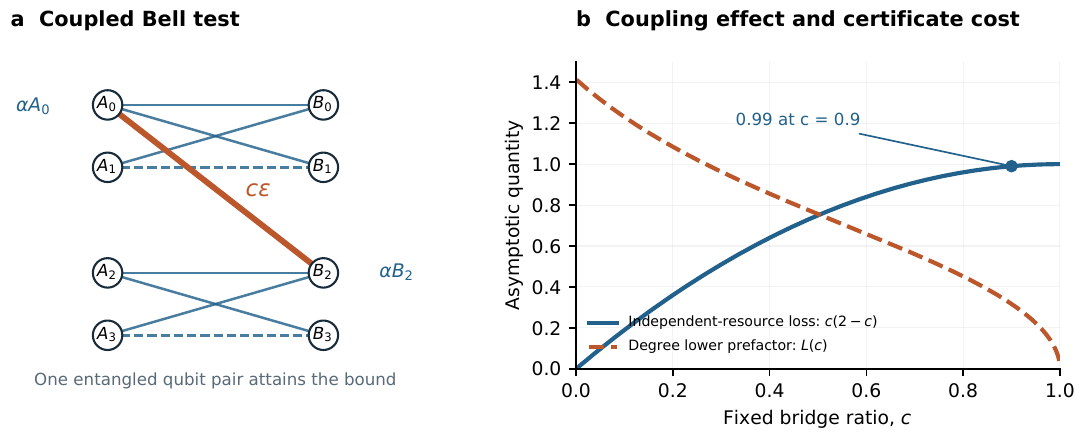}
\caption{\textbf{A shared qubit pair supports a coupled directional obstruction.}
\textbf{a}, Two oppositely tilted tests joined by a bridge of weight $c\epsilon$. All eight settings act on one qubit pair in the attaining strategy. Solid and dashed subtest edges denote positive and negative CHSH coefficients. \textbf{b}, The blue curve is the limiting fraction $c(2-c)$ of the quantum-minus-local advantage lost by independent subtest resources. The orange curve gives the lower prefactor $L(c)=\sqrt{2(1-c)/(1+2c-c^2)}$ for the best global exact degree. Both are proved analytic quantities. The bridge and the absolute advantage vanish as $\epsilon\to0$; the figure compares their effects on the available advantage scale. The endpoint $c=1$ and a fixed nonzero bridge are outside the divergence claim.}
\label{fig:coupling}
\end{figure}

\section*{The cost survives relative-accuracy requirements}
The approach to a local endpoint is a material part of the result: the bridge is $c\epsilon$, and the quantum advantage is
\begin{equation}
 V(\epsilon)=\beta_{\alpha,c}-\beta_{\rm loc}
 =2\{q(2-\epsilon)-(4-\epsilon)\}
 =\epsilon^2/4+O(\epsilon^3).
 \label{eq:advantage-main}
\end{equation}
A fixed absolute error can eventually exceed this entire advantage. The informative approximate task is therefore to bound the quantum maximum with an error budget $\rho V$, for a chosen fixed ratio $\rho>0$.

Let $E_k=\min(\omega_k^A,\omega_k^B)-\beta$ be the error of the better level-$k$ conditioned upper bound, and let $K_\rho$ be the smallest level for which $E_k\leq\rho V$. The existing coupled witness and exact upper construction also yield the following quantitative consequence for relative accuracy
\begin{equation}
 K_\rho(\epsilon,c)=\Theta_{c,\rho}(\epsilon^{-1/2})
                  =\Theta_{c,\rho}(V^{-1/4}),
 \qquad
 \liminf_{\epsilon\to0}\sqrt\epsilon K_\rho
 \geq\sqrt{\frac{2(1-c)}{1+2c-c^2+2\rho}}.
 \label{eq:relative-prefactor-main}
\end{equation}
This corollary makes explicit an accuracy consequence of the same construction; it does not introduce a separate witness or proof mechanism. The same formula includes exact certification at $\rho=0$. At every fixed level, $E_k/V$ diverges. Consequently, allowing a fixed fractional error does not restore a bounded conditioned level (Fig.~\ref{fig:relative}). The ratio is invariant under positive rescaling or shifting of the Bell score, so this conclusion is independent of its normalization.

This has two direct consequences. First, it rules out a uniform conversion from standard degree-two certificates to either conditioned system, even at relative accuracy and with the direction chosen for each instance. Second, it gives a necessary scale for selecting a hierarchy level before attempting such a certificate. A smaller numerical tolerance at a fixed level cannot remove the excess, because a feasible moment witness causes it. The result concerns resolving the quantum upper boundary; observing a Bell violation against a local bound is a different task.

These conclusions do not assume that an arbitrarily small advantage is experimentally accessible. They concern certification at a specified relative resolution, and do not establish finite-noise robustness or divergence at a nonvanishing bridge. The vanishing physical scale is accounted for explicitly rather than used as an absolute-error claim.

\begin{figure}[t]
\centering
\includegraphics[width=\textwidth]{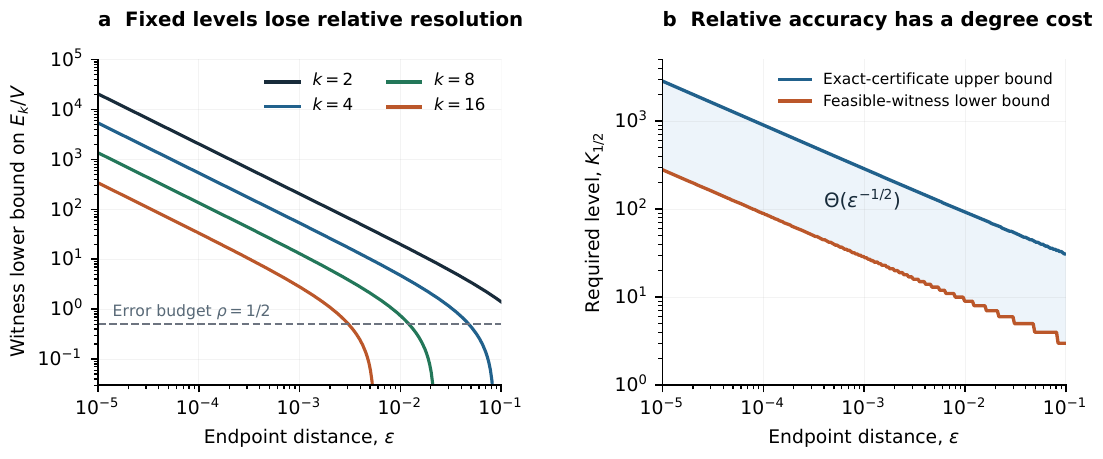}
\caption{\textbf{Approximate certification still needs growing conditioned levels.}
The coupled family has $c=1/8$. \textbf{a}, The explicit feasible witness gives $E_k/V\geq\max(0,D_k/V)$ for several levels. Only its positive portion is shown on logarithmic axes. The dashed line is the illustrative error budget $E_k\leq V/2$; a witness above this line rigorously excludes that level. \textbf{b}, A lower bound on $K_{1/2}$ obtained by excluding all levels below the witness threshold, and an upper bound from the exact certificate, $K_{1/2}\leq\lfloor9/\sqrt\epsilon+3\rfloor$. These bound curves are not estimates of the optimum. Standard degree two and mixed degree one remain exact throughout. The sharp exponent $K_{1/2}=\Theta(\epsilon^{-1/2})$ follows analytically.}
\label{fig:relative}
\end{figure}

\section*{What mixing directions can and cannot recover}
The coupled result concerns a global direction. Its first tilted gap has a short Bob-conditioned certificate, the second a short Alice-conditioned certificate, and the bridge is easy in either. Their sum therefore lies in the mixed degree-one cone:
\begin{equation}
 \beta_{\alpha,c}I-G_{\alpha,c}\in\Mix_1,
 \qquad \Mix_k=\Os^A_k+\Os^B_k.
 \label{eq:mixed-easy}
\end{equation}
The union of the two directional cones and their sum have sharply different power. This supplies a concrete change of admissible proof terms that removes the coupled obstruction, provided the intended certification framework supports both directions within one proof.

A complementary algebraic result limits that repair. A symmetric two-question family $J_t$, defined in Methods, has an attained target $\beta_t$ and
\begin{equation}
 \omega_Q(J_t)=\omega_L(J_t)=\beta_t,\qquad
 d_{\rm std}(J_t,\beta_t)=3,\qquad d_{\rm mix}(J_t,\beta_t)=8
 \label{eq:three-eight}
\end{equation}
throughout $0\leq t\leq10^{-5}$. Four matching rational certificates determine these minima: standard degree two and mixed degree seven are excluded, while degrees three and eight attain the target. The mixed lower witness uses one common Bell behaviour with compatible conditional moments in both directions; separate directional witnesses would not suffice.

The role of this example is a finite obstruction to degree-preserving transfer into the mixed system. Its target is locally attained, the certified interval is narrow, and no generic perturbation region is established. An auxiliary CHSH term preserves the degree pair while adding quantum advantage, but supplies that advantage independently; it is not the source of the structural separation. The construction and extension are given in Supplementary Sections~9.2 and~9.3. They neither replace the coupled family as the main physical example nor establish asymptotic mixed-degree divergence.

\section*{Contact geometry and certification precision}
The directional cost has an identifiable mechanism. A level-$k$ tilted-CHSH witness hides a deterministic response in a rare outcome of weight $(8k^2+2)^{-1}$ while retaining all tested conditional moments. Its unresolved scale is therefore $\epsilon\sim k^{-2}$. The upper proof encounters the same scale in a polynomial separator for two binary measurements. Positivity constraints touch zero at the attaining strategy, so an approximation must preserve those contacts. A contact-preserving construction cancels a nearby exterior pole and reaches degree $O(\epsilon^{-1/2})$.

For the singly tilted family, this geometry also yields the uniform error law
\begin{equation}
 \sup_{0\leq\alpha\leq2}
 [\omega_k^A(F_\alpha)-q(\alpha)]=\Theta(k^{-4}).
 \label{eq:precision}
\end{equation}
Exactness outside an $O(k^{-2})$ endpoint interval confines the remaining error to the quadratic chord defect of the quantum support curve. The rare-outcome witness attains the matching order.

This rate enters an operational certification task. At CHSH value $s$, the optimal quantum guessing probability for Alice's output is $G_Q(s)=(1+\sqrt{2-s^2/4})/2$~\cite{MasanesPironioAcin2011}. Replacing Eve-conditioned quantum behaviours by level-$k$ conditioned moment cones~\cite{NietoSillerasPironioSilman2014} gives a worst-case excess guessing probability and min-entropy deficit of $\Theta(k^{-4})$ over the full CHSH interval. Thus an additive accuracy $\delta$ uniformly over that interval costs $\Theta(\delta^{-1/4})$ levels in this relaxation. Standard degree two certifies the entire curve exactly. These single-round statements include arbitrary quantum side information, and refer to the singly tilted construction, not to an asserted randomness rate for the coupled family.

The mechanism also suggests a controlled repair. Giving the two CHSH correlator pairs unequal weights removes the singular contact and restores a uniform finite conditioned level at each fixed weight $\lambda>1$. Its cost grows as $\Theta((\lambda-1)^{-1/2})$ as symmetry is restored. Supplementary Section~7 gives the construction and the exact threshold $\lambda\geq\sqrt5/2$ for whole-family level-two closure.

\section*{Consequences for reusable certification}
The main lesson is a failure of uniform transfer between proof representations. The family keeps the physical realization and standard proof degree bounded, while the conditional moment information required for a chosen relative resolution grows without bound. It therefore separates knowing a quantum limit from expressing that limit with the positive terms a downstream certification argument is allowed to use. The relevant comparison is with limitations of the standard NPA hierarchy itself. Fanizza \emph{et al.} construct a game whose commuting-operator value is not reached at any finite NPA level~\cite{FanizzaEtAl2025}. For symmetric doubly tilted CHSH, Pakhunov constructs a feasible moment curve giving a quadratic excess near an endpoint where the quantum advantage is cubic~\cite{Pakhunov2026}; Chaturvedi uses a Motzkin-polynomial obstruction to show that every fixed standard NPA level has error unbounded relative to that advantage~\cite{Chaturvedi2026}. Here standard degree two is exact for every member: the diverging cost arises upon restricting the proof to a conditioned cone, even when either global direction is available. The relative-accuracy corollary quantifies that representation cost; it is not a new proof of the failure of standard NPA to describe the complete quantum set.

For a compiled-game argument that imports its Bell upper bound through a conditioned certificate, the hierarchy contributes algebraic slack in addition to any cryptographic error~\cite{CuiFalorNatarajanZhang2025}. Tightening the cryptographic estimate while holding the hierarchy level fixed cannot remove this slack. Our bounds quantify that obstruction within the stated proof route; they do not assert that a compiled prover can attain the relaxed score. Turning the degree law into a protocol resource bound still requires the relevant soundness estimate and coefficient control.

For hierarchy users, the constructive upper bound supplies a sufficient level, while the feasible witnesses show when increasing numerical accuracy at the current level cannot help. For framework designers, the mixed degree-one identity exhibits the precise extra freedom that removes this family's cost. Whether an equally strong asymptotic separation survives mixed terms, or a bridge bounded away from zero, remains open. 
% MAIN_TEXT_END

\clearpage
\section*{Methods}

% METHODS_START
\subsection*{Certificate conventions and exact targets}
We work in the free projective-measurement quotient with only cross-party commutation. Binary projectors are $P^A_{\pm|x}=(I\pm A_x)/2$, and similarly for Bob. Standard degree $k$ permits SOS factors of total reduced word length at most $k$. An Alice-conditioned term is $P^A_{a|x}f(B)^\dagger f(B)$ with Bob degree at most $k$; all four outcome blocks, or eight in the four-question construction, are included. Their moment matrices satisfy positivity, a common outcome sum across questions and total mass one. The reflection-to-projector change of basis is invertible and preserves the degree filtration.

For a target $\beta$, let $d_{\rm std}(G,\beta)=\min\{k:\beta I-G\in\D_k\}$, with the analogous definitions for $d_A,d_B,d_{\rm mix}$ and $\min\varnothing=\infty$. All levels start at one. Set $d_{\rm best}=\min(d_A,d_B)$. Since a conditioned term is the standard square $(fP^A_{a|x})^\dagger(fP^A_{a|x})$, one has $\Mix_k\subseteq\D_{k+1}$. Standard and conditioned levels count different word resources; in particular $d_{\rm mix}=1$ and $d_{\rm std}=2$ are consistent. In the coupled four-question scenario, the standard degree-two basis has 49 words. Each of the eight conditioned outcome blocks uses $2\cdot3^k-1$ independent words through degree $k$. These dimensions describe the complete word-basis formulation; they are not runtime lower bounds for implementations exploiting additional structure.

The standard and conditioned moment bodies are compact and strictly feasible on independent reduced bases, so their dual bounds are attained. All explicit targets in this paper have finite-dimensional attaining strategies and universal SOS upper identities. Thus the finite-dimensional tensor-product and commuting-operator quantum values coincide at these targets. The lower witnesses are relaxed moment functionals, not physical quantum behaviours. Proofs and the convention audit are in Supplementary Section~2.

\subsection*{Shared-qubit attainment and directional lower bounds}
Write $m=2\alpha/q$, $b=q/4$ and $d=\sqrt{1-b^2}$. On the state with Schmidt probabilities $(1\pm m)/2$, take $A_0=B_2=Z$, $A_1=B_3=X$, $B_0=A_2=bZ+dX$ and $B_1=A_3=bZ-dX$, on their respective parties. Both tilted terms attain $q$, and $\langle A_0B_2\rangle=1$. The BP15 identities~\cite{BampsPironio2015}, together with $I-A_0B_2=(A_0-B_2)^2/2$, give the standard degree-two upper bound. A feasible standard level-one Gram yields $4\sqrt2+2\alpha+c\eps>\beta_{\alpha,c}$, proving minimality.

For the lower bound, the tilted Fej\'er assemblage has $r=(8k^2+2)^{-1}$, Alice bias $h=1-2r$, and score $w=4-(1-2r)\eps$. Tensor it with a true quantum assemblage at effective tilt $t=\alpha+c\eps h$. The map from four Bob generators to the old two-generator algebra tensored with auxiliary matrices does not increase old word degree. Expanding each auxiliary matrix coefficient writes every new square as a sum of old tested squares, preserving all outcome sums. The score is exactly $w+q(t)$. If $\eps k^2\to C>0$, then
\begin{equation}
 k^4\left[w+q(t)-\beta_{\alpha,c}\right]
 \longrightarrow\frac{C[2(1-c)-(1+2c-c^2)C]}8.
\end{equation}
Consequently $\liminf_{\eps\downarrow0}\sqrt\eps\,d_{\rm best}\geq\sqrt{2(1-c)/(1+2c-c^2)}$. Party exchange with a setting permutation equates the two degrees. Combining the singly tilted upper certificate, the opposite-tilt degree-one identity and the bridge proves $d_A=d_B\leq12/\sqrt\eps+3$ for $6/\sqrt{23}\leq\alpha<2$, with prefactor 9 for $30/\sqrt{287}\leq\alpha<2$. The leading constants are not asserted optimal.

For independent subtest resources, the best score is
\begin{equation}
 \beta_{\rm ind}=\max_{u,v\in[0,1]}
 \{\alpha(u+v)+2\sqrt{2-u^2}+2\sqrt{2-v^2}+c\eps uv\}.
\end{equation}
The optimiser is symmetric and interior for sufficiently small $\eps$, with $u=v=1-(1-c)\eps/4+O_c(\eps^2)$. Hence $\beta-\beta_{\rm ind}=c(2-c)\eps^2/4+O_c(\eps^3)$, whereas $\beta-\beta_{\rm loc}=\eps^2/4+O(\eps^3)$. This proves equation~\eqref{eq:physical-ratio}. Supplementary Sections~3 and 4 give the universal proofs and an explicit separating sequence at $c=1/8$.

\subsection*{Accuracy relative to the Bell advantage}
For the coupled family, $V=\beta-\beta_{\rm loc}=2\epsilon^2/[q(\alpha)+\alpha+2]$. Let $r=(8k^2+2)^{-1}$, $h=1-2r$, $t=\alpha+c\epsilon h$ and $s=2-t$. Rationalizing $q(x)-(x+2)$ gives the feasible witness excess
\begin{equation}
 D_k=2(1-c)r\epsilon+\frac{s^2}{q(t)+t+2}-V,
 \qquad E_k\geq\max\{0,D_k\}.
 \label{eq:stable-relative-main}
\end{equation}
For $\epsilon k^2\to C>0$, $D_k/V\to(1-c)/C-(1+2c-c^2)/2$. Choosing $k=\lfloor v/\sqrt\epsilon\rfloor$ below the prefactor in equation~\eqref{eq:relative-prefactor-main} forces $D_k>\rho V$. Nestedness excludes all lower levels. The exact upper certificate supplies the matching order. For fixed $k$, $D_k=2(1-c)r\epsilon+O(\epsilon^2)$, so $D_k/V$ diverges. Supplementary Section~5 proves the complete corollary, including normalization invariance and word counts. Figure~\ref{fig:relative} evaluates this stable witness formula and the proved upper bound; it does not estimate SDP optima.

\subsection*{Four exact certificates for the mixed separation}
The symmetric two-question Bell family is
\begin{equation}
 J_t=(a+t)(A_0+B_0)-b(A_1+B_1)+c_0A_0B_0
       +d(A_0B_1+A_1B_0)-eA_1B_1,
 \label{eq:mixed-family}
\end{equation}
with $\beta_t=20000213/20000000+2t$ and
\begin{equation}
 (a,b,c_0,d,e)=10^{-8}(31096805,31141733,33401776,35456380,4405679).
\end{equation}
The coefficients range from approximately $0.044$ to $0.355$; $10^{-8}$ specifies their rational precision. The assignments $(A_0,A_1,B_0,B_1)=(1,-1,1,1)$ and $(1,1,1,-1)$ attain $\beta_t$. On the 25-word standard degree-three space, their evaluation matrix has a 23-dimensional kernel. A rational positive definite Gram on that kernel gives the identity $\beta_0I-J_0=T^*QT$, with 23 positive LDL pivots. A normalised standard level-two witness has a positive definite $13\times13$ Gram and excess greater than $0.0015$ over $\beta_t$ throughout the interval. Together these certify $d_{\rm std}=3$.

For the mixed lower bound, two complete assemblages share their eight Bell moments. Each contains four $15\times15$ Grams on all words through length seven, with all 120 rational LDL pivots positive. The two conditioned support spaces intersect exactly in $\operatorname{span}\{I,A_x,B_y,A_xB_y\}$, so their agreed moments define a functional positive on their sum. Its excess at $t=0$ is
\begin{equation}
 \Delta=\frac{593429101169281}{50000000000000000000},\qquad
 \ell(J_t)-\beta_t=\Delta-
 \frac{20024920147}{62500000000}t>8\times10^{-6}.
\end{equation}
The inequality holds for $0\leq t\leq10^{-5}$. For the matching upper bound, four rational positive definite Grams of sizes $15,17,16,16$ form an Alice-conditioned polynomial $R_A$ of degree eight. With $\sigma$ the party exchange, the full coefficient identity is $\beta_0I-J_0=R_A+\sigma(R_A)$. All 64 independent LDL pivots are positive. The factors obey the zeros forced by the two attaining assignments. Adding $t[(I-A_0)+(I-B_0)]$ transports both upper certificates to every $t\geq0$. These data establish an interval by exact identities and an affine margin, not by sampling parameters.

The data and three standard-library rational checkers are in \nolinkurl{research/mixed_divergence/}. Supplementary Section~9.2 specifies the identities and word bases; Supplementary Section~11 identifies the data and checkers. Deliberate changes to coefficients, Gram entries and moment consistency are rejected by the checkers. This verification establishes the stated finite separation without an assumed formula for its coefficients or a classification of mixed degree.

\subsection*{A degree-preserving nonlocal extension}
Let $C$ be CHSH on disjoint auxiliary questions and add $\eta C$ to a Bell operator $J$ with attained target $\beta$. The new quantum target is $\beta+2\sqrt2\eta$. The CHSH gap belongs to $\D_1,\Os^A_1,\Os^B_1$, proving every upper construction. Conversely, substitute optimal auxiliary CHSH matrices for the new generators and take expectation in their Bell state. This completely positive map sends the enlarged gap to $\beta I-J$. An expanded degree-$k$ factor yields a positive auxiliary Gram multiplying old words of degree at most $k$. An old conditioning projector is retained; a new projector produces a positive Bob-only or Alice-only SOS, which can be assigned to any old question by completeness. The map therefore preserves all four certificate cones at every $k\geq1$. This proves equality of the full degree profiles, rather than only a one-way upper bound. See Supplementary Section~9.3.

\subsection*{Contact-preserving approximation and precision}
Set $X=(B_0+B_1)/2$ and $Y=(B_0-B_1)/2$. A scalar polynomial separator $P$ for tilted CHSH obeys, for $-1\leq z\leq1$,
\begin{equation}
 P(z)\geq|\alpha+2z|,\qquad Q(z)=q-P(z)\geq0,\qquad
 Q(z)Q(-z)\geq4(1-z^2).
\end{equation}
Positivity of a Hermitian two-involution polynomial of length at most $2k$ implies an SOS with factors of length at most $k$. A closed-cone separation argument and positive-definite extension on the infinite dihedral group establish this degree-preserving step~\cite{BakonyiTimotin2011}. Alice completeness cancels the separator. Polynomial approximation preserving the forced contacts, together with localised Chebyshev pole cancellation, supplies the upper bound. The rare-outcome lower witness uses the Gram $M_{ij}=1-|i-j|/(2k)$ and $v_i=(-1)^i$, for which $v^TM^{-1}v=8k^2+2$. Supplementary Sections~6 and 7 give the full analytic constructions.

Convexity of the relaxed support function and the smooth quantum curve convert an $O(k^{-2})$ interval of possible non-exactness into $O(k^{-4})$ uniform error. The matching witness and the supporting tilted bounds transfer this rate to the guessing-probability relaxation. The logarithm has bounded derivative on guessing probabilities in $[1/2,1]$, giving the matching min-entropy rate. These randomness statements refer to the singly tilted CHSH family and the specified output, rather than asserting the same precision law for the mixed example. Supplementary Section~8 details quantum side information and the operational formulation.
% METHODS_END

\section*{Data availability}
The exact certificate data, analytic audit outputs, numerical controls and figure data are available in the version 0.3.0 Zenodo archive~\cite{Wang2026Data}. The archive includes the coupled and mixed certificates supporting the present article and a SHA-256 file manifest. Supplementary Section~11 maps the results to their public data directories and verification commands. The data are available under the Creative Commons Attribution 4.0 licence.

\section*{Code availability}
The verification, numerical search and figure-generation code is available under the MIT licence in the versioned GitHub repository~\cite{Wang2026Code}. Release \href{\CodeRepositoryURL/tree/\CodeRevision}{\texttt{v0.3.0}} matches the archived data. The README specifies the Python environment and dependencies, and distinguishes exact certificate verification from numerical discovery. Extracting the companion dataset into the code root makes all paths in the article and supplement executable.

\section*{Acknowledgements}
We acknowledge the support of the National Natural Science Foundation of China (Grants No.\ 92576203 and No.\ 12174301), the Natural Science Basic Research Program of Shaanxi (Grant No.\ 2023-JC-JQ-01), the Open Fund of the State Key Laboratory of Acoustics (Grant No.\ SKLA202312), the Shaanxi Province Postdoctoral Science Foundation (Grant No.\ 2024BSHEDZZ021), the Aeronautical Science Fund of China (Grant No.\ 2024M075070001), and the Project of the Institute of Medical Intelligence of the First Affiliated Hospital of Xi'an Jiaotong University.

\section*{Author contributions}
F.W. developed the framework and wrote the manuscript draft.

{\small
\setlength{\bibsep}{3pt}
\bibliographystyle{unsrtnat}
\bibliography{references}}
\end{document}

\maketitle

\begin{abstract}
This supplement gives the proofs, conventions and exact certificates for
directional costs in quantum Bell certification. The central construction
has bounded standard degree and diverging degree in either conditioning
direction. A relative-accuracy corollary quantifies the level needed to
resolve its quantum bound on the Bell-advantage scale. Analytic
contact geometry supplies matching upper bounds and randomness
consequences. Mixed-direction results specify the scope of the
obstruction. Auxiliary finite-level constructions are collected in one
additional-results section, followed by reproducibility instructions
and the relation to other hierarchies.
\end{abstract}

\setcounter{tocdepth}{1}
\tableofcontents
\clearpage
\section{Guide to the claims and their evidence}
\label{sec:intro}
The main result is a coupled Bell family whose standard degree stays at
two while the degree in either conditioning direction diverges. The same
obstruction persists when the allowed error is a fixed fraction of the
Bell advantage. The supplement first establishes these statements and
their analytic mechanism, then treats the mixed-direction qualification
and additional finite-level constructions.

\begin{center}
\small
\begin{tabular}{@{}p{.32\textwidth}p{.36\textwidth}p{.24\textwidth}@{}}
\toprule
Claim & Evidence & Location\\
\midrule
Standard degree two; both global directions diverge & Attaining qubit strategy, complete moment lift and matching upper bound & Sections~\ref{sec:bidirectional-coupling}, \ref{sec:bridge-strength-research}\\[4pt]
Fixed fractional error still costs $\Theta(V^{-1/4})$ levels & Explicit witness error divided by the Bell advantage & Corollary~\ref{cor:relative-precision}\\[4pt]
Inverse-square-root degree and fourth-power precision laws & Fej\'er functionals and contact-preserving approximation & Sections~\ref{sec:unbounded}, \ref{sec:structure}\\[4pt]
Randomness certification with quantum side information & Eve-conditioned optimisation and matching error bounds & Section~\ref{sec:randomness}\\[4pt]
Mixed degree one for the coupled family; exact finite pair $(3,8)$ elsewhere & Short mixed identity; four rational upper/lower certificates & Proposition~\ref{prop:bidirectional-mixed}, Section~\ref{sec:exact-mixed-eight}\\
\bottomrule
\end{tabular}
\end{center}

\paragraph{Reading routes.}
For the central result, read the conventions, the coupled construction
and the relative-accuracy corollary first. Sections~\ref{sec:unbounded}
and~\ref{sec:structure} supply the analytic upper certificates used there.
The mixed-direction section explains both the degree-one repair and its
finite limitations. Section~\ref{sec:additional} collects auxiliary
constructions and the exact finite-level identities needed by the proofs.
The reproduction section maps each principal result to executable checks.

\paragraph{Meaning of the accuracy statements.}
The coupled relative error is measured against its own quantum-minus-local
advantage, which tends to zero. It is invariant under score rescaling.
The uniform CHSH randomness error is a separate, additive-error statement
for the singly tilted family. Neither assertion is a protocol-independent
runtime, sample-complexity or cryptographic-security lower bound.
The mixed pair $(3,8)$ is a finite separation at a locally attained quantum
target; it does not establish mixed-degree divergence.

\section{Preliminaries}
\label{sec:prelim}

\subsection{Scenario and the PVM quotient algebra}

A $(2,2,2,2)$ Bell scenario has two parties, two questions each, two
outcomes each.  Work in the free $\ast$-algebra with generators
$A_{a|x}$, $B_{b|y}$ modulo the projective-measurement relations
$A_{a|x}^2=A_{a|x}$, $A_{a|x}A_{a'|x}=0$ for $a\neq a'$,
$\sum_a A_{a|x}=I$, the same for Bob, and $[A_{a|x},B_{b|y}]=0$.  In
computations we eliminate one outcome per measurement and keep the
remaining projectors as generators; \cref{sec:full} audits this choice with
an outcome-complete implementation.  The \emph{native level} of the
standard hierarchy is the total reduced word length; the native level of
the Alice-conditioned hierarchy is the reduced \emph{Bob} word length, with Alice's
$(x,a)$ carried as a block label.  These are different filtrations and we
never identify them without an explicit level-preserving map.
We call the block hierarchy \emph{Alice-conditioned}; it is the
``one-sided'' hierarchy of \cite{CuiFalorNatarajanZhang2025}, and retains
the notation $\omega_k^{\rm os}$, $\mathsf O_k$, and $d_{\rm os}$.
Here the qualifier describes certificate structure. A singly tilted
CHSH functional instead has only one marginal tilt; these are different
uses of ``one-sided.''

\subsection{The two hierarchies}

The standard level-$k$ relaxation maximizes $\langle \Bell_G,\Gamma\rangle$
over PSD moment matrices $\Gamma$ indexed by words of length $\le k$, with
normalization $\Gamma_{I,I}=1$ and moment consistency
$\Gamma_{s_1,t_1}=\Gamma_{s_2,t_2}$ whenever $s_1^\dagger t_1=s_2^\dagger
t_2$ in the quotient algebra \cite[Eq.~(2.7)]{CuiFalorNatarajanZhang2025}.
The Alice-conditioned level-$k$ relaxation \cite[Definition~4.1]{CuiFalorNatarajanZhang2025}
maximizes $\sum_{a,b,x,y}c_{abxy}\,\phi^k_{a|x}(B_{b|y})$ over PSD blocks
$\Gamma^k_{a|x}$ indexed by Bob words of length $\le k$, with the
consistency condition $\sum_a\Gamma^k_{a|x}$ independent of $x$ and
normalization $\sum_a\Gamma^k_{a|0}(I,I)=1$.  Its dual is the degree-$k$
\emph{nice} SOS cone: each square involves projectors of only one Alice
question \cite[Definition~3.1]{CuiFalorNatarajanZhang2025}.  The dual of
the sequential hierarchy of \cite{KlepEtAl2025} is the \emph{sparse} SOS
cone; on projectivized data the block presentations coincide, and we keep
the POVM/localizer and PVM-quotient conventions distinct throughout.

In real Bell-gap space, let $\mathsf D_k$ and $\mathsf O_k$ be the cones
of actual standard and Alice-conditioned degree-$k$ SOS certificates,
respectively. In this reduced binary scenario both coincide with their
closed Bell-gap duals by \cref{lem:attainment}. Thus
$\omega_k^{\rm std}(G)=\min\{\eta:\eta I-\Bell_G\in\mathsf D_k\}$,
and likewise for $\mathsf O_k$ at every finite native level.

\begin{definition}[Conversion degree at a target bound]
\label{def:convdeg}
For a Bell functional $G$ and a target bound $\beta$, the conversion
degrees are
\begin{equation*}
d_{\rm std}(G,\beta)=\min\{k:\beta I-\Bell_G\in\mathsf D_k\},
\qquad
d_{\rm os}(G,\beta)=\min\{k:\beta I-\Bell_G\in\mathsf O_k\},
\end{equation*}
with $\min\varnothing=\infty$. We write $d_{\rm os}(G)$ for
$d_{\rm os}(G,\omega_q(G))$ when the target is the quantum value, and
likewise $d_{\rm std}(G)$.  All conversion-degree statements below carry
their target bound explicitly.
\end{definition}

\begin{lemma}[Attainment at every finite native level]
\label{lem:attainment}
For every finite $k\geq1$, both normalized moment bodies in the reduced
raw-PVM $(2,2,2,2)$ scenario are compact and strictly feasible.
Strong duality holds with dual attainment, and the actual certificate
cones $\mathsf D_k$ and $\mathsf O_k$ in real Bell-gap space are closed.
\end{lemma}

\begin{proof}
For either a standard moment functional or a block functional $m$, write
an allowed word as $u=Pv$, $|u|\leq k$. Both $Pv$ and $(I-P)v$ belong to
the truncated word span. Their positivity gives
\begin{equation*}
0\leq m(u^\dagger u)=m(v^\dagger Pv)\leq m(v^\dagger v).
\end{equation*}
Iteration bounds all diagonals by $m(I)$, which is one for the standard
body and at most one for a normalized block. Cauchy--Schwarz bounds the
off-diagonal entries. Both feasible bodies are closed, hence compact.

For standard strict feasibility use the canonical group trace
$\tau=\tau_A\otimes\tau_B$ on
\begin{equation*}
\mathcal G=(\mathbb Z_2*\mathbb Z_2)_A\times
           (\mathbb Z_2*\mathbb Z_2)_B.
\end{equation*}
Reduced group words $w_Aw_B$ are orthonormal for
$\langle u,v\rangle=\tau(u^\dagger v)$. Replace each retained projector
by $P=(I+U)/2$. Ordered by total length, this change of basis is
triangular with diagonal $2^{-|w|}$: the leading involution word has the
same reduced letters, and all remaining terms have shorter length.
The projector-word trace Gram matrix is therefore positive definite for
every finite $k$, normalized, and consistent with the PVM relations.
For the Alice-conditioned body use the analogous Bob trace Gram matrix
$M_{\tau_B}$ and put $\Phi_{a|x}=M_{\tau_B}/2$ in each of the four blocks.
Their outcome sums coincide and have identity entry one, so this is also
strictly feasible. This argument uses the independent, reduced basis;
unreduced outcome-complete matrices have linear dependencies.

Slater's condition gives zero duality gap and dual attainment in both
cases. Adding a nonnegative multiple of $I$ preserves each SOS cone, so
$\mathsf C_k=\{\eta I-\Bell_G:\eta\geq\omega_k^{\mathsf C}(G)\}$
for $\mathsf C\in\{\mathsf D,\mathsf O\}$. Each support function is finite
and continuous on the finite-dimensional Bell coefficient space; its
epigraph, and hence the corresponding certificate cone, is closed.
\end{proof}

\paragraph{Explicit level-one checks.}
At level one the standard trace Gram matrix also equals the moment matrix
of the uniform mixture of the $16$ deterministic strategies. In the
retained projector basis $(I,P_{A,0},P_{A,1},P_{B,0},P_{B,1})$, it is
\begin{equation*}
\Gamma_0=\frac{1}{16}\sum_{v\in\{0,1\}^4}\binom{1}{v}\binom{1}{v}^{\!\top}
=\frac14\begin{pmatrix}
4&2&2&2&2\\
2&2&1&1&1\\
2&1&2&1&1\\
2&1&1&2&1\\
2&1&1&1&2
\end{pmatrix},
\end{equation*}
whose Schur complement of the top-left entry is $I_4/4$, so
$\det\Gamma_0=1/256>0$: strictly positive definite, normalized, and
moment-consistent.  For the Alice-conditioned level-one body take all four blocks
equal, in the Bob basis $(I,P_{B,0},P_{B,1})$:
\begin{equation*}
\Phi_{a|x}=\frac18\begin{pmatrix}4&2&2\\2&2&1\\2&1&2\end{pmatrix}
=\frac18\sum_{e\in\{0,1\}^2}u_eu_e^\top,\qquad u_e=(1,e_0,e_1),
\end{equation*}
whose Schur complement of the top-left entry is $I_2/8$, so
$\det\Phi_{a|x}=1/128>0$; the outcome sum has top-left entry $1$ and is
independent of $x$, so consistency and normalization hold.  Both matrices
are checked exactly (determinants and Schur complements) in
\codepath{scripts/verify_certificates.py}; the trace construction above
proves the all-level result.

\begin{remark}[Values versus cone membership]
\label{rem:cone-values}
By \cref{lem:attainment}, at every finite native level and for either
hierarchy, $\omega_k^{\mathsf C}(G)\leq\beta$ if and only if
$\beta I-\Bell_G\in\mathsf C_k$, with $\mathsf C\in\{\mathsf D,\mathsf O\}$.
Thus both membership statements mean actual SOS certificates under the
same reduced raw-PVM convention. Our quantitative separations additionally
supply explicit primal witnesses and dual certificates.
\end{remark}

\begin{table}[t]
\centering
\caption{Definitions audit: the conventions on both sides of the level-one
comparison of Corollary~\ref{cor:cfnz} and Remark~\ref{rem:52}, and why they coincide.  Every row
matches our implementation verbatim (\cref{sec:repro}); the only
genuinely distinct convention is the POVM/localizer filtration of
\cite{KlepEtAl2025}, which we never identify with the PVM quotient without
an explicit level-preserving map.}
\label{tab:audit}
\small
\begin{tabular}{>{\raggedright\arraybackslash}p{2.5cm}>{\raggedright\arraybackslash}p{3.4cm}>{\raggedright\arraybackslash}p{3.4cm}>{\raggedright\arraybackslash}p{3.4cm}}
\toprule
object & CFNZ v1 \cite{CuiFalorNatarajanZhang2025} & KPRSTXZ v2 \cite{KlepEtAl2025} & this paper \\
\midrule
standard moment matrix & $\Gamma^d$, Eq.~(2.7); basis $\mathbf b^d$ of PVM monomials of total degree $\le d$ & $\Gamma^{(n)}$, Eq.~(23); letters $\{A_{a|x},B_{b|y}\}$ of length $\le n$ & standard level $k$: total reduced word length \\
Alice-conditioned blocks & $\phi^d_{ax}$, Def.~4.1; Bob words of degree $\le d$ & $\Theta^{(n)}(a|x)$, Eq.~(22); Bob letters of length $\le n$ & Alice-conditioned level $k$: Bob reduced word length \\
consistency & $\sum_a\phi_{ax}$ independent of $x$ & strong no-signaling, $\sum_a\Theta(a|x)=\Theta(a|x')$ & identical constraint \\
normalization & $\sum_a\phi_{a0}(I)=1$ & $\Theta_{1,1}=1$ & identical constraint \\
outcomes & PVM; one outcome per measurement removed in the proof of Thm.~5.3 & POVM via localizers; projectivized for duality (Sec.~3.3) & reduced PVM quotient; outcome-complete audit, \cref{sec:full} \\
dual cone & nice SOS (Def.~3.1) & sparse SOS (Eq.~(26)) & $\mathsf O_k$ via the block dual; projectivized data coincide \\
probability positivity & no separate inequalities & POVM bounds via localizers & implied at standard $k\geq2$; see below \\
\bottomrule
\end{tabular}
\end{table}

\paragraph{Joint-probability positivity at standard level two.}
For every pair of outcomes, $R=A_{a|x}B_{b|y}$ belongs to the standard
level-two word span, including after outcome elimination. Commutation and
idempotency give $R^\dagger R=R$, so moment positivity implies
$p(ab|xy)=m(R)=m(R^\dagger R)\geq0$. Consequently adding explicit joint
probability inequalities leaves the raw standard level-$k$ body unchanged
for every $k\geq2$. This resolves the probability-positivity convention
for our level-two comparison without a separate numerical variant.

\begin{remark}
\label{rem:audit}
Two consequences of \cref{tab:audit}.  First, our Alice-conditioned hierarchy is
literally CFNZ Definition~4.1 over the reduced PVM quotient (the outcome
removal is the same one used in the proof of their Theorem~5.3), so the
comparison with their Section~5.2 in \cref{cor:cfnz} is made at their own
level convention.  Second, a caution concerning
\cite[Section~3]{KlepEtAl2025}: their sequential level-one blocks also
carry the degree-three moments of \cref{rem:52}, so a standard-to-Alice-conditioned
restriction at level one cannot be read off verbatim in their system
either.  We do \emph{not} claim a correction under their POVM/localizer
convention: the level mapping between the localizer filtration and the raw
PVM quotient must be established separately, and the POVM bounds added via
localizers could in principle repair the restriction step; the phrasing
should be re-verified directly in their constraint system.
\end{remark}

\begin{proposition}[{Known input: \cite[Theorem~5.3]{CuiFalorNatarajanZhang2025}}]
\label{prop:d1}
$\mathsf D_1\subseteq\mathsf O_1$: every standard degree-one SOS
certificate can be converted into a nice degree-one certificate.
\end{proposition}

\begin{remark}[The failed strengthening]
\label{rem:52}
\cite[Section~5.2]{CuiFalorNatarajanZhang2025} states that the two
level-one primal values are equal.  Its proof opens with the sentence
\begin{quote}
``This restricts to a feasible solution for Equation~(4.4) with the same
value.''
\end{quote}
--- the claim that a standard level-one feasible point restricts to an
Alice-conditioned level-one feasible point.  That sentence is unproved, and under
the convention exactly as stated in v1 (Table~\ref{tab:audit}) it is
false: an Alice-conditioned level-one block contains the quadratic Bob moments
$\phi_{a|x}(B_{b|y}^\dagger B_{b'|y'})$ --- total degree three once the
Alice label is counted --- which the standard level-one moment body (total
degree at most two) does not contain.  There is nothing to restrict.  The
remainder of their Section~5.2 proves the \emph{opposite} direction
(Alice-conditioned to standard, by a Gram-vector construction), which survives and
is consistent with all our computations.  \cref{thm:L1-sep} gives an exact
refutation of the equality.
\end{remark}

\begin{figure}[t]
\centering
\includegraphics[width=0.92\textwidth]{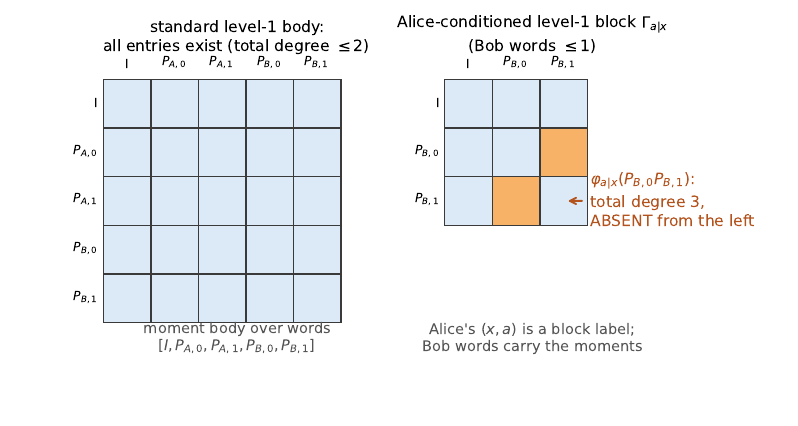}
\caption{The level-one filtration obstruction (\cref{rem:52}).  Left: the
standard level-one moment body contains moments of total degree at most
two.  Right: an Alice-conditioned level-one block $\Gamma_{a|x}$, indexed by Bob
words of length at most one, contains the corner moments
$\varphi_{a|x}(P_{B,0}P_{B,1})$ --- total degree three once the Alice
label is counted --- which the standard level-one body cannot supply.
}
\label{fig:filtration}
\end{figure}

\section{A coupled family with diverging degree in both directions}
\label{sec:bidirectional-coupling}

This section enlarges the measurement scenario: each party has four
binary questions, with reflections $A_0,\ldots,A_3$ and
$B_0,\ldots,B_3$. Only cross-party commutation is imposed.
Standard degree is total reduced word length. Alice-conditioned
degree is Bob word length in eight outcome blocks, and Bob-conditioned
degree is defined with the parties exchanged. Denote the corresponding
cones by $\mathsf D_k^{(4,4)}$, $\mathsf O_k^A$, and $\mathsf O_k^B$.
Their degrees below are at the stated exact quantum target.

For $0<\alpha<2$, set
\begin{align}
 q&=\sqrt{8+2\alpha^2},\qquad
 \epsilon=2-\alpha,\qquad \eta=\epsilon/8,\notag\\
 F^{(1)}_\alpha
 &=\alpha A_0+A_0(B_0+B_1)+A_1(B_0-B_1),\notag\\
 F^{(2)}_\alpha
 &=\alpha B_2+A_2(B_2+B_3)+A_3(B_2-B_3),\notag\\
 G_\alpha&=F^{(1)}_\alpha+F^{(2)}_\alpha+\eta A_0B_2,\qquad
 \beta_\alpha=2q+\eta.
 \label{eq:bidirectional-family}
\end{align}
The cross correlator $A_0B_2$ connects the two CHSH setting graphs.
The scenario is fixed and all Bell coefficients are uniformly bounded.
The coupling is positive at every subcritical tilt; it tends to zero
at the limiting tilt. This qualification is essential.

\begin{theorem}[Coupled, orientation-independent degree divergence]
\label{thm:bidirectional-coupling}
For the family in \eqref{eq:bidirectional-family}, the quantum value is
$\beta_\alpha$ and the standard exact degree is two. Both conditioned
degrees are finite for every $0<\alpha<2$ and satisfy
\begin{equation}
 d_A(G_{2-\epsilon},\beta_{2-\epsilon})
 =d_B(G_{2-\epsilon},\beta_{2-\epsilon})
 =\Theta(\epsilon^{-1/2})\qquad(\epsilon\downarrow0).
 \label{eq:bidirectional-sharp}
\end{equation}
In particular,
\begin{equation}
 \liminf_{\epsilon\downarrow0}\sqrt\epsilon\,
 \min\{d_A,d_B\}\geq\sqrt{7/5}.
 \label{eq:bidirectional-constant}
\end{equation}
For each $k\geq1$, take
$r_k=(8k^2+2)^{-1}$ and $\alpha_k=2-8r_k$. Then
\begin{equation}
 \beta_{\alpha_k}I-G_{\alpha_k}
 \in\mathsf D_2^{(4,4)}\setminus
       (\mathsf O_k^A\cup\mathsf O_k^B).
 \label{eq:bidirectional-cones}
\end{equation}
An explicit feasible level-$k$ witness for each direction exceeds
$\beta_{\alpha_k}$ by more than $5449r_k^2/5000$.
The quantum value is attained on a single pair of entangled qubits;
independent optimal implementations of the two subtests cannot attain
the coupled bound.
\end{theorem}

\begin{proof}
We prove the quantum target, the standard degree, and the conditioned
bounds separately.

\paragraph{The quantum value and a standard degree-two identity.}
Let $Z$ and $X$ be the real Pauli reflections and put
\begin{equation}
 m=2\alpha/q,\quad b=q/4,\quad d=\sqrt{1-b^2},\quad
 s=\sqrt{1-m^2},\quad
 |\psi\rangle=\sqrt{(1+m)/2}|00\rangle+
               \sqrt{(1-m)/2}|11\rangle.
 \label{eq:bidirectional-strategy}
\end{equation}
Here $0<m<1$, $d=bs$, and $b^2+d^2=1$. Choose
\begin{equation}
 A_0=B_2=Z,\quad A_1=B_3=X,\quad
 B_0=A_2=bZ+dX,\quad B_1=A_3=bZ-dX,
 \label{eq:bidirectional-observables}
\end{equation}
where the operators are on their respective local qubits.
Both tilted subtests have value $\alpha m+2b+2ds=q$, and
$\langle A_0B_2\rangle=1$. Thus this one shared state attains
$\beta_\alpha$.

The following identity proves the bound in the original quotient
algebra, without assumptions on dimension or same-party commutation.
Define
\begin{align*}
 R_i&=qI-F^{(i)}_\alpha\quad(i=1,2),\\
 T_1&=\alpha A_1-A_0(B_0-B_1)-A_1(B_0+B_1),\\
 T_2&=\alpha B_3-B_2(A_2-A_3)-B_3(A_2+A_3).
\end{align*}
The tilted-CHSH identity \eqref{eq:bp15}, in its two orientations,
gives
\begin{equation}
 \beta_\alpha I-G_\alpha
 =\frac{R_1^2+T_1^2+R_2^2+T_2^2}{2q}
  +\frac\eta2(A_0-B_2)^2.
 \label{eq:bidirectional-standard-sos}
\end{equation}
Every factor has total degree at most two. This proves both
$\omega_Q(G_\alpha)=\beta_\alpha$ and membership in
$\mathsf D_2^{(4,4)}$.

Standard level one is not exact. Assign unit Gram vectors
$I,A_0,B_2$ to $z$, $A_1,B_3$ to $x$, and
$B_0,A_2$ and $B_1,A_3$ to $(z+x)/\sqrt2$ and $(z-x)/\sqrt2$,
respectively, for orthonormal $z,x$.
This is a feasible involution moment matrix, and the invertible PVM
basis change gives the original projector convention. Its value is
$4\sqrt2+2\alpha+\eta>2q+\eta$ for $0<\alpha<2$, since
$(2\sqrt2+\alpha)^2-q^2=\alpha(4\sqrt2-\alpha)>0$.
The standard exact degree is therefore exactly two.

\paragraph{Lifting an entire truncated assemblage.}
We require positivity on every four-generator Bob word, rather than
only a Bell projection. Let $\phi_{a|x}$, $x=0,1$, be any feasible
level-$k$ assemblage on the algebra of $B_0,B_1$, with common sum $L$.
Let $\nu_{a|x}$, $x=2,3$, be a genuine quantum assemblage on two new
reflections $C_0,C_1$, with common state $\nu$.
There is a degree-nonincreasing unital star homomorphism
\begin{equation*}
 \pi:\mathbb C[\mathbb Z_2^{*4}]\longrightarrow
 \mathbb C[\mathbb Z_2^{*2}]\otimes M_n(\mathbb C),\qquad
 B_{0,1}\mapsto B_{0,1}\otimes I,\quad
 B_{2,3}\mapsto I\otimes C_{0,1}.
\end{equation*}
Define, on words through degree $2k$,
\begin{equation}
 \Phi_{a|x}=\begin{cases}
 (\phi_{a|x}\otimes\nu)\circ\pi,&x=0,1,\\
 (L\otimes\nu_{a|x})\circ\pi,&x=2,3.
 \end{cases}
 \label{eq:bidirectional-lift}
\end{equation}
All outcome sums equal $(L\otimes\nu)\circ\pi$, with mass one.
To prove positivity, write the image of any degree-$k$ Bob polynomial
as $\pi(f)=\sum_{w\in W_k}w\otimes C_w$. For a positive matrix
state $\rho=DD^*$, one has
\begin{equation}
 (\phi\otimes\operatorname{Tr}(\rho\,\cdot))
       (\pi(f)^*\pi(f))
 =\sum_{i,j}\phi\left[
   \left(\sum_w(C_wD)_{ij}w\right)^*
   \left(\sum_w(C_wD)_{ij}w\right)\right]\geq0.
 \label{eq:bidirectional-matrix-positivity}
\end{equation}
Every polynomial inside $\phi$ has old Bob degree at most $k$.
This applies to $\phi_{a|x}$ and $L$, and to every conditional state
$\nu_{a|x}$. Quotient and adjoint identities hold under $\pi$.
Thus \eqref{eq:bidirectional-lift} is feasible on all native words.
Expanding each $P_y=(I+B_y)/2$ gives an invertible triangular map on
the four-generator degree-$k$ word basis, with diagonal $2^{-|w|}$.
Congruence transfers every block, including both outcomes, to the raw
PVM convention. No simultaneous full extension of the original four
functionals is assumed.

\paragraph{An explicit witness surviving the coupling.}
Use the complete Fejer-weighted assemblage
\eqref{eq:path-four}, with $r=r_k$. Its first-test Alice bias is
$h=1-2r$ and its first-test value at $\alpha=2-\epsilon$ is
$w=4-(1-2r)\epsilon$.
For the second test take the quantum assemblage of
\eqref{eq:bidirectional-strategy} with $B_2=Z$, $B_3=X$, and
$A_{2,3}=bZ\pm dX$.
It attains $q$ and has $\nu(B_2)=m=2\alpha/q$.
The lifted witness has the exact score
\begin{equation}
 W_k=w+q+\eta hm,\qquad
 \Delta_k:=W_k-\beta_\alpha
 =w-q-\eta(1-hm).
 \label{eq:bidirectional-witness-score}
\end{equation}

At $\epsilon=8r$, $\eta=r$, set
\begin{align*}
 A&=w-r=4-9r+16r^2,\\
 C&=q^2-2r\alpha h=16-68r+152r^2-32r^3,\\
 p(r)&=496r^4-680r^3+341r^2-86r+8.
\end{align*}
Direct coefficient comparison gives
\begin{equation}
 q\Delta_k=qA-C,\qquad q^2A^2-C^2=64r^2p(r).
 \label{eq:bidirectional-gap-polynomial}
\end{equation}
The signs needed for taking square roots must be retained:
$A>0$ and $C>0$ for $0<r\leq1/10$.
With $u=1/10-r\geq0$, the positive polynomial identity is
\begin{equation}
 p(r)=496u^4+\frac{2408}{5}u^3+\frac{4169}{25}u^2
             +\frac{4527}{125}u+\frac{5449}{2500}>0.
 \label{eq:bidirectional-gap-positive}
\end{equation}
Consequently $qA>C$ and the lifted witness separates at every $k$.
Moreover $q<4$, $A\leq4$, and $C\leq16$, so
\begin{equation}
 \Delta_k=\frac{64r^2p(r)}{q(qA+C)}
 >\frac{5449}{5000}r^2.
 \label{eq:bidirectional-exact-gap}
\end{equation}
Weak duality excludes $\mathsf O_k^A$.
The involution exchanging the parties and relabeling settings
$0\leftrightarrow2$, $1\leftrightarrow3$ leaves $G_\alpha$ invariant,
preserves native degree, and exchanges $\mathsf O_k^A$ with
$\mathsf O_k^B$. It excludes the other direction as well.
This proves \eqref{eq:bidirectional-cones}.

\paragraph{Bounds for every sufficiently small epsilon.}
The witness in \eqref{eq:bidirectional-witness-score} is available at
every $\alpha$, not just the displayed sequence.
For $\epsilon k^2\to C_0>0$, Taylor expansion with
$r_k=(8k^2+2)^{-1}$ gives
\begin{equation}
 k^4\Delta_k\longrightarrow \frac{C_0(7-5C_0)}{32}.
 \label{eq:bidirectional-gap-asymptotic}
\end{equation}
For example, this follows from
$q=4-\epsilon+\epsilon^2/8+O(\epsilon^3)$,
$m=1-\epsilon/4+O(\epsilon^2)$, and
$\Delta_k=(7/4)r_k\epsilon-(5/32)\epsilon^2+
O(r_k\epsilon^2+\epsilon^3)$.
The remainders are uniform when $\epsilon k^2$ ranges over a compact
positive interval. Choose $k=\lfloor c_0/\sqrt\epsilon\rfloor$ for
any $0<c_0<\sqrt{7/5}$. The limit is positive, so both conditioned
degrees exceed $k$ for all sufficiently small $\epsilon$.
Letting $c_0$ approach $\sqrt{7/5}$ proves
\eqref{eq:bidirectional-constant}.

For the upper bound, the first test has the existing finite
Alice-conditioned certificate of degree $d_A(F_\alpha,q)$.
The second has the opposite-tilt degree-one certificate of
\cref{prop:opposite-tilt} after relabeling. Finally,
\begin{equation}
 I-A_0B_2=P_{+|0}(I-B_2)+P_{-|0}(I+B_2),\qquad
 I\pm B_2=\tfrac12(I\pm B_2)^2,
 \label{eq:bidirectional-bridge-nice}
\end{equation}
is Alice-conditioned degree one. Summing these actual identities yields
\begin{equation}
 d_A(G_\alpha,\beta_\alpha)
 \leq d_A(F_\alpha,q).
 \label{eq:bidirectional-upper}
\end{equation}
The degree-preserving symmetry gives the same bound for Bob.
Finite exactness follows from \cref{thm:fixed-tilt-sm}.
The existing matching upper bound in \cref{thm:matching-upper} is
$12/\sqrt\epsilon+3$ on its stated range (with the endpoint refinement
of prefactor 9). Together with the lower bound this proves
\eqref{eq:bidirectional-sharp}.

\paragraph{Why the attaining strategy is coupled.}
In any strategy attaining the individual tilted optimum at an interior
$\alpha$, the tilted marginal is $q'(\alpha)=2\alpha/q=m$.
Indeed the linear payoff of that strategy is a supporting line to the
differentiable quantum support function $q$. In independent optimal
implementations of the two tests, the bridge expectation would be
$\langle A_0\rangle\langle B_2\rangle=m^2<1$.
This also holds for classical mixtures of such independent optimal
implementations. But equality at $\beta_\alpha$ forces both subtests
to attain $q$ and the bridge to attain one. The shared-qubit strategy
does so, while independent optimal subtests do not.
\end{proof}

\paragraph{The coupling remains relevant on the scale of nonlocality.}
There is a stronger quantitative distinction than the exclusion of
individually optimal independent subtests. Define independent test
resources to mean a tensor product of two bipartite quantum strategies,
with each subtest's measurements acting on its own factor, and allow
classical mixtures of these product strategies. Each subtest may
sacrifice its own maximal value to improve the bridge correlator.
For a fixed tilted marginal $v\in[0,1]$, the maximal CHSH value is
$2\sqrt{2-v^2}$: minimize $q(t)-tv$ over $0\leq t\leq2$ and use
the attaining tilted strategy. Therefore the best independent value is
\begin{equation}
 \beta_{\rm ind}=\max_{u,v\in[0,1]}
 \left\{\alpha(u+v)+2\sqrt{2-u^2}+2\sqrt{2-v^2}+\eta uv\right\}.
 \label{eq:bidirectional-independent}
\end{equation}
Restricting to nonnegative biases loses nothing, since simultaneous
outcome reversal in either subtest preserves its CHSH value.
Classical mixtures cannot exceed this maximum.
The diagonal Hessian entries are $-4/(2-v^2)^{3/2}\leq-\sqrt2$;
the off-diagonal entry is $\eta<1/4$. Thus the objective is strictly
concave. Symmetry and the boundary derivatives imply a unique interior
optimizer $u=v=m_*$, satisfying
\begin{equation}
 \frac{2m_*}{\sqrt{2-m_*^2}}=\alpha+\eta m_*.
 \label{eq:bidirectional-independent-bias}
\end{equation}
With $\alpha=2-\epsilon$ and $\eta=\epsilon/8$, its derivative at the
critical endpoint is nondegenerate and
$m_*=1-7\epsilon/32+O(\epsilon^2)$. Hence
\begin{align}
 \beta_\alpha-\beta_{\rm ind}
 &=\frac{15}{256}\epsilon^2+O(\epsilon^3),\notag\\
 \beta_\alpha-\beta_{\rm loc}
 &=\frac14\epsilon^2+O(\epsilon^3),\qquad
 \beta_{\rm loc}=4+2\alpha+\eta,\notag\\
 \frac{\beta_\alpha-\beta_{\rm ind}}
      {\beta_\alpha-\beta_{\rm loc}}&\longrightarrow\frac{15}{64}.
 \label{eq:bidirectional-physical-scale}
\end{align}
The bridge coefficient tends to zero, but its effect on the best
independent implementations is a nonzero fraction of the nonlocal
violation on the critical scale. This gives a quantitative meaning to
the coupling without claiming that its coefficient stays bounded below.

\paragraph{Scope and what remains open.}
The result resolves the existence version of a fixed finite Bell
scenario with bounded standard degree and diverging minimum conditioned
degree, allowing a positive parameter-dependent bridge. It rules out
repair by a global exchange of the conditioning party within these two
native nice-SOS classes. It does not prove a cryptographic attack,
failure of compiled-game soundness, or a lower bound for every possible
proof method or compilation protocol.
It does not establish the analogous statement in the original
two-question-per-party scenario, or with bridge strength bounded away
from zero as $\alpha\to2$.
Two additional limitations can be proved exactly, rather than inferred
from unsuccessful numerical searches.

\begin{proposition}[Low-degree mixed-direction decomposition]
\label{prop:bidirectional-mixed}
For every $0<\alpha<2$ and every $\eta\geq0$, put
$G_{\alpha,\eta}=F^{(1)}_\alpha+F^{(2)}_\alpha+\eta A_0B_2$.
Then its exact gap satisfies
\begin{equation}
 (2q+\eta)I-G_{\alpha,\eta}
 \in\mathsf O^A_1+\mathsf O^B_1.
 \label{eq:bidirectional-mixed}
\end{equation}
\end{proposition}
\begin{proof}
The first tilted gap belongs to $\mathsf O^B_1$ and the second to
$\mathsf O^A_1$ by the opposite-tilt certificate, with settings relabeled.
The bridge gap belongs to both cones by
\eqref{eq:bidirectional-bridge-nice} and its party-exchanged identity.
Adding these three identities proves the claim. Their quantum target
is simultaneously attained by \eqref{eq:bidirectional-observables}
for every $\eta\geq0$.
\end{proof}

The sum of cones allows different summands to use different conditioning
parties; it is larger than their union. Thus
\eqref{eq:bidirectional-cones} and \eqref{eq:bidirectional-mixed} are
consistent. The result excludes a global choice of an easy party but
does not exclude a mixed-direction algebraic proof. Such a decomposition
is not, by itself, a soundness theorem for a compilation that mixes
directions: any operational application still requires its own argument.

\begin{proposition}[Fixed-bridge obstruction for the product-lift witness]
\label{prop:bidirectional-product-limit}
Keep the first-test Fejer assemblage with
$r=r_k$, $h=1-2r$, and $w=\alpha+2+2r\epsilon$.
Optimize the second factor over all genuine quantum strategies, instead
of fixing it to the strategy at tilt $\alpha$. Its best product-lift
score is
\begin{equation}
 W_{\rm prod}=w+Q(\alpha+\eta h),\qquad
 Q(t)=\begin{cases}\sqrt{8+2t^2},&0\leq t\leq2,\\
                         t+2,&t\geq2.
       \end{cases}
 \label{eq:bidirectional-product-opt}
\end{equation}
For every $k\geq1$, this score is strictly below the exact quantum
target whenever $0<\epsilon\leq\eta$.
\end{proposition}
\begin{proof}
The bridge changes only the second test's tilted marginal coefficient,
from $\alpha$ to $\alpha+\eta h$. This gives the displayed optimum;
classical mixtures of these products cannot increase it.
Since the tilted observable has norm one, the quantum support function
$Q$ is 1-Lipschitz. Hence
\begin{align}
 W_{\rm prod}-(2q+\eta)
 &\leq w-q-2\eta r\notag\\
 &=2r(\epsilon-\eta)-\bigl[q-(\alpha+2)\bigr]<0,
 \label{eq:bidirectional-product-no-go}
\end{align}
using $q^2-(\alpha+2)^2=(2-\alpha)^2>0$.
This proves a limitation of this fixed first-assemblage product-lift
construction, not finite exactness of the fixed-coupling hierarchy and
not impossibility of correlated truncated witnesses.
\end{proof}

\paragraph{Verification.}
The exact checker
\codepath{research/bidirectional_coupling/verify_bidirectional_coupling.py}
checks the full noncommutative SOS identity, the parameter-dependent
attaining values, symmetry, the signed gap polynomial, its explicit
positive expansion, and the asymptotic constant exactly.
At levels 1, 2, and 3 it constructs all 5, 17, and 53 four-generator
native Bob words and all eight outcome blocks, and verifies an explicit
PSD congruence and the invertible full PVM basis change.
These finite checks audit the universal lifting proof, rather than
replacing it with sampling. An independent check at $\alpha=6/5$
uses the actual separating assemblage, evaluates all 248 principal
minors of its eight full level-one Grams in exact algebraic arithmetic,
and reconstructs the separating score from their entries.
The certificate package contains these data and checkers.

\section{A tunable bridge and the scale of the coupling}
\label{sec:bridge-strength-research}

The ratio $\eta/\epsilon=1/8$ is convenient for an explicit positive
polynomial at every integer level, but is not intrinsic to the
construction. Optimizing the second factor in the lifted witness gives
a stronger statement for every fixed bridge ratio below one.

\begin{theorem}[Tunable coupling with double degree divergence]
\label{thm:bridge-ratio}
Fix $0<c<1$, set $\alpha=2-\epsilon$, $\eta=c\epsilon$, and keep
$G_{\alpha,\eta}=F^{(1)}_\alpha+F^{(2)}_\alpha+\eta A_0B_2$.
Its exact quantum value is $\beta_{\alpha,\eta}=2q(\alpha)+\eta$;
its standard exact degree is two. Both conditioned exact degrees are
finite for every $0<\epsilon<2$ and obey
\begin{equation}
 d_A=d_B=\Theta_c(\epsilon^{-1/2}),\qquad
 \liminf_{\epsilon\downarrow0}\sqrt\epsilon\min(d_A,d_B)
 \geq\sqrt{\frac{2(1-c)}{1+2c-c^2}}.
 \label{eq:bridge-ratio-cost}
\end{equation}
For the independent test-resource class defined in
\eqref{eq:bidirectional-independent}, the relative loss obeys
\begin{equation}
 \frac{\beta_{\alpha,\eta}-\beta_{\rm ind}}
      {\beta_{\alpha,\eta}-\beta_{\rm loc}}
 \longrightarrow c(2-c).
 \label{eq:bridge-ratio-fraction}
\end{equation}
In particular, with $c=9/10$, this fraction tends to $99/100$ and
the degree lower prefactor is at least $\sqrt{20/199}>0$.
\end{theorem}
\begin{proof}
The common attaining strategy and the standard SOS in
\eqref{eq:bidirectional-standard-sos} apply for every $\eta\geq0$.
The standard level-one witness also persists. The existing upper
conditioned certificate and its party-exchanged version yield the
same $O(\epsilon^{-1/2})$ bound independently of the bridge ratio.

For the lower bound, keep the first-test Fejer assemblage with
$r=r_k=(8k^2+2)^{-1}$ and $h=1-2r$. In the tensor lift choose the
second test's genuine quantum strategy at its effective tilt
\begin{equation}
 t=\alpha+\eta h=2-(1-c)\epsilon-2cr\epsilon\in(0,2).
\end{equation}
The lifted score is $4-\epsilon+2r\epsilon+q(t)$. Subtracting
$\beta_{\alpha,\eta}$ and expanding the analytic square roots at
$\epsilon=0$ gives
\begin{equation}
 \Delta=2(1-c)r\epsilon
       -\frac{1+2c-c^2}{8}\epsilon^2
       +O_c(r\epsilon^2+\epsilon^3).
 \label{eq:bridge-ratio-gap}
\end{equation}
The remainder is uniform for $r$ in a sufficiently small fixed
nonnegative interval. If $\epsilon k^2\to C>0$, then
\begin{equation}
 k^4\Delta\longrightarrow
 \frac{C\{2(1-c)-(1+2c-c^2)C\}}8.
 \label{eq:bridge-ratio-limit}
\end{equation}
For any $v$ strictly below the right-hand prefactor in
\eqref{eq:bridge-ratio-cost}, take
$k=\lfloor v/\sqrt\epsilon\rfloor$. The limit is strictly positive.
Thus both conditioned degrees exceed $k$ for all sufficiently small
$\epsilon$. Party exchange with the same setting permutation proves
equality of the two degrees, completing the degree claim.

For the independent-resource comparison, the optimization is
\eqref{eq:bidirectional-independent} with $\eta=c\epsilon$.
For fixed $c$ and sufficiently small $\epsilon$, its Hessian is
negative definite since the diagonal entries are at most $-\sqrt2$
and the off-diagonal entries equal $c\epsilon$.
The unique maximizer is symmetric and interior: the scalar derivative
at bias zero is positive, whereas at bias one it is
$\alpha+\eta-2=-(1-c)\epsilon<0$.
Its common bias therefore satisfies
\begin{equation}
 \frac{2m_*}{\sqrt{2-m_*^2}}=2-\epsilon+c\epsilon m_*.
\end{equation}
At $(\epsilon,m_*)=(0,1)$ the derivative of the left side minus
the right side with respect to $m_*$ is four. The implicit function
theorem gives $m_*=1-(1-c)\epsilon/4+O_c(\epsilon^2)$.
Substitution into the objective gives
\begin{align*}
 \beta_{\alpha,\eta}-\beta_{\rm ind}
  &=\frac{c(2-c)}4\epsilon^2+O_c(\epsilon^3),\\
 \beta_{\alpha,\eta}-\beta_{\rm loc}
  &=\frac14\epsilon^2+O(\epsilon^3).
\end{align*}
This proves the stated ratio.
\end{proof}

\paragraph{Interpretation and limits.}
For any prescribed $\delta\in(0,1)$, choose a fixed
$c\in(1-\sqrt\delta,1)$. The independent-resource loss then accounts
asymptotically for more than $1-\delta$ of the full nonlocal advantage,
while both fixed-direction conditioned degrees still diverge.
This is a precise relative-scale statement. The coefficient
$\eta=c\epsilon$ still vanishes; the lower prefactor also tends to zero
as $c\uparrow1$, so the argument does not permit interchanging these
two limits or fixing $\eta>0$.
Moreover the exact gap is in $\mathsf O^A_1+\mathsf O^B_1$ for every
$c$. Thus a strong coupling effect in this independent-resource
comparison does not remove the low-degree mixed-direction decomposition.
The product-lift obstruction at $\epsilon\leq\eta$ only identifies
the boundary of this particular witness construction.

\paragraph{Exact audit.}
\codepath{research/bidirectional_coupling/verify_scope_extensions.py}
checks the noncommutative mixed-direction identities, all-parameter
Jordan-block determinant formulas, the product-lift comparison, both
asymptotic coefficients, and their rational specializations.
The displayed universal analytic arguments, rather than sampling,
justify the parameter intervals and limit statements.

\section{Resolving the quantum bound on the Bell-advantage scale}
\label{sec:relative-precision}

The following corollary records a quantitative consequence of the existing
construction in \cref{thm:bridge-ratio} for approximate upper bounds.
It is derived from the same witness and upper construction, not from a
separate separation mechanism. It makes explicit what a finite conditioned level can
resolve when the physical advantage tends to zero. No change to the
underlying Bell family or its witness is required.

Fix $0<c<1$, put $\alpha=2-\epsilon$, and write
$G=G_{\alpha,c\epsilon}$ and $\beta=2q(\alpha)+c\epsilon$.
Each tilted term has local value $\alpha+2$, and a single deterministic
assignment attains both local bounds and unit bridge correlation. Hence
\begin{equation}
 \beta_{\rm loc}=2(\alpha+2)+c\epsilon,\qquad
 V(\epsilon):=\beta-\beta_{\rm loc}
 =2\{q(\alpha)-\alpha-2\}
 =\frac{2\epsilon^2}{q(\alpha)+\alpha+2}
 =\frac{\epsilon^2}{4}+O(\epsilon^3)>0.
 \label{eq:relative-advantage}
\end{equation}
Let $\omega_k^A(G)$ and $\omega_k^B(G)$ be the full level-$k$
conditioned relaxation values. Define
\begin{equation}
 E_k(\epsilon,c)=\min\{\omega_k^A(G),\omega_k^B(G)\}-\beta,\qquad
 K_\rho(\epsilon,c)=\min\{k\geq1:E_k(\epsilon,c)\leq\rho V(\epsilon)\}.
 \label{eq:relative-cost-definition}
\end{equation}
Here $\rho\geq0$ is a fixed tolerance ratio. The two relaxation values
are equal by party exchange with the setting permutation used in
\cref{thm:bidirectional-coupling}. Taking their minimum explicitly
allows the more favourable global direction at every instance.

\begin{corollary}[Sharp cost at fixed relative accuracy]
\label{cor:relative-precision}
For every fixed $0<c<1$ and finite $\rho\geq0$,
\begin{equation}
 K_\rho(\epsilon,c)=\Theta_{c,\rho}(\epsilon^{-1/2})
                   =\Theta_{c,\rho}(V(\epsilon)^{-1/4}),
 \qquad \epsilon\downarrow0,
 \label{eq:relative-cost-law}
\end{equation}
with the explicit bound
\begin{equation}
 \liminf_{\epsilon\downarrow0}\sqrt\epsilon K_\rho(\epsilon,c)
 \geq\sqrt{\frac{2(1-c)}{1+2c-c^2+2\rho}}.
 \label{eq:relative-cost-prefactor}
\end{equation}
At any fixed level $k$, $E_k(\epsilon,c)/V(\epsilon)\to+\infty$.
In contrast, standard level two and mixed level one have zero
relaxation error throughout the family.
\end{corollary}
\begin{proof}
Use the feasible lifted witness of \cref{thm:bridge-ratio}. Set
$r=(8k^2+2)^{-1}$, $h=1-2r$, $t=\alpha+c\epsilon h$ and
$s=2-t=\epsilon(1-ch)$. Its score excess has the exact, numerically
stable expression
\begin{equation}
 D_k=4-h\epsilon+q(t)-\beta
 =2(1-c)r\epsilon+
 \frac{s^2}{q(t)+t+2}-
 \frac{2\epsilon^2}{q(\alpha)+\alpha+2}.
 \label{eq:relative-stable-witness}
\end{equation}
Indeed $q(x)^2-(x+2)^2=(2-x)^2$, so rationalizing each quantum
gap gives the equality. Feasibility in each orientation proves
$E_k\geq\max\{0,D_k\}$. For fixed $k$, equation
\eqref{eq:relative-stable-witness} gives
$D_k=2(1-c)r\epsilon+O_{c,k}(\epsilon^2)$;
division by \eqref{eq:relative-advantage} proves the fixed-level claim.

For $\epsilon k^2\to C>0$, divide
\eqref{eq:bridge-ratio-limit} by
$k^4V\to C^2/4$ to obtain
\begin{equation}
 \frac{D_k}{V}\longrightarrow
 \frac{1-c}{C}-\frac{1+2c-c^2}{2}.
 \label{eq:relative-witness-limit}
\end{equation}
Choose any positive $v$ below the right side of
\eqref{eq:relative-cost-prefactor}, and take
$k=\lfloor v/\sqrt\epsilon\rfloor$. The limiting ratio exceeds
$\rho$, so $E_k>\rho V$ for sufficiently small $\epsilon$.
Nestedness of the moment relaxations implies $K_\rho>k$.
Letting $v$ approach the prefactor proves the lower bound.
The exact conditioned upper certificate gives
$K_\rho\leq d_A=d_B=O_c(\epsilon^{-1/2})$;
in particular, $K_\rho\leq9/\sqrt\epsilon+3$ in the endpoint range
of \cref{thm:matching-upper}. The relation to $V$ follows from
\eqref{eq:relative-advantage}. The standard and mixed upper
identities were established in \cref{thm:bridge-ratio,prop:bidirectional-mixed}.
\end{proof}

\paragraph{What is, and is not, normalized.}
The ratio $E_k/V$ is unchanged by a positive rescaling of the Bell
functional or addition of a constant. Thus the conclusion is independent
of score normalization. It fixes the desired error as a fraction of the
available Bell advantage, not as a fixed absolute error. It does not
establish finite-noise robustness, a nonvanishing bridge, or a lower bound
on the number of experimental trials. In particular, certifying a Bell
violation against the local bound is a different task from upper-bounding
the quantum maximum to this relative accuracy.

\paragraph{Consequences for certificate transfer.}
There is no uniform finite level, determined only by the standard
degree, the number of questions and the dimension of an attaining
strategy, that converts these standard certificates into a certificate
in either conditioned direction. The same obstruction holds when the
target is relaxed to $\beta+\rho V$ at any fixed finite $\rho$.
An argument that uses a level-$k$ conditioned Bell upper bound therefore
retains at least $E_k$ of algebraic slack, before any separate protocol
error is estimated. This statement concerns the upper bound available
from the specified certificate cone; it is not an achievable score for
a compiled prover and does not compare all soundness methods.

\paragraph{Moment dimensions.}
For four binary Bob observables, the complete independent reflection-word
basis through degree $k$ contains
$1+4\sum_{j=0}^{k-1}3^j=2\cdot3^k-1$ words. The full Alice-conditioned
formulation has eight positive semidefinite blocks on this basis, with
linear relations between them. The standard degree-two word basis for
four questions on each side has $1+8+12+12+16=49$ elements.
These are sizes of the complete word-basis formulations. Symmetry,
sparsity or analytic certificates may reduce computational work;
the word counts do not prove a lower bound on every implementation's
runtime or memory.

\paragraph{Verification.}
\codepath{scripts/verify_relative_precision.py} checks the rationalized
identity, the asymptotic coefficient and its tolerance threshold
symbolically. It also verifies the word counts and the local bound on
all deterministic assignments. The analytic argument above supplies
the uniform limit and the bound on the minimal level.

\section{An analytic obstruction at every finite conditioned level}
\label{sec:unbounded}

All statements in this section use the reduced raw-PVM hierarchy of
\cref{sec:prelim}. In particular, positivity is tested on all Bob words
of length at most the stated native level, with both Alice outcomes.
The construction is analytic and requires no optimization or
rationalization of numerical output.

\subsection{Fej\'er-weighted positive functionals on the infinite dihedral group}

Let $\mathcal A=\mathbb R[\mathbb Z_2*\mathbb Z_2]$ be generated by
$B_0,B_1$, with $B_y^2=I$ and $B_y^*=B_y$. Its canonical trace $\tau$
extracts the identity coefficient and satisfies
$\tau(f^*f)=\sum_w c_w^2$ for $f=\sum_w c_ww$.
Let $\chi(B_0)=\chi(B_1)=-1$ and put $U=B_0B_1$. For each $k\geq1$,
define
\begin{equation}
 S_k=\sum_{j=0}^{k-1}U^j,\qquad H_k=S_k^*S_k/k,\qquad
 \tau_k(f)=\tau(H_kf).
 \label{eq:fejer-trace}
\end{equation}
The expansion
$H_k=\sum_{|j|<k}(1-|j|/k)U^j$ is invariant under $U\mapsto U^{-1}$.
Since each $B_y$ conjugates $U$ to $U^{-1}$, $H_k$ is central.
It is positive, $\tau(H_k)=1$, and
$\tau_k(f^*f)=\tau((fS_k^*)^*fS_k^*)/k\geq0$.
Centrality also makes $\tau_k$ tracial. Its reduced-word values are
\begin{equation}
 \tau_k(w)=\begin{cases}
 (1-|w|/(2k))_+,& |w|\text{ even},\\
 0,& |w|\text{ odd},
 \end{cases}
 \label{eq:fejer-moments}
\end{equation}
where $(t)_+=\max(t,0)$; every even word is $U^j$ or $U^{-j}$.
Set
\begin{equation}
 L_y(f)=\tfrac12\tau_k((I+B_y)f),\qquad L=L_0+L_1.
 \label{eq:path-L}
\end{equation}
Writing $P_y=(I+B_y)/2$, the trace property gives
$L_y(f^*f)=\tau_k((fP_y)^*fP_y)\geq0$.
Thus $L_y$ is positive on the full algebra, $L_y(I)=1/2$, and $L(I)=1$.
The functionals depend on $k$; we suppress that index on $L_y,L$.

Let $W_k$ contain the identity and the two alternating words of each
length $1,\ldots,k$, with $|W_k|=2k+1$. Set
\begin{equation}
 N=2k,\qquad r_k=\frac1{8k^2+2},\qquad v_w=(-1)^{|w|}.
 \label{eq:path-data}
\end{equation}
The right Cayley graph on $W_k$ is a path. In path order, adjacent words
differ by one generator and $|u_i^{-1}u_j|=|i-j|$.
For an even reduced word of length $d\leq2k$, $L(w)=1-d/(2k)$.
For odd $d$, multiplication by $B_0$ and $B_1$ gives even lengths
$d-1$ and $d+1$; averaging \cref{eq:fejer-moments} gives the same formula.
Consequently the Gram matrices
\begin{equation}
 (M_y)_{u,v}=L_y(u^{-1}v),\qquad
 M=M_0+M_1,\qquad M_{ij}=1-\frac{|i-j|}{N}\quad(0\leq i,j\leq N)
 \label{eq:path-grams}
\end{equation}
are specified on every entry, not only on the Bell projection.
Restrict the four functionals
\begin{equation}
 \phi_{0|0}=L-r_k\chi,\quad \phi_{1|0}=r_k\chi,\quad
 \phi_{0|1}=L_0,\quad\phi_{1|1}=L_1
 \label{eq:path-four}
\end{equation}
to moments through Bob degree $2k$. Their Grams are
$M-r_kvv^T$, $r_kvv^T$, $M_0$, and $M_1$.
All quotient and adjoint identities hold by construction; the outcome
sums equal $L$ on every word and have identity entry one.
Only the positivity of the first truncated block needs proof.

\begin{lemma}[Exact rank-one subtraction from a moving-average Gram]
\label{lem:path-psd}
For \cref{eq:path-data,eq:path-grams}, $M\succ0$ and
$M-r_kvv^T\succeq0$, with rank $2k$.
\end{lemma}
\begin{proof}
The vectors $a_i=N^{-1/2}\boldsymbol1_{\{i,\ldots,i+N-1\}}$ in
$\mathbb R^{2N}$, $0\leq i\leq N$, satisfy
$\langle a_i,a_j\rangle=1-|i-j|/N$.
Their first nonzero coordinates occur at distinct positions, so they
are linearly independent and $M\succ0$.
The common sign of the character vector is immaterial; choose
$v_i=(-1)^i$. Define
\begin{equation}
 z_i=(-1)^ix_i,\qquad x_0=x_N=N+1,\qquad x_i=2N\ (0<i<N).
 \label{eq:fejer-inverse-vector}
\end{equation}
We verify $Mz=v$ for every even $N=2k$.
For $0<i<N$ the second row difference of $M$ is
$M_{i-1,j}-2M_{ij}+M_{i+1,j}=-2\delta_{ij}/N$.
Thus $Mz$ and $v$ have the same second differences, $-4(-1)^i$;
their difference is affine in $i$. Both vectors are symmetric under
$i\mapsto N-i$, so this affine difference is constant.
At $i=0$,
\begin{equation*}
 (Mz)_0=N+1+2N\sum_{j=1}^{N-1}(-1)^j(1-j/N)=1.
\end{equation*}
Indeed pairing $j=2\ell-1,2\ell$ and adding the final odd term gives
the sum $-k/N=-1/2$. Therefore $Mz=v$ at every index and
\begin{equation}
 v^TM^{-1}v=v^Tz=2(N+1)+2N(N-1)=8k^2+2=r_k^{-1}.
 \label{eq:path-sum}
\end{equation}
Cauchy--Schwarz gives
$(v^Tf)^2\leq(v^TM^{-1}v)(f^TMf)$, proving the PSD subtraction.
Equality holds precisely on the span of $M^{-1}v$, so the rank is $2k$.
\end{proof}

The substitution $P_y=(I+B_y)/2$ is an invertible triangular change
of basis on the degree-$k$ word span, with diagonal $2^{-|w|}$.
Congruence transfers all four Grams to the original PVM hierarchy,
preserving positivity and all functional identities. Hence
\cref{eq:path-four} is feasible with both Alice outcomes and every
original projector moment. Joint probabilities are also nonnegative,
as already enforced by the degree-one principal blocks.

\subsection{Exact separation and square-root degree growth}

\begin{theorem}[No finite conditioned level contains standard level two]
\label{thm:unbounded-sm}
For each integer $k\geq1$, put $\alpha_k=2-8r_k$. Then
\begin{equation}
 0<\alpha_k<2,\qquad q(\alpha_k)I-F_{\alpha_k}
 \in\mathsf D_2\setminus\mathsf O_k.
 \label{eq:no-uniform-conversion}
\end{equation}
Thus $\mathsf D_2\not\subseteq\mathsf O_k$ for every finite $k$,
including $k=3$.
\end{theorem}
\begin{proof}
Equation~\eqref{eq:fejer-moments} gives
\begin{equation*}
 L(B_y)=1-\frac1{2k},\qquad (L_0-L_1)(B_0-B_1)=\frac1k.
\end{equation*}
Since $\chi(B_y)=-1$, the witness has
\begin{align}
 h_k&=\phi_{0|0}(I)-\phi_{1|0}(I)=1-2r_k,\\
 g_k&=(\phi_{0|0}-\phi_{1|0})(B_0+B_1)
       +(\phi_{0|1}-\phi_{1|1})(B_0-B_1)=2+4r_k.
\end{align}
Its value at $\alpha=2-\epsilon$ is $w_k=4-(1-2r_k)\epsilon$ and
\begin{equation}
 w_k^2-q(2-\epsilon)^2
 =\epsilon[16r_k-(1+4r_k-4r_k^2)\epsilon].
 \label{eq:path-residual}
\end{equation}
Since $0<r_k\leq1/10$, it strictly exceeds $q$ whenever
\begin{equation}
 0<\epsilon<t_k:=\frac{16r_k}{1+4r_k-4r_k^2}
 =\frac{4(4k^2+1)}{8k^4+8k^2+1}.
 \label{eq:fejer-threshold}
\end{equation}
In particular $8r_k<t_k$ and
\begin{equation}
 w_k^2-q(\alpha_k)^2=64r_k^2(1-2r_k)^2>0.
 \label{eq:path-sequence-gap}
\end{equation}
The standard SOS in \cref{lem:standard-tilted} supplies membership
in $\mathsf D_2$; weak duality with the full feasible witness excludes
membership in $\mathsf O_k$.
\end{proof}

The rare Alice outcome carries a deterministic Bob response of weight
$r_k$. Fej\'er weighting changes the common Bob functional as $k$ grows,
allowing $r_k\asymp k^{-2}$. For comparison, the unweighted choice
$H=I$ gives $M=I+T/2$ on the word path and
$r_k^{\rm tr}=6/[(2k+1)(2k+2)(2k+3)]\asymp k^{-3}$.
Conjugating that matrix by the sign diagonal gives a Dirichlet path
Laplacian; its inverse applied to $\boldsymbol1$ has entries
$j(2k+2-j)$, $1\leq j\leq2k+1$. The rank-one criterion is elementary
in both constructions; the degree improvement comes from changing the
positive trace, not from a different estimate of the unweighted matrix.

At $\alpha_k$, $w_k,q(\alpha_k)<4$, so
\begin{equation}
 \omega_k^{\rm os}(F_{\alpha_k})-q(\alpha_k)
 >8r_k^2(1-2r_k)^2.
 \label{eq:path-value-gap}
\end{equation}
For every $0<\epsilon\leq4/5$, define the conservative explicit bound
\begin{equation}
 K(\epsilon)=\left\lfloor\sqrt{\epsilon^{-1}-1/4}\right\rfloor.
\end{equation}
Then $K\geq1$ and $\epsilon\leq4/(4K^2+1)=8r_K<t_K$, giving
\begin{equation}
 d_{\rm os}(F_{2-\epsilon},q(2-\epsilon))>K(\epsilon),
 \qquad d_{\rm std}(F_{2-\epsilon},q(2-\epsilon))=2.
 \label{eq:degree-lower-growth}
\end{equation}
The standard equality uses \cref{thm:L1-std,lem:standard-tilted}.
The full threshold improves the asymptotic constant:
\begin{equation}
 \liminf_{\epsilon\downarrow0}\sqrt\epsilon\,
 d_{\rm os}(F_{2-\epsilon},q(2-\epsilon))\geq\sqrt2.
 \label{eq:fejer-liminf}
\end{equation}
To see this without an asymptotic inversion assumption, fix
$0<c<\sqrt2$ and choose $k=\lfloor c/\sqrt\epsilon\rfloor$.
Then $\epsilon k^2\to c^2<2$, while $k^2t_k\to2$, so eventually
$\epsilon<t_k$ and $d_{\rm os}>k$. Let $c\uparrow\sqrt2$.
This proves an $\Omega(\epsilon^{-1/2})$ lower bound without assuming
finite exactness. Below, \cref{thm:fixed-tilt-sm} proves that the degree
is finite for every fixed subcritical tilt. The quantifiers
$\forall k\,\exists\alpha_k$ do not assert $\exists\alpha\,\forall k$.

For $k=2$, $r_2=1/34$ and $t_2=68/161$. Thus this single feasible
witness separates throughout $\alpha\in(254/161,2)$.
The existing rational fan covers $(\alpha_-,\alpha_+)$ and contains
$8/5>254/161$, so their union is $(\alpha_-,2)$, proving the extended
non-closure interval in \cref{thm:L2-threshold}.

\subsection{Finite exact closure at every fixed subcritical tilt}
\label{sec:fixed-tilt}

We first record how positivity of a single Bob polynomial can be used
without selecting compatible extensions of four primal blocks.
\begin{lemma}[Degree-preserving positivity on the dihedral algebra]
\label{lem:dihedral-saturation}
A Hermitian element of $\mathbb C[\mathbb Z_2*\mathbb Z_2]$ supported
on words of length at most $2k$ is positive in every unitary
representation if and only if it is a sum of squares of polynomials
supported on words of length at most $k$.
\end{lemma}
\begin{proof}
The SOS cone is closed in this finite-dimensional coefficient space.
Indeed the canonical trace of a Gram representation equals the trace
of its PSD Gram matrix, since distinct group words are orthonormal
for that trace. For a convergent sequence of represented polynomials,
these Gram traces stay bounded; a convergent subsequence of Grams
therefore represents the limit.
If a polynomial were outside this cone, separation would give a
linear functional nonnegative on every degree-$k$ square but negative
on that polynomial. Its moment matrix on $W_k$ is PSD.
Every finite set of vertices of diameter at most $2k$ in the Cayley
path lies in a left translate of $W_k$: write $g_i=hu_i$ with
$u_i\in W_k$, so $g_i^{-1}g_j=u_i^{-1}u_j$. Thus the functional defines a
partially positive-definite function on $W_{2k}$ in the precise sense
of \cite{BakonyiTimotin2011}. Proposition~3.4 of that reference extends
it to a positive-definite function on the entire dihedral group.
Its GNS representation contradicts the polynomial's assumed positivity.
The converse is immediate. For a polynomial with real coefficients,
averaging a Gram matrix with its coefficientwise conjugate gives a real
PSD Gram matrix. The substitution $B_y=2P_{B,y}-I$ preserves the word
spans, so the same degree statement holds in the native PVM quotient.
\end{proof}

\begin{theorem}[Finite exact closure uniformly away from the endpoint]
\label{thm:fixed-tilt-sm}
For every $0\leq A<2$ there is a finite native level $k_A$ such that
$q(\alpha)I-F_\alpha\in\mathsf O_{k_A}$ for all $0\leq\alpha\leq A$.
In particular every fixed $0<\alpha<2$ has a finite exact conditioned
certificate, although these degrees diverge as $\alpha\uparrow2$.
\end{theorem}
\begin{proof}
Fix $0\leq\alpha<2$ and put
\begin{equation}
 q=\sqrt{8+2\alpha^2},\quad a=q/4,\quad b=\alpha/2,\quad
 d=q-\alpha-2>0.
 \label{eq:separator-parameters}
\end{equation}
Here $1/\sqrt2\leq a<1$, $b<a$, and $8a^2=\alpha^2+4$.
Let $X=(B_0+B_1)/2$, $Y=(B_0-B_1)/2$; then
$X^2+Y^2=I$ and $XY=-YX$. The standard two-involution spectral
decomposition reduces their representations to blocks
$X=\operatorname{diag}(c,-c)$,
$Y=\bigl(\begin{smallmatrix}0&\sqrt{1-c^2}\\\sqrt{1-c^2}&0\end{smallmatrix}\bigr)$,
$0\leq c\leq1$, including the scalar endpoint representations.
For a real polynomial $P$, positivity of the four polynomials
\begin{equation}
 \begin{gathered}
 P(X)-(\alpha I+2X),\quad P(X)+(\alpha I+2X),\\
 qI-P(X)-2Y,\quad qI-P(X)+2Y
 \end{gathered}
 \label{eq:four-separator-polynomials}
\end{equation}
is therefore equivalent to
\begin{equation}
 \begin{gathered}
 P(c)\geq|\alpha+2c|\quad(-1\leq c\leq1),\\
 Q(c):=q-P(c)\geq0,\qquad
 Q(c)Q(-c)\geq4(1-c^2).
 \end{gathered}
 \label{eq:separator-conditions}
\end{equation}
We construct a polynomial satisfying these inequalities exactly at
the quantum bound $q$.

Start with $p_0(c)=|\alpha+2c|$ and $Q_0=q-p_0$, so $Q_0(c)\geq d$.
Its determinant slack, for $c\geq0$, is
\begin{equation}
 J(c)=Q_0(c)Q_0(-c)-4(1-c^2)=
 \begin{cases}
 (q-\alpha)^2-4,&0\leq c\leq b,\\
 8(c-a)^2,&b\leq c\leq1.
 \end{cases}
 \label{eq:separator-slack}
\end{equation}
Set $D(c)=(c^2-a^2)^2$. Then $0\leq D(c)\leq a^4$ and
\begin{equation}
 J(c)\geq\eta D(c),\qquad
 \eta=\min\left\{\frac{8}{(1+a)^2},
                      \frac{(q-\alpha)^2-4}{a^4}\right\}>0.
 \label{eq:separator-margin}
\end{equation}
For the inner interval use $D\leq a^4$; for the outer interval use
$D=(c-a)^2(c+a)^2$. The same bound holds for negative $c$ by symmetry.

The Hermite polynomial
\begin{equation}
 H(c)=a+\frac{c^2}{a}+\frac{\alpha}{2a}
                 \left(3c-\frac{c^3}{a^2}\right)
 \label{eq:separator-Hermite}
\end{equation}
has $H(\pm a)=2a\pm\alpha$ and $H'(\pm a)=\pm2$.
Consequently the quotient $R=(p_0-H)/D$ has removable singularities
at both contacts. Explicitly its continuous extension is
\begin{equation}
 R(c)=\begin{cases}
 \displaystyle\frac{\alpha(c-2a)-2a^2}{2a^3(c-a)^2},&c\leq-b,\\[6pt]
 \displaystyle\frac{\alpha(c+2a)-2a^2}{2a^3(c+a)^2},&c\geq-b.
 \end{cases}
 \label{eq:separator-remainder}
\end{equation}
The displayed denominators never vanish on their respective branches,
and both values at the join equal $-1/a^3$.

Choose
\begin{equation}
 \tau=\min\left\{\frac{d}{4a^4},\frac{\eta}{8q}\right\}>0
 \label{eq:separator-tau}
\end{equation}
and a real polynomial $R_n$ with $\|R_n-R\|_\infty\leq\tau$.
Such a polynomial exists by uniform approximation of a continuous
function; an explicit choice and degree bound are given below. Set
\begin{equation}
 P=H+D(R_n+2\tau).
 \label{eq:separator-P}
\end{equation}
Then $E:=P-p_0$ satisfies $\tau D\leq E\leq3\tau D$, which proves
the first inequality in \cref{eq:separator-conditions}. Moreover
$Q=Q_0-E\geq d-3\tau a^4\geq d/4>0$, and
\begin{align}
 Q(c)Q(-c)-4(1-c^2)
 &=J(c)-E(c)Q_0(-c)-E(-c)Q_0(c)+E(c)E(-c)\notag\\
 &\geq(\eta-6q\tau)D(c)\geq\tfrac14\eta D(c)\geq0.
 \label{eq:separator-perturbed-slack}
\end{align}
Here $Q_0(c)+Q_0(-c)\leq2q$ and $D(-c)=D(c)$.
All four polynomials in \cref{eq:four-separator-polynomials} are
therefore positive in every representation.

For $k\geq\lceil\deg(P)/2\rceil$, with $k\geq1$,
\cref{lem:dihedral-saturation} makes each of them a Bob degree-$k$
SOS. Multiply the first two by Alice's two $x=0$ projectors and the
last two by her two $x=1$ projectors. Completeness cancels $P(X)$
and leaves exactly $qI-F_\alpha$. Commutation with Bob converts each
term to an allowed conditioned square. Thus this is an actual
$\mathsf O_k$ certificate, not merely membership in a closure.

To make the uniformity and effectiveness explicit, differentiating
\cref{eq:separator-remainder} on its two branches gives the Lipschitz bound
\begin{equation}
 \operatorname{Lip}(R)\leq L_\alpha:=
 \frac{\alpha+4a}{2a^3(a-b)^2}.
 \label{eq:separator-Lipschitz}
\end{equation}
The Bernstein polynomial on $[-1,1]$,
\begin{equation}
 R_n(c)=\sum_{j=0}^n R(-1+2j/n)\binom nj
 \left(\frac{1+c}{2}\right)^j
 \left(\frac{1-c}{2}\right)^{n-j},
 \label{eq:separator-Bernstein}
\end{equation}
satisfies $\|R_n-R\|_\infty\leq L_\alpha/\sqrt n$:
use the binomial expectation formula, Lipschitz continuity, and
$2\sqrt{x(1-x)/n}\leq1/\sqrt n$.
Hence $n\geq\max\{1,\lceil(L_\alpha/\tau)^2\rceil\}$ and
$k=\lceil(n+4)/2\rceil$ suffice. On $0\leq\alpha\leq A<2$,
$L_\alpha$ is bounded and $\tau$ has a positive lower bound, so one
finite $n$ and $k$ work for the whole compact interval.
\end{proof}

This construction resolves finite exactness for fixed tilts but not
its sharp degree. Its deliberately conservative Bernstein estimate
gives only $d_{\rm os}=O((2-\alpha)^{-8})$ near the endpoint:
$L_\alpha=O((2-\alpha)^{-2})$ and $\tau=\Theta((2-\alpha)^2)$.
It does not by itself give the sharp rate. The rational separator and
polynomial pole cancellation in \cref{sec:matching} improve this bound
to the matching square-root order.
The verifier \texttt{python scripts/verify\_fixed\_tilt\_closure.py}
checks the contact identities, the two explicit remainder branches,
the determinant slack and perturbation formulas, and the singular
scales. The inequalities and extension argument above, rather than
finite parameter sampling, establish the theorem.

\paragraph{A sufficient target for the matching rate.}
The uniform constant $\tau$ discards useful pointwise slack. Define
\begin{equation}
 W_\alpha(c)=\frac{J(c)}{4D(c)Q_0(-c)},
 \label{eq:weighted-separator-envelope}
\end{equation}
with the continuous values at $c=\pm a$. If a polynomial $T_n$ obeys
$R(c)\leq T_n(c)\leq R(c)+W_\alpha(c)$ throughout $[-1,1]$, then
$P=H+DT_n$ is already feasible: its excess satisfies
$0\leq E(c)\leq J(c)/(4Q_0(-c))\leq Q_0(c)/4$,
and the determinant slack is at least $J(c)/2$.
Producing such a sandwich with $n=O((2-\alpha)^{-1/2})$ is sufficient
for a matching upper bound. Below we instead approximate a rational
separator with its own explicit positive margins, avoiding direct
approximation of the nonsmooth function $R$.

Define, with the Bell coefficients held fixed,
\begin{equation}
 \Delta_k=\sup_{0\leq\alpha\leq2}
 \{\omega_k^{\rm os}(F_\alpha)-q(\alpha)\},\qquad
 K_{\rm tol}(\delta)=\min\{k\geq1:\Delta_k\leq\delta\}.
 \label{eq:precision-def-sm}
\end{equation}
The subscript distinguishes tolerance from the tilt-dependent integer
$K(\epsilon)$ in \cref{eq:degree-lower-growth}.

\begin{proposition}[From an exact closure threshold to a precision rate]
\label{prop:closure-to-precision}
If level $k$ is exact for every $0\leq\alpha\leq2-\epsilon_k$, then
\begin{equation}
 \Delta_k\leq\frac{\epsilon_k^2}{8\sqrt2}.
 \label{eq:closure-to-precision}
\end{equation}
Consequently a uniform exact-degree bound
$d_{\rm os}(F_{2-\epsilon},q(2-\epsilon))\leq C\epsilon^{-1/2}+C_0$
for sufficiently small $\epsilon>0$, together with
\cref{thm:fixed-tilt-sm}, implies $\Delta_k=O(k^{-4})$.
\end{proposition}
\begin{proof}
The level-$k$ value is convex in $\alpha$ and equals $q$ at both
$2-\epsilon_k$ and $2$; the latter endpoint is exact already at
conditioned level one. On the intervening interval it lies below the
chord of $q$. Since $q''(\alpha)=16/q(\alpha)^3\leq1/\sqrt2$,
the interpolation error is at most
$(\alpha-2+\epsilon_k)(2-\alpha)/(2\sqrt2)
\leq\epsilon_k^2/(8\sqrt2)$. Outside that interval the error is zero.
Under the additional degree hypothesis choose
$\epsilon_k=C^2/(k-C_0)^2$ for sufficiently large $k$; compact
subcritical exactness covers the remaining parameters.
\end{proof}

\subsection{Sharp degree and uniform precision rates}
\label{sec:matching}

\begin{theorem}[Matching endpoint degree and precision bounds]
\label{thm:matching-upper}
For $6/\sqrt{23}\leq\alpha<2$, put $\epsilon=2-\alpha$. Then
\begin{equation}
 d_{\rm os}(F_\alpha,q(\alpha))\leq12\epsilon^{-1/2}+3.
 \label{eq:matching-degree-upper}
\end{equation}
For $30/\sqrt{287}\leq\alpha<2$, the prefactor improves to 9:
$d_{\rm os}\leq9\epsilon^{-1/2}+3$.
Consequently, as $\epsilon,\delta\downarrow0$ and $k\to\infty$,
\begin{equation}
 \begin{gathered}
 d_{\rm os}(F_{2-\epsilon},q(2-\epsilon))=\Theta(\epsilon^{-1/2}),\\
 \Delta_k=\Theta(k^{-4}),\qquad
 K_{\rm tol}(\delta)=\Theta(\delta^{-1/4}).
 \end{gathered}
 \label{eq:matching-rates}
\end{equation}
The constants are not claimed sharp. An explicit eventual error bound is
$\Delta_k\leq9^4/[8\sqrt2(k-3)^4]$ for all sufficiently large $k$.
The sharper limiting bounds are
\begin{equation}
 \sqrt2\leq\liminf_{\epsilon\downarrow0}\sqrt\epsilon\,d_{\rm os}
 \leq\limsup_{\epsilon\downarrow0}\sqrt\epsilon\,d_{\rm os}\leq9,
 \qquad
 \frac18\leq\liminf_{k\to\infty}k^4\Delta_k
 \leq\limsup_{k\to\infty}k^4\Delta_k\leq\frac{6561}{32}.
 \label{eq:matching-leading-bounds}
\end{equation}
\end{theorem}
\begin{proof}
\emph{A rational separator.}
Retain $a=q/4$ and set
\begin{equation}
 \begin{gathered}
 \sigma=1-\frac{\alpha}{2a},\qquad
 a=(1+2\sigma-\sigma^2)^{-1/2},\\
 \gamma=\frac{a(1+\sigma)}{1-\sigma},\qquad t=\gamma-1.
 \end{gathered}
 \label{eq:matching-parameters}
\end{equation}
The specified tilt range is exactly $0<\sigma\leq1/4$.
The identity
$(1+2\sigma-\sigma^2)(1-\sigma)^2
=1-\sigma^2(2-\sigma)^2$
gives $a\geq1-\sigma$. Also
\begin{equation}
 (1+2\sigma-\sigma^2)(1-\sigma/2)^2-1
 =\frac{\sigma}{4}(4-11\sigma+6\sigma^2-\sigma^3)>0,
\end{equation}
so $a\leq1-\sigma/2$. The bracket is positive since it is at least
$4-11/4-1/64>0$. These bounds give
\begin{equation}
 \sigma\leq t\leq\frac{\sigma(3-\sigma)}{2(1-\sigma)}
 \leq2\sigma\leq\frac12,\qquad
 \gamma<2,\qquad\epsilon\leq4\sigma.
 \label{eq:pole-distance}
\end{equation}
The last inequality uses
$\epsilon=2[1-a(1-\sigma)]\leq2[1-(1-\sigma)^2]$.
Write $z=c/a$, $v(z)=(1-\sigma)z+1+\sigma$, and define
\begin{equation}
 P_{\rm r}(c)=2a(1-\sigma+z)
       +\frac{4a\sigma^3(z-1)^2}{v(z)^2}.
 \label{eq:rational-separator}
\end{equation}
Its only pole is $c=-\gamma<-1$. Direct factorization gives
\begin{align}
 P_{\rm r}(c)-(\alpha+2c)
 &=\frac{4a\sigma^3(z-1)^2}{v(z)^2}\geq0,\notag\\
 P_{\rm r}(c)+(\alpha+2c)
 &=\frac{4a(z+1)^2[(1-\sigma)^2z+1+\sigma-\sigma^2]}{v(z)^2}\geq0.
 \label{eq:rational-branch-positivity}
\end{align}
Indeed $z\geq-1/a\geq-(1+\sigma)$, so the last linear factor
is at least $\sigma(2-\sigma^2)>0$.
Put $Q_{\rm r}=q-P_{\rm r}$ and $D=(c^2-a^2)^2$. Its determinant slack is
\begin{equation}
 J_{\rm r}(c):=Q_{\rm r}(c)Q_{\rm r}(-c)-4(1-c^2)
 =\frac{8a^2\sigma^2(1+2\sigma-2\sigma^3+\sigma^4)D(c)}
 {(1-\sigma)^4(c+\gamma)^2(\gamma-c)^2}.
 \label{eq:rational-determinant}
\end{equation}
This is nonnegative, and it is strictly positive at $c=\pm1$.
Thus $Q_{\rm r}(c)Q_{\rm r}(-c)>0$ throughout $[-1,1]$.
Also $Q_{\rm r}(0)=2a(1+3\sigma+3\sigma^2-\sigma^3)/(1+\sigma)^2>0$;
continuity gives $Q_{\rm r}>0$ everywhere. Hence $P_{\rm r}$ satisfies
all three separator conditions in \cref{eq:separator-conditions}.

\emph{A sufficient approximation error.}
For the Hermite polynomial $H$ in \cref{eq:separator-Hermite},
another exact identity is
\begin{equation}
 \frac{P_{\rm r}-H}{D}=R_{\rm r}(c)
 =\frac{\lambda}{c+\gamma}+\frac{\mu}{(c+\gamma)^2},\quad
 \lambda=\frac{1-\sigma}{a^2},\quad\mu=\frac{\sigma}{a}.
 \label{eq:rational-remainder}
\end{equation}
We have $\lambda\leq2$ and $\mu/t\leq2$.
Let $u=c+\gamma\geq t$. Since
$Q_{\rm r}(-c)\leq q-\alpha+2c\leq2u$,
$\gamma-c\leq\gamma+1\leq2/(1-\sigma)$, and $a\geq1-\sigma$,
\cref{eq:rational-determinant} implies
\begin{equation}
 \frac{J_{\rm r}(c)}{2D(c)Q_{\rm r}(-c)}
 \geq\frac{2a^2\sigma^2(1+2\sigma-2\sigma^3+\sigma^4)}
 {(1-\sigma)^4u^3(\gamma-c)^2}
 \geq\frac{\sigma^2}{2u^3}.
 \label{eq:rational-error-budget}
\end{equation}
Here $1+2\sigma-2\sigma^3+\sigma^4\geq1$; ratios at $c=\pm a$
are interpreted continuously. It is enough to find a polynomial $\Pi$
such that $0\leq\Pi-R_{\rm r}\leq\sigma^2/(2u^3)$.
Then $P=H+D\Pi$ has excess $E=P-P_{\rm r}$ satisfying
$0\leq E(c)\leq J_{\rm r}(c)/(2Q_{\rm r}(-c))\leq Q_{\rm r}(c)/2$.
The new diagonal stays positive. The two linear losses in the
determinant consume at most $J_{\rm r}$, leaving the nonnegative
term $E(c)E(-c)$. Also $P\geq P_{\rm r}\geq|\alpha+2c|$.

\emph{Polynomial cancellation of the double pole.}
Choose $N=\lceil4/\sqrt t\rceil$. Let $\mathsf T_N$ be the
Chebyshev polynomial and set
\begin{equation}
 f_N(c)=\frac{1-\mathsf T_N(-c)}{N^2(1+c)},\qquad
 F=f_N(-\gamma),\qquad
 \mathcal K(c)=\left(\frac{f_N(c)}{F}\right)^3.
 \label{eq:matching-kernel}
\end{equation}
The removable value $f_N(-1)=1$ makes $f_N$ a polynomial of degree
$N-1$. For $c=-\cos\theta$, it is
$[\sin(N\theta/2)/(N\sin(\theta/2))]^2$, so
\begin{equation}
 0\leq f_N(c)\leq\min\left\{1,\frac{2}{N^2(1+c)}\right\}
 \leq\frac{t+2/N^2}{u}\leq\frac{9t}{8u}.
 \label{eq:matching-kernel-localization}
\end{equation}
Write $h=\operatorname{arcosh}(1+t)$. For $x^2<12$, the factorial
bound $(2j)!\geq2\cdot12^{j-1}$ gives
\begin{equation}
 \cosh x-1\leq\frac{x^2}{2(1-x^2/12)}.
 \label{eq:matching-cosh-bound}
\end{equation}
Setting $x=(4/3)\sqrt t$ proves $h\geq(4/3)\sqrt t$ for
$t\leq3/4$. Also $4\leq N\sqrt t<5$. Each retained Taylor term
in $[\cosh(Nh)-1]/(N^2t)$ increases with $N\sqrt t$, so
\begin{equation}
 F\geq\sum_{j=1}^{5}\frac{(4/3)^{2j}4^{2j-2}}{(2j)!}
 =\frac{5297547256}{837019575}>\frac{25}{4},\qquad
 0\leq\mathcal K(c)\leq\left(\frac{9t}{50u}\right)^3.
 \label{eq:matching-kernel-bounds}
\end{equation}
All roots of $f_N$ lie in $[-1,1]$, as follows from
$\mathsf T_N(-c)=1$ and removal of the root at $c=-1$.
Therefore $f_N'/f_N<0$ at $-\gamma$, and $\mathcal K'(-\gamma)<0$.
Logarithmic differentiation gives
\begin{equation}
 d_N:=-t\mathcal K'(-\gamma)
 =3\left[\frac{Nt}{\sqrt{t(2+t)}}\coth\frac{Nh}{2}-1\right]
 <3\left(\frac{165}{32\sqrt2}-1\right)<8.
 \label{eq:matching-kernel-derivative}
\end{equation}
Indeed $Nh\geq16/3$ and
$e^{16/3}>\sum_{j=0}^4(16/3)^j/j!=19331/243>65$,
so $\coth(Nh/2)<33/32$. The last inequality follows from
$495^2<2\cdot352^2$.

Define the positive real constants
\begin{equation}
 C=\lambda-\mu\mathcal K'(-\gamma)>0,\qquad
 M=\frac{\mu}{t^2}+\frac C t,\qquad
 \Pi(c)=\frac{\lambda u+\mu-
          \mathcal K(c)(\mu+Cu-Mu^2)}{u^2}.
 \label{eq:matching-polynomial}
\end{equation}
The numerator and its first derivative vanish at $u=0$, because
$\mathcal K(-\gamma)=1$ and
$\mu\mathcal K'(-\gamma)+C=\lambda$.
Thus $\Pi$ is an actual polynomial of degree at most $3(N-1)$.
For $u\geq t$, the choice of $M$ gives $\mu+Cu-Mu^2\leq0$.
Moreover $tM=\lambda+(\mu/t)(1+d_N)<20$. Consequently
\begin{align}
 0\leq\Pi-R_{\rm r}
 &=\mathcal K\left(M-\frac C u-\frac\mu{u^2}\right)\notag\\
 &\leq M\mathcal K
 \leq\frac{20t^2}{u^3}\left(\frac9{50}\right)^3
 \leq\frac{1458\sigma^2}{3125u^3}<\frac{\sigma^2}{2u^3}.
 \label{eq:matching-weighted-error}
\end{align}
The last comparison is the exact integer inequality $2916<3125$.
This proves the sufficient error bound throughout the interval.

\emph{Degree and uniform precision.}
The polynomial $P=H+D\Pi$ has degree at most $3N+1$ and satisfies
\cref{eq:separator-conditions}. By \cref{lem:dihedral-saturation},
it gives an actual conditioned certificate at
$k=\lceil(3N+1)/2\rceil$. Since $t\geq\sigma\geq\epsilon/4$,
\begin{equation}
 k\leq\frac{6}{\sqrt t}+3\leq\frac{12}{\sqrt\epsilon}+3.
\end{equation}
This proves \cref{eq:matching-degree-upper}; the lower degree bound is
\cref{thm:unbounded-sm}.
\emph{Endpoint refinement.}
For $0<\sigma\leq1/16$, equivalently $30/\sqrt{287}\leq\alpha<2$,
choose instead $N=\lceil3/\sqrt t\rceil$. The identity
\begin{equation}
 (1+2\sigma-\sigma^2)(1-\sigma+2\sigma^2)^2-1
 =\sigma^3(8-9\sigma+12\sigma^2-4\sigma^3)>0
\end{equation}
gives $a\leq1-\sigma+2\sigma^2$ and hence
\begin{equation}
 \frac t\sigma\leq\frac{1+\sigma+2\sigma^2}{1-\sigma}
 \leq\frac76,\qquad t\leq\frac7{96}<\frac6{49},\qquad
 \lambda,\frac\mu t\leq\frac98.
\end{equation}
For the last bounds use $\lambda=(1-\sigma)(1+2\sigma-\sigma^2)
\leq1+\sigma$ and $\mu/t\leq1/a\leq1+\sigma$.
Equation~\eqref{eq:matching-cosh-bound} now gives
$h\geq(7/5)\sqrt t$, while $3\leq N\sqrt t<10/3$. Thus
\begin{equation}
 F\geq\sum_{j=1}^{5}\frac{(7/5)^{2j}3^{2j-2}}{(2j)!}
 =\frac{224174937917}{62500000000}>\frac72,\qquad
 \mathcal K(c)\leq\left(\frac{22t}{63u}\right)^3.
\end{equation}
The rational Taylor bound
$e^{21/5}>\sum_{j=0}^9(21/5)^j/j!=41215195639/625000000>65$
again gives $\coth(Nh/2)<33/32$. Therefore
\begin{equation}
 d_N<3\left(\frac{10}{3}\frac{33}{32\sqrt2}-1\right)<\frac92,
 \qquad tM\leq\frac{117}{16},
\end{equation}
where the derivative comparison reduces to $121<128$. With the same
polynomial cancellation,
\begin{equation}
 0\leq\Pi-R_{\rm r}
 \leq\frac{117}{16}\left(\frac76\right)^2
       \left(\frac{22}{63}\right)^3\frac{\sigma^2}{u^3}
 =\frac{17303}{40824}\frac{\sigma^2}{u^3}
 <\frac{\sigma^2}{2u^3}.
\end{equation}
This gives $k\leq(9/2)t^{-1/2}+3\leq9\epsilon^{-1/2}+3$.

Take $\epsilon_k=81/(k-3)^2$ for sufficiently large $k$.
The refined bound covers $30/\sqrt{287}\leq\alpha\leq2-\epsilon_k$,
and \cref{thm:fixed-tilt-sm} covers the complementary compact interval
at a fixed finite level. Proposition~\ref{prop:closure-to-precision}
gives the eventual $k^{-4}$ upper bound. More precisely,
\begin{equation}
 \Delta_k\leq\frac{\epsilon_k^2}{8}
       \sup_{[2-\epsilon_k,2]}q''
 =\left(\frac1{32}+o(1)\right)\epsilon_k^2,
\end{equation}
because $q''(2)=1/4$. This proves the upper constant in
\cref{eq:matching-leading-bounds}. The lower constant follows directly
from the Fej\'er witness in \cref{eq:path-value-gap}, independently of
the precision upper bound. Monotonicity of $\Delta_k$ then yields
the tolerance order in \cref{eq:matching-rates}.
\end{proof}

The command \texttt{python scripts/verify\_matching\_upper.py} checks
the rational factorizations, pole cancellation, contact identities and
exact safety constants. It also constructs complete rational polynomials
at three tilts and divides out the double pole exactly. The all-parameter
inequalities above establish the rate; these audits do not replace them.
The prefactors 12 and 9 follow from rational safety estimates.
They improve a more conservative construction but are not asserted
optimal; the proof concerns certificate existence, not practical
solver performance.

\begin{corollary}[Sharp uniform randomness precision]
\label{cor:matching-randomness}
Let $H_k(s)=-\log_2G_k^{\rm os}(s)$ and $H_Q(s)=-\log_2G_Q(s)$.
For the single-round CHSH certification task, including quantum side information,
\begin{equation}
 \begin{split}
 \sup_{2\leq s\leq2\sqrt2}[G_k^{\rm os}(s)-G_Q(s)]&=\Theta(k^{-4}),\\
 \sup_{2\leq s\leq2\sqrt2}[H_Q(s)-H_k(s)]&=\Theta(k^{-4}).
 \end{split}
 \label{eq:matching-randomness}
\end{equation}
The same holds when the available statistic is $S\geq s$.
\end{corollary}
\begin{proof}
For sufficiently large $k$, the preceding proof gives exact tilted bounds
for all $0\leq\alpha\leq\alpha_*=2-\epsilon_k$, with
$\epsilon_k=81/(k-3)^2\leq1/2$.
Set $s_*=8/q(\alpha_*)$ and $h_Q(s)=\sqrt{2-s^2/4}$.
The contact criterion gives exact randomness for $s_*\leq s<2\sqrt2$;
at $2\sqrt2$ take the exact positive-tilt bounds as $\alpha\downarrow0$.
For $2\leq s\leq s_*$ the exact bound at $\alpha_*$ gives
\begin{equation}
 2G_k^{\rm os}(s)-1\leq
 \frac{q(\alpha_*)-s}{\alpha_*}
 =h_Q(s_*)+h_Q'(s_*)(s-s_*).
\end{equation}
On this interval $h_Q\geq3/4$ and
$|h_Q''|=1/(2h_Q^3)\leq4$.
Also $q(\alpha_*)\geq4-\epsilon_k$ gives $s_*-2\leq\epsilon_k$.
Taylor's theorem therefore bounds the guessing-probability error by
$\epsilon_k^2$, uniformly in $s$.
Since $G_Q\geq1/2$, the entropy error is at most
$2\epsilon_k^2/\ln2$. The witness in \cref{eq:entropy-witness-scale-sm}
gives matching $k^{-4}$ lower orders for both quantities.
Positive-tilt supporting bounds remain valid under $S\geq s$,
and the same witnesses satisfy that constraint.
The limiting constants can also be controlled. Since
$s_*-2=\epsilon_k/2+O(\epsilon_k^2)$ and $|h_Q''(2)|=1/2$,
the tangent estimate gives
\begin{equation}
 \sup_s(G_k^{\rm os}-G_Q)\leq
 \left(\frac1{32}+o(1)\right)\epsilon_k^2.
\end{equation}
Only the shrinking interval $[2,s_*]$ can contribute, and there
$G_Q(s)\geq G_Q(s_*)\to1$. Using the derivative of $\log G$ and
the same lower witnesses therefore gives
\begin{equation}
 \begin{gathered}
 \frac1{32}\leq\liminf k^4\sup_s(G_k^{\rm os}-G_Q)
 \leq\limsup k^4\sup_s(G_k^{\rm os}-G_Q)\leq\frac{6561}{32},\\
 \frac1{32\ln2}\leq\liminf k^4\sup_s(H_Q-H_k)
 \leq\limsup k^4\sup_s(H_Q-H_k)\leq\frac{6561}{32\ln2}.
 \end{gathered}
 \label{eq:randomness-leading-bounds}
\end{equation}
\end{proof}

\subsection{Uniform precision, simultaneous rounding, and physical scales}
\label{sec:precision}

\begin{lemma}[Simultaneous positive-functional rounding]
\label{lem:rounding}
Let $\phi_{a|x}$ be feasible native level-$k$ functionals on
$G=\mathbb Z_2*\mathbb Z_2$. There exist positive functionals
$\widetilde\phi_{a|x}$ on the full group algebra with common outcome
sum and the same outcome masses, such that
\begin{equation}
 \widetilde\phi_{a|x}(B_y)=\rho_k\phi_{a|x}(B_y),\qquad
 \rho_k=\cos\frac{\pi}{2k+2}.
 \label{eq:rounding-first-moments}
\end{equation}
They define an all-level feasible assemblage and hence a quantum
behavior for the Bell bounds and guessing tradeoffs considered here.
\end{lemma}
\begin{proof}
Let $W_k$ be the reduced word ball. First take any real window $c$ supported
on $W_k$ with $\sum_u c(u)^2=1$, and define
\begin{equation}
 h_c(g)=\sum_{u\in G}c(u)c(ug),\qquad
 \widetilde\phi(g)=
 \begin{cases}
 h_c(g)\phi(g),& |g|\leq2k,\\
 0,& |g|>2k.
 \end{cases}
 \label{eq:window-extension}
\end{equation}
The support condition implies $h_c(g)=0$ for $|g|>2k$, so this is
well-defined from truncated data. For arbitrary $g_1,\ldots,g_m\in G$
and complex $z_1,\ldots,z_m$, its quadratic form is
\begin{equation}
 \begin{split}
 &\sum_{i,j}\overline z_i z_j\widetilde\phi(g_i^{-1}g_j)\\
 &=\sum_{v\in G}\sum_{\substack{i,j:\,vg_i,vg_j\in W_k}}
 \overline{z_i c(vg_i)}z_jc(vg_j)
 \phi\big((vg_i)^{-1}(vg_j)\big)\geq0.
 \end{split}
 \label{eq:window-positivity}
\end{equation}
Only finitely many $v$ contribute. Each inner sum is a quadratic form
in the original level-$k$ moment matrix, with repeated indices combined
if needed. Thus $\widetilde\phi$ is positive on every square in the
full group algebra. Applying the same linear map to all four blocks
preserves their common outcome sums; $h_c(I)=1$ preserves each mass.
This simultaneous compatibility is the relevant distinction from
extending individual positive functions separately
\cite{BakonyiTimotin2011}.

The right Cayley graph induced on $W_k$ is a path with $2k+1$ vertices.
In path order $u_0,\ldots,u_{2k}$, choose
\begin{equation}
 \theta_k=\frac{\pi}{2k+2},\qquad
 c(u_i)=\frac{\sin((i+1)\theta_k)}{\sqrt{k+1}}.
 \label{eq:sine-window}
\end{equation}
The sine-square sum is $k+1$, so the window is normalized. The path
adjacency has this vector as an eigenvector of eigenvalue
$2\cos\theta_k$, by
$\sin(i\theta)+\sin((i+2)\theta)=2\cos\theta\sin((i+1)\theta)$
and the zero boundary values. Hence
$h_c(B_0)+h_c(B_1)=2\cos\theta_k$.
Reflection of this odd-vertex path exchanges the two edge colors
and leaves the sine vector invariant, so the two overlaps are equal.
This proves \cref{eq:rounding-first-moments}.

For completeness, the all-level assemblage has a commuting realization.
Apply the GNS construction to its normalized common positive functional
$\widetilde L$. Each $0\leq\widetilde\phi_{a|x}\leq\widetilde L$
is represented by a positive contraction $E_{a|x}$ in the commutant
of Bob's representation, with $\sum_aE_{a|x}=I$ on the cyclic space.
Indeed, the sesquilinear form
$\widetilde\phi_{a|x}(p^*q)$ is bounded by the GNS inner product and
is invariant under simultaneous left multiplication by group unitaries;
its representing operator therefore lies in that commutant.
These effects are Alice's POVMs and give the stated moments.
They may be dilated to projective measurements commuting with Bob.
The tilted-CHSH bound used here is valid in this commuting model by
its operator SOS, as is the randomness bound obtained from its tangents.
No tensor-product versus commuting-model equality for general games
is needed.
\end{proof}

The common first-moment contraction is optimal within this scalar-window
construction: for any normalized window on $W_k$,
$\min\{h_c(B_0),h_c(B_1)\}\leq[h_c(B_0)+h_c(B_1)]/2\leq\rho_k$
by the largest eigenvalue of the path adjacency. This does not prove
that the resulting Bell-error upper bound is optimal. A sharper rate
would require additional information about the feasible Bell data or
a different rounding, rather than only a better scalar window.

\begin{proposition}[Uniform error and tolerance cost]
\label{prop:precision-sm}
For every $k\geq1$,
\begin{equation}
 \frac{32}{625k^4}<\Delta_k
 \leq2\sqrt2(1-\rho_k)
 \leq\frac{\pi^2}{2\sqrt2(k+1)^2},\qquad
 \liminf_{k\to\infty}k^4\Delta_k\geq\frac18.
 \label{eq:precision-bounds-sm}
\end{equation}
In particular $K_{\rm tol}(\delta)$ is finite for every $\delta>0$, and
\begin{equation}
 K_{\rm tol}(\delta)=\Theta(\delta^{-1/4}),\qquad
 \Delta_k=\Theta(k^{-4}).
 \label{eq:tolerance-cost-sm}
\end{equation}
These are worst-case uniform bounds over the fixed family, not lower
bounds for every fixed tilt.
\end{proposition}
\begin{proof}
The lower bound in \cref{eq:path-value-gap} and
$r_k\geq1/(10k^2)$, $1-2r_k\geq4/5$ give the first inequality.
At $\epsilon=8r$, direct expansion gives
\begin{equation}
 w-q(2-8r)=8r^2-16r^3+O(r^4).
 \label{eq:gap-expansion-sm}
\end{equation}
Since $k^2r_k\to1/8$, this proves the liminf bound.
For the upper bound, write $h$ and $g$ for the Alice bias and CHSH
value of any feasible level-$k$ assemblage. The rounded assemblage has
values $h,\rho_kg$, so $\alpha h+\rho_kg\leq q(\alpha)$.
The level-one CHSH certificate (\cref{rem:chsh-nice}) already gives $g\leq2\sqrt2$ at every
$k\geq1$. Therefore
\begin{equation}
 \alpha h+g-q(\alpha)\leq(1-\rho_k)g
 \leq2\sqrt2(1-\rho_k).
\end{equation}
Take the maximum over feasible assemblages and $\alpha$; use
$1-\cos\theta\leq\theta^2/2$ for the last bound.
If $\Delta_k\leq\delta$, then
$k>(32/(625\delta))^{1/4}$. Conversely any $k\geq1$ with
$k+1\geq\pi/(2^{3/4}\sqrt\delta)$ suffices. Finite-level dual
attainment converts these value bounds into actual degree-$k$
certificates for $q(\alpha)+\delta$.
The sharper asymptotic orders in \cref{eq:tolerance-cost-sm} follow
from \cref{thm:matching-upper}; the rounding gives the separate,
explicit upper envelope valid at every $k$.
\end{proof}

The rounding also applies to any real Bell functional written in
observable coordinates
$F=c+\sum_xa_xA_x+\sum_yb_yB_y+\sum_{xy}c_{xy}A_xB_y$.
Writing $M_B=\sum_y|b_y|+\sum_{xy}|c_{xy}|$, it gives
\begin{equation}
 0\leq\omega_k^{\rm os}(F)-\omega_{\rm qc}(F)
 \leq M_B(1-\rho_k).
 \label{eq:general-rounding-error}
\end{equation}
Here $\omega_{\rm qc}$ explicitly denotes the commuting value.
All moments containing Bob are contracted by $\rho_k$, while
Alice-only moments and constants are unchanged; each contracted
moment has magnitude at most one. The stronger constant in
\cref{eq:precision-bounds-sm} uses the CHSH bound instead of this
coefficient estimate. We do not claim optimality of either upper bound.
For $\widehat F_\alpha=F_\alpha/(4+\alpha)$, the uniform normalized
gap lies between $\Delta_k/6$ and $\Delta_k/4$, so both exponents
in \cref{eq:tolerance-cost-sm} are unchanged.

In particular, exact nonuniformity must not be interpreted as
nonuniformity at a fixed nonzero absolute tolerance. Independently of
the constructive estimate, this follows from compactness: each
$f_k(\alpha)=\omega_k^{\rm os}(F_\alpha)$ is 1-Lipschitz because
$|h|\leq1$, the functions decrease with $k$, and their pointwise
limit is the continuous $q$ on $[0,2]$. Dini's theorem gives uniform
convergence. The rounding supplies an explicit rate for that conclusion.

\paragraph{Physical scale of the obstruction.}
The local tilted bound is $2+\alpha$. Thus, for $\epsilon\downarrow0$,
\begin{align}
 q(2-\epsilon)-(4-\epsilon)
 &=\frac{\epsilon^2}{8}+\frac{\epsilon^3}{32}+O(\epsilon^4),\\
 1-\frac{4-\epsilon}{q(2-\epsilon)}
 &=\frac{\epsilon^2}{32}+\frac{\epsilon^3}{64}+O(\epsilon^4).
 \label{eq:noise-scale-sm}
\end{align}
For an optimal strategy mixed with uniform white noise of fraction
$1-v$, all terms of $F_\alpha$ scale by $v$. The second line is exactly
the largest noise fraction compatible with a strict violation, with
strict inequality at the threshold. Along $\epsilon=8r_k$, the explicit
witness excess divided by the ideal quantum violation tends to one;
both absolute quantities are asymptotic to $8r_k^2$. The tolerated
white-noise fraction is of order $k^{-4}$. This assertion concerns this
noise model and these chosen Bell tests; it is not a general noise
threshold for all measurements or all tasks.

\paragraph{Precision of randomness certification.}
Applying \cref{lem:rounding} to each subnormalized Eve branch preserves
its mass and the corresponding guessing probability, and multiplies
the total CHSH value by $\rho_k$. Consequently, for $2\leq s\leq2\sqrt2$,
\begin{equation}
 G_Q(s)\leq G_k^{\rm os}(s)
 \leq G_Q(\max\{2,\rho_ks\}).
 \label{eq:guess-rounding-sm}
\end{equation}
If $\rho_ks<2$, the right side is the probability bound 1.
The same argument works for $S\geq s$ by monotonicity of $G_Q$.
It gives, uniformly in $s$,
\begin{align}
 0\leq G_k^{\rm os}(s)-G_Q(s)
 &\leq\frac{\sin\theta_k}{\sqrt2},\\
 0\leq H_Q(s)-H_k^{\rm os}(s)
 &\leq\frac{\sqrt2\sin\theta_k}{\ln2},
 \label{eq:entropy-rounding-sm}
\end{align}
where $H=-\log_2G$. To prove the first estimate use
$\sqrt{a+b}-\sqrt a\leq\sqrt b$ with
$a=2-s^2/4$ and $b=s^2(1-\rho_k^2)/4$; clipping at 2 only lowers
the bound. The second follows from $G_Q\geq1/2$ and
$\ln(1+x)\leq x$. Away from maximal CHSH, the derivative instead gives
\begin{equation}
 G_k^{\rm os}(s)-G_Q(s)
 \leq\frac{s^2(1-\rho_k)}{8\sqrt{2-s^2/4}}
 \quad(2\leq s<2\sqrt2),
 \label{eq:guess-away-endpoint-sm}
\end{equation}
so the guaranteed error is $O(k^{-2})$ on compact subintervals away
from $2\sqrt2$, compared with the uniform $O(k^{-1})$ bound.
These are explicit rounding guarantees valid at every level.
The sharp uniform asymptotic rate is instead $\Theta(k^{-4})$ by
\cref{cor:matching-randomness}.

The separating witness gives a complementary lower bound at a
different sequence of statistics, $\bar s_k=2+4r_k$:
\begin{equation}
 \begin{split}
 G_k^{\rm os}(\bar s_k)-G_Q(\bar s_k)
 &\geq1-r_k-G_Q(2+4r_k)\\
 &=2r_k^2+4r_k^3+O(r_k^4),\\
 H_Q(\bar s_k)-H_k^{\rm os}(\bar s_k)
 &\geq\log_2\frac{1-r_k}{G_Q(2+4r_k)}\\
 &=\frac{2r_k^2+6r_k^3}{\ln2}+O(r_k^4).
 \end{split}
 \label{eq:entropy-witness-scale-sm}
\end{equation}
The exact guessing-probability difference is
$4r_k^2/[1-2r_k+\sqrt{1-4r_k-4r_k^2}]\geq2r_k^2$;
the radical is positive because $0<r_k\leq1/10$.
These witness lower bounds must not be identified with the optimal
relaxation gaps. Also $\bar s_k$ is not the tangent statistic
$s_k=8/q(\alpha_k)$ in \cref{cor:unbounded-randomness}.

If an experiment supplies a valid lower confidence bound $s_L$ for
the relevant CHSH statistic, insert $s_L$ into
\cref{eq:guess-rounding-sm} to obtain a conservative single-round
min-entropy guarantee. For $s_L\leq2$ use zero. This implication does
not specify a sampling model or treat repeated trials as independent;
constructing the confidence bound and a finite-key or sequential
security proof remains separate from the algebraic relaxation analysis.

\paragraph{Verification scope.}
The additional command
\texttt{python scripts/verify\_precision\_bounds.py} checks the symbolic
gap, noise and entropy expansions, the sine recurrence, exact window
normalizations and generator overlaps at several levels, and the
translation decomposition \cref{eq:window-positivity} on larger word
sets for all four rational witness blocks. A complex-valued quantum
example independently audits the multiplication orientation and the
common outcome sums. These finite audits supplement the all-level
proof above; no numerical optimizer is used and no sampled positivity
test replaces that proof.

\subsection{An explicit counterexample and solver-free audits}

For $k=3$, use path order
\begin{equation*}
 (B_0B_1B_0,B_0B_1,B_0,I,B_1,B_1B_0,B_1B_0B_1).
\end{equation*}
Here $M_{ij}=1-|i-j|/6$, $r_3=1/74$, and the inverse vector is
$z=(7,-12,12,-12,12,-12,7)^T$ up to the common sign of $v$.
The other two blocks follow explicitly from
$L_y(w)=[\tau_3(w)+\tau_3(B_yw)]/2$, so all four $7\times7$
rational matrices are specified. The exact values are
\begin{equation}
 \alpha_3=\frac{70}{37},\quad w_3=\frac{5332}{1369},\quad
 q(\alpha_3)=\frac{4\sqrt{1297}}{37},\quad
 w_3^2-q(\alpha_3)^2=\frac{20736}{1874161}>0.
\end{equation}

The command \texttt{python scripts/higher\_level\_witness.py {-}{-}fejer}
checks seven symbolic identities for the inverse-vector argument,
the inverse quadratic form, the general and specialized gaps, and
the exact level-two endpoint. Independently, at $k=1,2,3,4,5,8$ it
reconstructs the full quotient-functional Grams and every original
PVM moment, checks the moving-average factorization and inverse
vector, and verifies exact PSD, congruence, moment identities,
normalization, outcome sums, and objective values. The all-level
proof is \cref{lem:path-psd,thm:unbounded-sm}; finite tests are audits.
The older unweighted construction remains independently checked by
\texttt{{-}{-}analytic} as a comparison.

Two auxiliary rational witnesses at $\alpha=9/5$ give gaps greater
than $1/1000$ at level three and $1/10000$ at level four. Their
complete matrices are in
\datapath{data/analytic/higher_level/alpha_9_5_L3.json} and
\datapath{data/analytic/higher_level/alpha_9_5_L4.json}; all blocks are
positive definite. These examples are independent cross-checks,
not prerequisites of the analytic theorem.

\begin{corollary}[No finite conditioned level gives the full randomness tradeoff]
\label{cor:unbounded-randomness}
For $s_k=8/q(\alpha_k)$ and every integer $k\geq1$,
\begin{equation}
 G_k^{\rm os}(s_k)>G_Q(s_k)=G_2^{\rm std}(s_k),\qquad s_k\downarrow2.
\end{equation}
Standard level two gives $G_Q(s)$ for every $2\leq s\leq2\sqrt2$.
For every fixed $2<s<2\sqrt2$, some finite conditioned level also
gives $G_Q(s)$; no single finite level works for the entire interval.
The same statements hold with the constraint $S\geq s$.
\end{corollary}
\begin{proof}
Apply the contact criterion in \cref{sec:randomness} to
\cref{thm:unbounded-sm} and the standard SOS valid for all $0<\alpha<2$.
This proves the separation assertions for interior CHSH statistics.
For the pointwise finite-exactness assertion, choose the unique
$0<\alpha<2$ with $s=8/q(\alpha)$ and apply
\cref{thm:fixed-tilt-sm} and the same contact criterion.
At $s=2$,
the probability bound and a deterministic strategy give $G=1$.
At $s=2\sqrt2$, the standard tilted bounds imply
$h\leq[q(\alpha)-2\sqrt2]/\alpha\to0$ as $\alpha\downarrow0$,
giving $G=1/2$, attained by the quantum strategy. For $S\geq s$,
each positive-tilt tangent remains a valid upper bound, while a witness
at the exact statistic remains feasible for the lower-bound constraint.
\end{proof}

\subsection{Reverse conversion and the nice-SOS soundness route}

\begin{proposition}[Reverse conversion at every native level]
\label{prop:reverse-conversion}
For every $k\geq1$, $\mathsf O_k\subseteq\mathsf D_{k+1}$. Therefore
\begin{equation}
 \mathsf D_1\subsetneq\mathsf O_1\subseteq\mathsf D_2
 \not\subseteq\mathsf O_k\quad\text{for all }k\geq1.
\end{equation}
The best uniform additive overhead in the reverse direction is one.
\end{proposition}
\begin{proof}
Factor the positive Gram matrices of a conditioned dual certificate.
Modulo the original quotient, it is a sum of terms
$P_{A,a|x}f(B)^*f(B)$ with $\deg f\leq k$; normalization and
outcome-consistency multipliers cancel in the polynomial identity.
By commutation and idempotency,
\begin{equation}
 P_{A,a|x}f(B)^*f(B)
 =(f(B)P_{A,a|x})^*(f(B)P_{A,a|x}).
\end{equation}
Each factor belongs to the standard word span of degree $k+1$.
This remains true for the eliminated outcome
$P_{A,1|x}=I-P_{A,0|x}$, so no outcome is omitted.
Alternatively, for a feasible standard level-$(k+1)$ moment matrix,
restrict to $\{uP_{A,a|x}:u\in W_k\}$ for each $a,x$.
Both outcome blocks are PSD by congruence, and summing them gives
the common Bob moment functional by completeness. This proves the
same inclusion by support-function duality and attainment.
The strict inclusion $\mathsf D_1\subsetneq\mathsf O_1$ rules out
zero uniform overhead; it does not imply that every individual
certificate or every level needs the extra degree.
\end{proof}

\begin{corollary}[Degree of exact nice-SOS inputs to compiled soundness]
\label{cor:compiled-degree}
Consider a soundness argument that uses an exact nice-SOS identity
for $q(2-\epsilon)I-F_{2-\epsilon}$ in the stated raw-PVM convention,
with Alice as the conditioned party, the marginal on Alice, and
Bob-word factors of degree at most $d$.
Then $d=\Omega(\epsilon^{-1/2})$, with the lower asymptotic constant
in \cref{eq:fejer-liminf}. At $\alpha_k$, any such identity for
$[q(\alpha_k)+\delta]I-F_{\alpha_k}$ requires $d>k$ whenever
$0\leq\delta\leq8r_k^2(1-2r_k)^2$.
\end{corollary}
\begin{proof}
A nice-SOS identity with this factor degree is an $\mathsf O_d$
certificate. Apply \cref{eq:degree-lower-growth,eq:fejer-liminf};
for the relaxed target, evaluate on the witness and use
\cref{eq:path-value-gap}.
\end{proof}

The relevance to compiled nonlocal games is the certificate-to-soundness
route of \cite{CuiFalorNatarajanZhang2025}: the algebraic target bound
is supplemented by a cryptographic error depending on the certificate.
Our corollary limits the degree of that input, not every possible
soundness proof, and supplies no bound on the security parameter or
coefficient norms. The statement for exact targets is pointwise in each
fixed tilt; simultaneous limits in the tilt and security parameter need
separate uniform estimates. A POVM/localizer filtration such as
\cite{KlepEtAl2025} requires an explicit degree map before transferring
this quantitative statement. The fixed-party distinction for the
extended expression of \cite{MehtaPaddockWooltorton2025} is given in
\cref{sec:open}.

If exchanging the parties is permitted, \cref{prop:opposite-tilt}
provides an exact degree-one input for the opposite orientation.
The corollary therefore does not give a lower bound for all role
assignments or all compiled protocols. The Heisenberg-picture sequential
hierarchy of \cite{BaroniEtAl2025Multipartite} retains different
operator relations; transferring a finite-level obstruction to that
hierarchy would require an explicit comparison not proved here.

\section{Direction, contact geometry, and uniform finite closure}
\label{sec:structure}

\subsection{Exact certification after exchanging the tilt direction}
\label{sec:orientation}

The conditioned party is fixed throughout this subsection. Write
$S=A_0(B_0+B_1)+A_1(B_0-B_1)$ and
$\widetilde F_\alpha=\alpha B_0+S$. Party exchange preserves the
CHSH expression and the quantum value $q=\sqrt{8+2\alpha^2}$,
but it exchanges the two conditioned certificate classes.

\begin{proposition}[Opposite tilt is exact at conditioned level one]
\label{prop:opposite-tilt}
For every $0\leq\alpha\leq2$,
$qI-\widetilde F_\alpha\in\mathsf O_1$. Consequently, for $0<\alpha<2$,
\begin{equation}
 d_{A\text{-cond}}(\widetilde F_\alpha,q)=1,
 \qquad d_{\rm std}(\widetilde F_\alpha,q)=2.
 \label{eq:opposite-degrees}
\end{equation}
In the fixed Alice-conditioned hierarchy, the complete CHSH tradeoff
for guessing Bob's output at $y=0$ is already exact at level one,
including quantum side information and the constraint $S\geq s$.
\end{proposition}
\begin{proof}
Use $P_{a|x}=(I+(-1)^aA_x)/2$ and set
\begin{equation}
 V=\frac q2 I+\left(1-\frac\alpha2\right)B_0
   +\left(1+\frac\alpha2\right)B_1
   +\frac{4-\alpha^2}{4q}(B_0B_1+B_1B_0).
 \label{eq:opposite-V}
\end{equation}
Direct expansion in the original quotient gives
\begin{align}
 qI-\widetilde F_\alpha
 &=P_{0|0}(V-2B_0-2B_1)+P_{1|0}V\notag\\
 &\quad+P_{0|1}(qI-V-\alpha B_0+2B_1)
       +P_{1|1}(qI-V+(2-\alpha)B_0).
 \label{eq:opposite-certificate}
\end{align}
For $B_0=Z$, $B_1=cZ+\sqrt{1-c^2}X$, $-1\leq c\leq1$,
the four Bob blocks have the same determinant
\begin{equation}
 \det S_{a|x}=\frac{c^2(4-\alpha^2)^2}{4q^2}\geq0.
 \label{eq:opposite-determinants}
\end{equation}
Their scalar parts are $v_0$ or $q-v_0$, where
$v_0=q/2+(4-\alpha^2)c/(2q)$. Both are bounded below by
$(4+3\alpha^2)/(2q)>0$. Thus all blocks are PSD in every
two-dimensional representation, including the reducible endpoints.
These representations also cover the scalar characters.
Each block has word length at most two, so
\cref{lem:dihedral-saturation} gives Bob degree-one SOS factors.
Equation~\eqref{eq:opposite-certificate} is therefore an actual
$\mathsf O_1$ certificate. Attainment and the standard-level comparison
follow by party exchange from tilted CHSH.

For Bob's guessing probability, simultaneous outcome reversal again
preserves $S$ and reverses the target bias. Apply the argument of
\cref{lem:guess-contact} with $B_0$ as the target observable.
All tangent inequalities $s+\alpha\langle B_0\rangle\leq q(\alpha)$
are exact at level one, so the full curve is
$(1+\sqrt{2-s^2/4})/2$. At $s=2$ use the probability cap; at
$s=2\sqrt2$ take the limit $\alpha\downarrow0$. The curve is decreasing,
which also proves the version with $S\geq s$.
\end{proof}

In particular, minimizing over the two conditioning directions gives
degree one for every member of the singly tilted family. The
divergence in \cref{thm:unbounded-sm} concerns a fixed direction and,
in the randomness application, a fixed target output. Exchanging that
output changes the certification task. The difference is a property
of the relaxation, not intrinsic directional asymmetry of the shared
state. Equation~\eqref{eq:opposite-certificate} is supplied as an explicit
comparison; no priority claim is made for low-degree certificates of
Bob-marginal tilted expressions.

\subsection{A general separator formulation}
\label{sec:general-separators}

Allow any fixed number $m\geq2$ of binary Alice questions, still with
two binary Bob questions. Every real Bell polynomial has the form
\begin{equation}
 G=L_0(B)+\sum_{x=1}^m A_x L_x(B),
 \label{eq:general-bell-form}
\end{equation}
where each $L_x$ is affine in $B_0,B_1$. The conditioned cone is defined
by its $2m$ outcome blocks with a common sum, using the same Bob degree.

\begin{lemma}[Polynomial separator criterion]
\label{lem:general-separator}
For $k\geq1$, $\beta I-G\in\mathsf O_k$ if and only if there are
Hermitian Bob polynomials $H_x$ of word length at most $2k$ such that
\begin{equation}
 \sum_x H_x=\beta I-L_0,\qquad H_x\pm L_x\succeq0
 \quad\text{in every representation.}
 \label{eq:general-separator}
\end{equation}
\end{lemma}
\begin{proof}
For necessity, collect the Bob SOS polynomials $S_{\pm|x}$ in an
actual conditioned certificate. Comparing the independent
coefficients of $I,A_1,\ldots,A_m$ in the quotient yields
$(S_{+|x}-S_{-|x})/2=-L_x$ and
$\sum_x(S_{+|x}+S_{-|x})/2=\beta I-L_0$.
Take $H_x=(S_{+|x}+S_{-|x})/2$. Conversely,
\cref{lem:dihedral-saturation} supplies degree-$k$ SOS factors for
$H_x\mp L_x$, and
\begin{equation}
 \beta I-G=\sum_x\bigl[P_{+|x}(H_x-L_x)+P_{-|x}(H_x+L_x)\bigr].
\end{equation}
This uses actual positive polynomials, rather than an assumption about
compatible extensions of separate truncated moment blocks.
\end{proof}

At an attaining strategy with conditional Bob states $\rho_{\pm|x}$,
every summand has nonnegative expectation, so
\begin{equation}
 (H_x\mp L_x)\rho_{\pm|x}^{1/2}=0.
 \label{eq:forced-contact-kernel}
\end{equation}
This follows from $\operatorname{tr}(S\rho)=0$ for $S,\rho\succeq0$,
which implies $S^{1/2}\rho^{1/2}=0$. These forced kernels explain why
an unweighted small approximation error need not preserve exactness.
The next theorem keeps the zeros while approximating the remainder.

\begin{theorem}[Contact-preserving approximation]
\label{thm:contact-closure}
Consider a compact parameter family of polynomials
\eqref{eq:general-bell-form} with attained targets $\beta$.
Fix a Hermitian Bob polynomial $T$ of degree at most $\ell\geq1$
with $\|T\|\leq1$ in every representation. Fix Hermitian polynomials
$E_1,\ldots,E_p$ of degree at most $r$, with $\|E_j\|\leq1$ and
$[E_j,T]=0$. Suppose there are Hermitian separators satisfying
\eqref{eq:general-separator} in the universal $C^*$-algebra and,
for $x<m$, representations
\begin{equation}
 H_x=J_x+W\sum_{j=1}^p R_{xj}(T)E_j,\qquad W=w(T),
 \label{eq:weighted-remainder-form}
\end{equation}
where the real polynomial $w$ is nonnegative on $[-1,1]$ and has
degree at most $h$, the Hermitian polynomial $J_x$ has degree at
most $d$, and the real functions $R_{xj}$ are jointly continuous
in the family parameter and in $t\in[-1,1]$. The polynomials $w,J_x$
may depend on the parameter, with the stated common degree bounds.
Assume a uniform $\mu>0$ such that, in every representation,
\begin{equation}
 H_x\pm L_x\succeq\mu W\quad(1\leq x\leq m).
 \label{eq:weighted-positive-margin}
\end{equation}
Then one finite conditioned level certifies all these targets exactly.
If the $R_{xj}$ have uniform Lipschitz constants $L_{xj}$ and
$C=\sum_{x<m,j}L_{xj}$, an explicit sufficient level is
\begin{equation}
 \begin{split}
 n&=\max\{1,\lceil(2C/\mu)^2\rceil\},\\
 K&=\max\left\{1,\left\lceil
       \frac{\max\{d,1,\ell(n+h)+r\}}2\right\rceil\right\}.
 \end{split}
 \label{eq:contact-degree-bound}
\end{equation}
\end{theorem}
\begin{proof}
For $x<m$ approximate each scalar $R_{xj}$ by its degree-$n$
Bernstein polynomial $R_{xj,n}$ on $[-1,1]$. Joint continuity on a
compact set gives uniform approximation over the parameter family.
Choose $n$ so the sum of the $p(m-1)$ uniform errors is at most
$\mu/2$. Define
\begin{equation}
 \widehat H_x=J_x+W\sum_jR_{xj,n}(T)E_j\ (x<m),\qquad
 \widehat H_m=\beta I-L_0-\sum_{x<m}\widehat H_x.
\end{equation}
This preserves the separator sum as an exact polynomial identity.
Each approximation error commutes with $W$. Since $\|E_j\|\leq1$,
functional calculus and the triangle inequality give
\begin{equation}
 -\tfrac\mu2W\preceq\widehat H_x-H_x\preceq\tfrac\mu2W
 \quad\text{for every }x,
 \label{eq:weighted-error-bound}
\end{equation}
including $x=m$, whose error is the negative sum of the other errors.
Explicitly, set $D_x=\sum_j(R_{xj,n}-R_{xj})(T)E_j$ for $x<m$
and $D_m=-\sum_{x<m}D_x$. Then $\|D_x\|\leq\mu/2$ and
$\widehat H_x-H_x=WD_x=W^{1/2}D_xW^{1/2}$, proving the order
bound without dividing by $W$ at its zeros.
Thus $\widehat H_x\pm L_x\succeq\mu W/2\succeq0$.
The error vanishes on $\ker W$; no strict positivity at a forced
contact is required.

For a Lipschitz function on $[-1,1]$, the binomial expectation formula
gives $\|R_{xj,n}-R_{xj}\|_\infty\leq L_{xj}/\sqrt n$.
Consequently the stated choice of $n$ implies
\eqref{eq:weighted-error-bound}. Every $\widehat H_x\pm L_x$ has
word length at most $\max\{d,1,\ell(n+h)+r\}$, including the last
block fixed by the sum identity. Apply \cref{lem:general-separator}
at level \eqref{eq:contact-degree-bound}.
\end{proof}

The simplest case has $p=1,E_1=I$ and uses scalar functions of a
single contraction $T$. A matrix-valued version is included by taking
the central contraction $T=(B_0B_1+B_1B_0)/2$ and commuting Hermitian
coefficient polynomials $E_j$, which need not commute with one another.
The proof needs commutation with $T$, not pairwise commutation.
For rational remainders, absence of poles on a compact parameter--spectral
domain supplies bounded derivatives; the remaining requirements are
the exact sum identity and the weighted PSD inequalities. This is a
sufficient condition, not a classification of all Bell polynomials.
Its approximation step is elementary; its utility is preserving
optimal-strategy contacts and the common sum at an exact target.

\paragraph{Application to the tilted family.}
The compact finite-closure theorem is an instance of this approximation theorem.
Use the notation of \cref{sec:fixed-tilt}, put $T=X=(B_0+B_1)/2$,
$W=D(X)=(X^2-a^2I)^2$, and take the continuous separator
$\mathcal H_0=p_0(X)+2\tau W$,
$\mathcal H_1=qI-\mathcal H_0$.
Here $J_0=H(X)$ and $R_{01}=R+2\tau$ in
\eqref{eq:weighted-remainder-form}. The first two blocks are at least
$2\tau W$. For the other two, $Q_0-2\tau D\geq d/2$ and the
determinant is at least
$(\eta-4q\tau)D\geq\eta D/2$. Their trace is at most $2q$,
so their smallest eigenvalue is at least $\eta D/(4q)\geq2\tau D$.
In the two-dimensional blocks $W=D(c)I$, so this proves the required
operator inequality, including the scalar endpoints. Thus $\mu=2\tau$
and $C=L_\alpha$ recover the Bernstein degree choice in
\cref{eq:separator-Bernstein}. On any $[0,A]$, $A<2$, the remainders
are jointly continuous, $C$ is bounded, and $\mu$ stays positive.

Conversely, with $m,p,\ell,r,h,d$ fixed, diverging minimal degree
rules out a uniformly bounded ratio $C/\mu$ for representations of
the stated form. Along tilted CHSH approaching $\alpha=2$, some
regularity or weighted-margin control must therefore deteriorate.
This is a restriction on such separator descriptions, not a necessary
and sufficient geometric classification or a claim that the coarse
Bernstein bound gives the sharp exponent.

\begin{corollary}[A scalar geometric sufficient test]
\label{cor:geometric-contact-test}
Consider a compact coefficient family
$G=m_0I+m_1X+A_0(u+vX)+A_1wY$, where $v>0$,
$X=(B_0+B_1)/2$, $Y=(B_0-B_1)/2$, with attained targets $\beta$.
Suppose a continuous choice $a\in[a_{\min},1]$, $a_{\min}>0$, has
a uniform positive distance from the kink $-u/v$ to both contacts
$\pm a$. Put
\begin{equation}
 \begin{gathered}
 p_0(c)=|u+vc|,\quad D(c)=(c^2-a^2)^2,\\
 Q_0(c)=\beta-m_0-m_1c-p_0(c),\quad
 J_0(c)=Q_0(c)Q_0(-c)-w^2(1-c^2).
 \end{gathered}
\end{equation}
If fixed $\gamma,\eta>0$ satisfy
$Q_0(c)\geq\gamma D(c)$ and $J_0(c)\geq\eta D(c)$ on
$[-1,1]$ for the whole family, one finite conditioned level is exact
uniformly over the family.
\end{corollary}
\begin{proof}
Let $M>0$ uniformly bound $Q_0$, and set
$\tau=\min\{\gamma/4,\eta/(8M)\}$. The explicitly constructed
continuous separator is $\mathcal H_0=p_0(X)+2\tau D(X)$,
$\mathcal H_1=\beta-m_0-m_1X-\mathcal H_0$.
Its first two blocks are at least $2\tau D(X)$.
The remaining blocks have nonnegative diagonals, determinant at
least $\eta D/2$, and trace at most $2M$. Their smallest eigenvalue
is therefore at least $\eta D/(4M)\geq2\tau D$ away from $D=0$;
at $D=0$ the claimed order bound is just positivity.
The cubic Hermite interpolant $H$ matches $p_0,p_0'$ at $\pm a$.
Since the kink stays away from these contacts,
$(p_0-H)/D$ has a jointly continuous extension there. This follows
also by exact polynomial division on each affine branch of $p_0$;
the other contact is separated by $2a_{\min}$. The remainder is
continuous at the kink because $D$ is nonzero there.
Apply \cref{thm:contact-closure} with $T=X$, $W=D(X)$ and
$\mu=2\tau$.
\end{proof}

This test constructs the continuous separators from scalar Bell
coefficients and a proposed contact location; it does not assume
those separators as input. Its piecewise-polynomial inequalities
can be checked by exact interval positivity. It is sufficient only:
neither a kink-distance condition alone nor the absence of Alice
marginals is asserted to characterize finite closure.

\subsection{A fixed cross-question block repairs the tilted family}
\label{sec:fixed-repair}

The degree obstruction within the conditioned hierarchy can be
removed by restoring a fixed amount of cross-question information.
This follows directly from the known BP15 SOS
\cite[Eq.~(27)]{BampsPironio2015}; the observation here is its
parameter-independent Gram support and its comparison with the
conditioned degree divergence. Put
\begin{equation}
 S'=A_0(B_0-B_1)+A_1(B_0+B_1),\qquad
 \mathcal V=\operatorname{span}\{I,A_0,A_1,S,S'\}.
 \label{eq:repair-span}
\end{equation}
Write $\Sigma^2(\mathcal V)$ for sums of squares of elements of this
fixed five-dimensional span, and define the augmented certificate cone
$\mathsf R=\mathsf O_1+\Sigma^2(\mathcal V)$.

\begin{proposition}[A sufficient fixed-size repair]
\label{prop:fixed-repair}
The cone $\mathsf R$ certifies $q(\alpha)I-F_\alpha$ for every
$0\leq\alpha\leq2$. Its moment formulation adds one $5\times5$ PSD
Gram block with rows $(I,A_0,A_1,S,S')$ to the four conditioned
level-one blocks. All entries obey the common raw-PVM polynomial
identities, including identities involving different Alice questions.
The augmented relaxation certifies the complete Alice-output CHSH
guessing-probability tradeoff, including quantum side information.
\end{proposition}
\begin{proof}
The BP15 identity in the original quotient is
\begin{equation}
 qI-F_\alpha=\frac{(qI-F_\alpha)^2+(\alpha A_1-S')^2}{2q}.
 \label{eq:repair-bp15}
\end{equation}
Both factors belong to $\mathcal V$; its basis is independent of
$\alpha$. The coefficient vectors in the displayed order are
$(q,-\alpha,0,-1,0)^T$ and $(0,0,\alpha,0,-1)^T$.
Their two outer products, divided by $2q>0$, give a PSD Gram matrix.
This proves cone membership without an assumption of dual attainment.
The quantum attaining strategies are feasible for the augmented
moment constraints, so the bound is exact.

More explicitly the primal uses one Hermitian linear functional on
the span of all products defining the original four blocks and the
additional products $v_i^*v_j$. Reduction in the original quotient
enforces every overlapping entry and every linear polynomial identity.
Independent PSD matrices without these shared identities would not
justify \eqref{eq:repair-bp15}. The added rows have total word length
at most two and their products length at most four; their support
does not grow as $\alpha\uparrow2$.

Simultaneous outcome reversal acts on $\mathcal V$ by the diagonal
sign matrix $\operatorname{diag}(1,-1,-1,1,1)$ and preserves the
conditioned blocks up to relabeling. Thus the guessing-probability
reduction and tangent argument apply to the augmented relaxation.
Exactness of every tilted tangent proves the full Alice-output curve.
\end{proof}

This proposition supplies a sufficient repair for this family, not a
minimal block or a general equivalence with the standard or sequential
hierarchies. Its SOS is the previously known BP15 decomposition.
The divergent cost arises when only Bob-word length is increased
while these cross-question moments remain excluded. Determining the
smallest additional moment information, or a repair valid for a
broader family, remains open.

\subsection{Absence of Alice marginals is not a level-one criterion}
\label{sec:no-alice-counterexample}

The directional comparison does not generalize to all Bob-marginal
Bell functionals at level one. Consider
\begin{equation}
 G=B_0-3B_1+A_0(7B_0+14B_1)+A_1(11B_0-5B_1).
 \label{eq:no-alice-counterexample}
\end{equation}
\begin{proposition}
\label{prop:no-alice-counterexample}
This polynomial has quantum and local maximum 29, whereas
$\omega^{\rm os}_1(G)\geq145519/5000=29.1038$.
\end{proposition}
\begin{proof}
Jordan decomposition reduces the two pairs of reflections to real
two-dimensional blocks and their scalar limits. Set
$A_0=B_0=Z$, $A_1=uZ+\sqrt{1-u^2}X$,
$B_1=cZ+\sqrt{1-c^2}X$, with $u,c\in[-1,1]$.
For $v,w\in\mathbb R^2$ let
$\langle v,w\rangle_c=v^T\bigl(\begin{smallmatrix}1&c\\c&1\end{smallmatrix}\bigr)w$,
and put $p=(1,-3)^T$, $r=(7,14)^T$, $t=(11,-5)^T$.
The resulting Bell matrix obeys
\begin{equation}
 (G^2-dI)^2=4DI,\qquad d=238uc+14u+80c+401,
\end{equation}
where
\begin{equation}
 D=(\langle p,r\rangle_c+u\langle p,t\rangle_c)^2
 +(1-u^2)\langle p,t\rangle_c^2
 +(1-u^2)(1-c^2)\det(r,t)^2\geq0.
\end{equation}
Hence $\|G\|^2\leq d+2\sqrt D$. Since $d\leq733<841$,
it remains to verify $R=(841-d)^2-4D\geq0$ on the square.
With $u=2x-1$, $c=2y-1$, the bidegree-$(2,2)$ Bernstein coefficient
matrix of $R$ is exactly
\begin{equation}
 C_R=\begin{pmatrix}
 53760&170112&370944\\
 200704&-329520&58464\\
 548352&88128&0
 \end{pmatrix}.
\end{equation}
Subdivide each coordinate twice at its midpoint. The one-dimensional
subdivision matrices are
\begin{equation}
 M_L=\begin{pmatrix}1&0&0\\1/2&1/2&0\\1/4&1/2&1/4\end{pmatrix},\qquad
 M_R=\begin{pmatrix}1/4&1/2&1/4\\0&1/2&1/2\\0&0&1\end{pmatrix}.
\end{equation}
Replacing $C$ by the four matrices $M_aCM_b^T$, $a,b\in\{L,R\}$,
twice gives 16 coefficient matrices, all entrywise nonnegative.
Their minimum entries, in this recursive order, are
\begin{equation}
 \begin{gathered}
 53760,\ 140103/2,\ 83133,\ 43112,\ 81402,\ 130047,\ 26998,\ 46712,\\
 111168,\ 41852,\ 196785,\ 76362,\ 12214,\ 35979/2,\ 25803,\ 0.
 \end{gathered}
\end{equation}
Bernstein basis positivity proves the global bound. The deterministic
choice $A_0=A_1=B_0=B_1=-1$ attains 29.

For the conditioned witness, the four matrices on $(I,B_0,B_1)$ are
\begin{align}
 M_{0|0}&=10^{-4}\begin{pmatrix}1002&992&992\\992&1002&992\\992&992&1002\end{pmatrix},&
 M_{1|0}&=10^{-4}\begin{pmatrix}8998&-6148&-8680\\-6148&8998&4232\\-8680&4232&8998\end{pmatrix},\\
 M_{0|1}&=10^{-4}\begin{pmatrix}2430&2258&-1979\\2258&2430&-1344\\-1979&-1344&2430\end{pmatrix},&
 M_{1|1}&=10^{-4}\begin{pmatrix}7570&-7414&-5709\\-7414&7570&6568\\-5709&6568&7570\end{pmatrix}.
\end{align}
Their equal diagonals enforce the degree-two involution identities;
both outcome sums are the same matrix with unit diagonal. Each block
is positive definite by its three strictly positive leading principal
minors. Direct evaluation gives $145519/5000$.
\end{proof}

This example rules out universal level-one exactness without Alice
marginals; it does not rule out a higher uniform finite level for that
class. The separation also persists on an open coefficient neighborhood.
In the observable basis of the six terms in
\eqref{eq:no-alice-counterexample}, let $E$ have coefficient
$\ell_1$ norm $\rho$. Both quantum expectations and the rational
witness obey $|\langle E\rangle|\leq\rho$. Therefore
\begin{equation}
 \omega^{\rm os}_1(G+E)-\omega_Q(G+E)
 \geq\frac{519}{5000}-2\rho>0
 \quad\text{if }\rho<\frac{519}{10000}.
 \label{eq:no-alice-open-neighborhood}
\end{equation}
This elementary Lipschitz argument gives an open family within the
subspace without Alice marginals, not a new formula for its quantum
values or a higher-level obstruction.
The checker \codepath{scripts/verify_orientation_counterexample.py}
verifies the spectral polynomial identity, subdivision coefficients,
all rational principal minors and the witness value without a solver.
The separate \codepath{scripts/verify_orientation_structure.py} checks
the opposite-tilt certificate for symbolic $\alpha$ and the algebraic
identities used in the contact-preserving approximation argument.
It also verifies the parameter-independent repair support and the
BP15 identity coefficient-wise in the raw-PVM quotient.

\subsection{A second family: asymmetric weighting restores uniform closure}
\label{sec:asymmetric-family}

Consider the two-parameter family of Ac\'in, Massar, and Pironio
\cite{Acin2012Tilted}, also used in asymmetric CHSH randomness bounds
\cite{WoodheadAcinPironio2021},
\begin{equation}
 S_\lambda=\lambda A_0(B_0+B_1)+A_1(B_0-B_1),\qquad
 F_{\alpha,\lambda}=\alpha A_0+S_\lambda,\qquad \lambda\geq1.
 \label{eq:asymmetric-family}
\end{equation}
Its quantum maximum, for $\alpha\geq0$, is the known value
\begin{equation}
 \beta_\lambda(\alpha)=
 \begin{cases}
 q_\lambda(\alpha):=\sqrt{(\lambda^2+1)(\alpha^2+4)},
                      &0\leq\alpha\leq2/\lambda,\\
 \alpha+2\lambda,       &\alpha\geq2/\lambda.
 \end{cases}
 \label{eq:asymmetric-quantum}
\end{equation}
The local maximum is $\alpha+2\lambda$. For $\lambda>1$ this is not
a relabeling or rescaling of the $\lambda=1$ family: the absolute
correlator coefficients have ratio $\lambda:1$, invariant up to
inversion under party/question/outcome relabelings and overall scaling.
The quantum values and randomness curves are not new results here.
The following uniform conditioned-degree statement is the application
of \cref{thm:contact-closure}.

\begin{theorem}[Uniform closure under asymmetric weighting]
\label{thm:asymmetric-closure}
For every fixed $\lambda>1$ there is a finite $K_\lambda$ such that
\begin{equation}
 \beta_\lambda(\alpha)I-F_{\alpha,\lambda}\in\mathsf O_{K_\lambda}
 \qquad\text{for every }\alpha\geq0.
 \label{eq:asymmetric-uniform}
\end{equation}
One level works simultaneously for $\lambda$ in any compact interval
$[\lambda_0,\lambda_1]\subset(1,\infty)$ and all $\alpha\geq0$.
For $\lambda\geq\sqrt3$, level two already suffices, and the minimum
conditioned level is exactly two when $0<\alpha<2/\lambda$.
In contrast, the minimum such uniform level tends to infinity as
$\lambda\downarrow1$.
\end{theorem}
\begin{proof}
Fix $\lambda>1$ and initially restrict $0\leq\alpha\leq2/\lambda$.
Abbreviate $q=q_\lambda(\alpha)$ and set
\begin{equation}
 a=\frac{\lambda q}{2(\lambda^2+1)},\quad b=\frac{\alpha}{2\lambda},
 \quad D(c)=(c^2-a^2)^2,\quad p_0(c)=|\alpha+2\lambda c|,
 \quad Q_0=q-p_0.
 \label{eq:asymmetric-contacts}
\end{equation}
Here $a_{\min}:=\lambda/\sqrt{\lambda^2+1}\leq a\leq1$ and
\begin{equation}
 a-b\geq g:=1-\lambda^{-2}>0.
 \label{eq:asymmetric-contact-gap}
\end{equation}
Indeed $a-b$ decreases with $\alpha$ and equals $g$ at $2/\lambda$.
For the derivative, $a'=\lambda\alpha/(2q)<1/(2\lambda)=b'$:
$q^2-\lambda^4\alpha^2$ is affine in $\alpha^2$ and positive at
both endpoints of $[0,4/\lambda^2]$.

As before, put $X=(B_0+B_1)/2$, $Y=(B_0-B_1)/2$.
The required Bob polynomials are $P(X)\pm(\alpha I+2\lambda X)$
and $qI-P(X)\pm2Y$. In a two-dimensional block their determinant
slack at $p_0$ is, for $0\leq c\leq1$,
\begin{equation}
 J(c):=Q_0(c)Q_0(-c)-4(1-c^2)=
 \begin{cases}
 (q-\alpha)^2-4+4(1-\lambda^2)c^2,&c\leq b,\\
 4(\lambda^2+1)(c-a)^2,&c\geq b.
 \end{cases}
 \label{eq:asymmetric-slack}
\end{equation}
The inner branch decreases, with value
$4(\lambda^2+1)(a-b)^2$ at $b$. Since $D\leq1$ and
$(c+a)^2\leq4$, both branches imply
\begin{equation}
 J(c)\geq\eta D(c),\qquad
 \eta=(\lambda^2+1)\min\{1,4g^2\}>0.
 \label{eq:asymmetric-eta}
\end{equation}

The critical endpoint requires a relative, rather than absolute,
diagonal margin. Let $d=q-\alpha-2\lambda\geq0$. For
$\alpha<2/\lambda$,
\begin{equation}
 \frac{d}{(1-a^2)^2}
 =\frac{16(\lambda^2+1)^2}
 {(2+\lambda\alpha)^2(q+\alpha+2\lambda)}
 \geq\frac{\lambda(\lambda^2+1)}4.
 \label{eq:asymmetric-relative-diagonal}
\end{equation}
The resulting inequality extends to the endpoint, where both
$d$ and $1-a^2$ vanish. For $z=|c|$,
$D(c)\leq2(1-a^2)^2+8(1-z)$ and
$Q_0(c)\geq d+2\lambda(1-z)$. Since $\lambda\geq1$, these give
\begin{equation}
 Q_0(c)\geq\gamma D(c),\qquad \gamma=\lambda/4,
 \label{eq:asymmetric-gamma}
\end{equation}
uniformly in the tilt, including $\alpha=2/\lambda$.

Use the Hermite polynomial
\begin{equation}
 H(c)=\lambda\left(a+\frac{c^2}{a}\right)
       +\frac{\alpha}{2a}\left(3c-\frac{c^3}{a^2}\right).
 \label{eq:asymmetric-H}
\end{equation}
The continuous remainder $R=(p_0-H)/D$, extended at the contacts, is
\begin{equation}
 R(c)=\begin{cases}
 \dfrac{\alpha(c-2a)-2\lambda a^2}{2a^3(c-a)^2},&c\leq-b,\\[5pt]
 \dfrac{\alpha(c+2a)-2\lambda a^2}{2a^3(c+a)^2},&c\geq-b.
 \end{cases}
 \label{eq:asymmetric-remainder}
\end{equation}
The branches agree at $-b$ with value $-\lambda/a^3$.
Their denominators stay separated from zero by \cref{eq:asymmetric-contact-gap}.
Differentiation gives the uniform bound
\begin{equation}
 \operatorname{Lip}(R)\leq
 \frac{\alpha+4\lambda a}{2a^3(a-b)^2}
 \leq \overline L_\lambda:=
 \frac{2/\lambda+4\lambda}{2a_{\min}^3g^2}.
 \label{eq:asymmetric-L}
\end{equation}
Set $q_{\max}=2(\lambda+1/\lambda)$ and
\begin{equation}
 \tau_\lambda=\min\{\gamma/4,\eta/(8q_{\max})\}>0.
 \label{eq:asymmetric-tau}
\end{equation}
The continuous separators
$\mathcal H_0=p_0(X)+2\tau_\lambda D(X)$,
$\mathcal H_1=qI-\mathcal H_0$ have exact common sum $qI$.
The first two blocks are at least $2\tau_\lambda D(X)$.
For the other two, the diagonals are nonnegative by
\cref{eq:asymmetric-gamma}, and the determinant is at least
$(\eta-4q\tau_\lambda)D\geq\eta D/2$. Their trace is at most $2q$,
so their smallest eigenvalue is at least
$\eta D/(4q)\geq2\tau_\lambda D$ whenever $D>0$.
At $D=0$ positivity follows directly from the original blocks.
Thus \cref{thm:contact-closure} applies with $T=X$, $W=D(X)$,
$\mu=2\tau_\lambda$, polynomial part $H$, and remainder
$R+2\tau_\lambda$.

In particular, the Bernstein polynomial $R_n$ gives the explicit
polynomial separator $P=H+D(R_n+2\tau_\lambda)$ with the sufficient level
\begin{equation}
 n_\lambda=\max\{1,\lceil(\overline L_\lambda/\tau_\lambda)^2\rceil\},
 \qquad K_\lambda=\left\lceil\frac{n_\lambda+4}{2}\right\rceil.
 \label{eq:asymmetric-degree-bound}
\end{equation}
This conservative bound is finite and locally bounded for $\lambda>1$.
It proves an actual SOS certificate by \cref{lem:dihedral-saturation};
it is not a sampled positivity assertion or a bound on the optimal level.
All constants remain controlled on $[\lambda_0,\lambda_1]$.
For $\alpha\geq\alpha_c=2/\lambda$, use
\begin{equation}
 (\alpha+2\lambda)I-F_{\alpha,\lambda}
 =[q_\lambda(\alpha_c)I-F_{\alpha_c,\lambda}]
   +(\alpha-\alpha_c)(I-A_0),
 \label{eq:asymmetric-classical-extension}
\end{equation}
whose last term is an allowed degree-zero conditioned square.

For completeness, attainment in the nonlocal range uses
$A_0=\sigma_z$, $A_1=\sigma_x$,
$B_{0,1}=a\sigma_z\pm\sqrt{1-a^2}\sigma_x$ and
$|\psi\rangle=\cos\theta|00\rangle+\sin\theta|11\rangle$, with
\begin{equation}
 \cos2\theta=\frac{\alpha(\lambda^2+1)}q,\qquad
 \sin2\theta=\frac{2\sqrt{1-a^2}(\lambda^2+1)}q.
 \label{eq:asymmetric-attainment}
\end{equation}
Their squares sum to one and the Bell expectation is $q$.
The local value is attained by deterministic outputs.

It remains to prove the simpler level-two bound. If $\lambda\geq\sqrt3$,
take $P=H$ without approximation. Direct factorization gives
\begin{align}
 H(c)-(\alpha+2\lambda c)
 &=\frac{(c-a)^2[2\lambda a^2-\alpha(c+2a)]}{2a^3},\\
 H(c)+(\alpha+2\lambda c)
 &=\frac{(c+a)^2[2\lambda a^2-\alpha(c-2a)]}{2a^3}.
 \label{eq:asymmetric-cubic-scalar}
\end{align}
Both are nonnegative because
$\alpha(1+2a)/(2\lambda a^2)\leq3/\lambda^2\leq1$.
For the first inequality, $\alpha/a$ and $\alpha/a^2$ increase on
$0\leq\alpha\leq2/\lambda\leq2$, so their combination is largest
at the endpoint $a=1$.
The trace of the other two blocks is
$2[a^2(\lambda^2+2)-\lambda^2c^2]/(a\lambda)>0$, and their determinant is
\begin{equation}
 \frac{D(c)}{a^6\lambda^2}
 \left[a^4(\lambda^2+2)^2-4a^2\lambda^2
       +\{\lambda^2-a^2(\lambda^2+1)\}c^2\right].
 \label{eq:asymmetric-cubic-det}
\end{equation}
The coefficient of $c^2$ is $-\lambda^2\alpha^2/4\leq0$.
Writing $t=\lambda^2$, $A=a^2\geq A_*=t/(t+1)$, the bracket is bounded below by
\begin{align}
 f(A)&=A^2(t+2)^2-A(5t+1)+t\notag\\
 &=(t+2)^2(A-A_*)^2
   +\frac{2t^3+3t^2+2t-1}{t+1}(A-A_*)
   +\frac{t^4}{(t+1)^2}>0.
 \label{eq:asymmetric-cubic-positive}
\end{align}
Every coefficient on the last line is positive for $t\geq1$.
Thus all four Bob polynomials are universally PSD and have degree at
most three; \cref{lem:dihedral-saturation} puts the gap in $\mathsf O_2$.
Equation~\eqref{eq:asymmetric-classical-extension} covers larger tilts.

To exclude level one for $0<\alpha<2/\lambda$, average a hypothetical
separator pair under $B_0\leftrightarrow B_1$, $A_1\mapsto-A_1$,
a symmetry of the Bell polynomial. A swap-invariant Hermitian Bob
polynomial of word length at most two is $P(X)$ with $\deg P\leq2$:
its span is $I,B_0+B_1,B_0B_1+B_1B_0$.
At the attaining strategy both $A_0$ outcomes have positive probability,
so the contact conditions force
$P(a)=2\lambda a+\alpha$ and $P(-a)=2\lambda a-\alpha$.
Universal positivity gives $P(c)\geq|\alpha+2\lambda c|$ on $[-1,1]$.
Since $b<a<1$, tangency at the two interior contacts forces
$P'(a)=2\lambda$ and $P'(-a)=-2\lambda$.
For a quadratic polynomial these two slopes set its linear coefficient
to zero, contradicting $P(a)-P(-a)=2\alpha>0$.
Thus level one fails, proving the asserted exact level two.

Finally, fix any finite $k$ and the strict $\lambda=1$ witness at
$\alpha_k$ from \cref{thm:unbounded-sm}. The witness remains feasible
when Bell coefficients change, and both its value and
$q_\lambda(\alpha_k)$ are continuous in $\lambda$. The strict gap
therefore persists for all $\lambda>1$ sufficiently close to one;
also $\alpha_k<2/\lambda$ there. No uniform level at most $k$ can
cover those weights. This proves divergence of the minimum uniform
level as $\lambda\downarrow1$, without assuming monotonicity in $\lambda$.
\end{proof}

The geometric distinction is explicit. At the local transition,
the contact lies at $a=1$ but the kink is at $-b=-\lambda^{-2}$.
They remain separated for each $\lambda>1$, whereas at $\lambda=1$
the kink collides with the negative contact. The relative diagonal
estimate \cref{eq:asymmetric-relative-diagonal} prevents the vanishing
Bell violation alone from being mistaken for an obstruction to uniform
closure. This is a sufficient mechanism for the present family, not
a classification of all contact degeneracies in Bell optimization.

\begin{corollary}[A full asymmetric randomness curve at one finite level]
\label{cor:asymmetric-randomness}
Fix $\lambda>1$. For guessing Alice's $x=0$ output from the sole
statistic $\langle S_\lambda\rangle=s$, the same level $K_\lambda$
certifies the known optimal quantum curve
\begin{equation}
 G_{Q,\lambda}(s)=\frac{1+\sqrt{\lambda^2+1-s^2/4}}2,
 \qquad 2\lambda\leq s\leq2\sqrt{\lambda^2+1}.
 \label{eq:asymmetric-randomness}
\end{equation}
This includes quantum side information and $\langle S_\lambda\rangle\geq s$.
For $\lambda\geq\sqrt5/2$, conditioned level two suffices for the whole curve,
by the sharper threshold in \cref{thm:asymmetric-transition}.
\end{corollary}
\begin{proof}
Simultaneous outcome reversal preserves $S_\lambda$ and reverses $A_0$,
so the Eve-branch symmetrization of \cref{lem:guess-contact} applies.
All tangents $s+\alpha\langle A_0\rangle\leq q_\lambda(\alpha)$,
$0\leq\alpha\leq2/\lambda$, are exact at the same finite level.
At the attaining strategy their statistic and bias are
$s=4(\lambda^2+1)/q$ and
$\langle A_0\rangle=\alpha(\lambda^2+1)/q$; eliminating $\alpha$
gives \cref{eq:asymmetric-randomness}. The endpoint values follow by
continuity and the probability cap. The curve is decreasing, proving
the lower-statistic-constraint version as well.
\end{proof}

Equation~\eqref{eq:asymmetric-randomness} is the established guessing
bound of \cite{Acin2012Tilted,WoodheadAcinPironio2021}; the new statement
is its uniform finite-level realization in the fixed Alice-conditioned
hierarchy. Neither party nor the target output is exchanged. The Bell
statistic is changed, so this does not certify the same information
from an unchanged CHSH scalar or supply a finite-key security theorem.
The checker \codepath{scripts/verify_asymmetric_family.py} audits the
parameter-dependent identities, endpoint-relative margins, cubic
certificate, and tangent elimination exactly, without an SDP solver.

\subsection{Sharp weight dependence and the exact level-two threshold}
\label{sec:asymmetric-transition}

Define the minimum level that is exact uniformly over the tilt by
\begin{equation}
 K_*(\lambda)=\min\{k\geq1:\beta_\lambda(\alpha)I-F_{\alpha,\lambda}
                 \in\mathsf O_k\text{ for every }\alpha\geq0\}.
 \label{eq:uniform-weight-degree}
\end{equation}
This is an exact, continuum requirement, not closure to a numerical
tolerance on a finite grid. The following result sharpens both the
Bernstein degree bound and the cubic sufficient threshold above.

\begin{theorem}[Sharp weight transition]
\label{thm:asymmetric-transition}
For every $\lambda>1$, put $g=1-\lambda^{-2}$. Then
\begin{equation}
 \frac{1+\sqrt{(\lambda^2+1)/(\lambda^2-1)}}2
 \leq K_*(\lambda)\leq\frac{15}{2\sqrt g}+3.
 \label{eq:sharp-weight-bounds}
\end{equation}
Consequently $K_*(\lambda)=\Theta((\lambda-1)^{-1/2})$ as
$\lambda\downarrow1$. Moreover,
\begin{equation}
 K_*(\lambda)=2\quad\Longleftrightarrow\quad\lambda\geq\frac{\sqrt5}{2}.
 \label{eq:sharp-level-two-threshold}
\end{equation}
In this range the minimum level is two at every
$0<\alpha<2/\lambda$. The constants in the asymptotic bounds are
not asserted optimal.
\end{theorem}
\begin{proof}
\emph{A polynomial reduction at every finite level.}
Write $X=(B_0+B_1)/2$, $Y=(B_0-B_1)/2$. They satisfy
$XY=-YX$ and $Y^2=I-X^2$. Every Bob polynomial reduces to
$p(X)+r(X)Y$ without increasing word length. Averaging a separator
under the symmetry $B_0\leftrightarrow B_1$, $A_1\mapsto-A_1$
removes $r(X)Y$ while preserving the separator sum and positivity:
$\alpha I+2\lambda X$ is fixed and $2Y$ changes sign,
which only interchanges the two inequalities for the second block.
Thus an exact level-$k$ certificate gives a real
polynomial $P$ of degree at most $2k$ satisfying
\begin{equation}
 \begin{gathered}
 P(c)\geq|\alpha+2\lambda c|,\qquad Q(c):=q-P(c)\geq0,\\
 Q(c)Q(-c)\geq4(1-c^2)\qquad(-1\leq c\leq1).
 \end{gathered}
 \label{eq:asymmetric-scalar-separator}
\end{equation}
Conversely these conditions give a level-$k$ certificate by
\cref{lem:dihedral-saturation}. For $0<\alpha<2/\lambda$, the
attaining strategy and positivity force the four contact conditions
\begin{equation}
 P(\pm a)=2\lambda a\pm\alpha,\qquad P'(\pm a)=\pm2\lambda,
 \qquad a=\frac{\lambda q}{2(\lambda^2+1)}.
 \label{eq:asymmetric-all-degree-contacts}
\end{equation}
The conditional states at both contacts are nonzero in this open
interval. The derivatives follow from tangency to the respective
linear branches at interior points; no derivative condition is
assumed from the product-state strategy at the critical tilt itself.

\emph{The lower bound is a necessary endpoint derivative bound.}
Suppose level $k$ is exact for all tilts approaching $2/\lambda$
from below. Choose corresponding $P_j$. Since $0\leq P_j\leq q_j$
on $[-1,1]$ and their degrees are bounded, finite-dimensional norm
equivalence gives a coefficient-convergent subsequence, hence also
uniform convergence of the derivatives. Its limit $P$ satisfies
\cref{eq:asymmetric-scalar-separator} at $\alpha=2/\lambda$ and
inherits \cref{eq:asymmetric-all-degree-contacts} at $a=1$.
Set $E(c)=P(c)-(2/\lambda+2\lambda c)\geq0$.
For $-1<c<1$, the determinant inequality and
$Q(-c)\leq2\lambda(1+c)$ imply
$Q(c)\geq2(1-c)/\lambda$. Therefore, by continuity at the endpoints,
\begin{equation}
 0\leq E(c)\leq2\lambda g(1-c),\quad
 E(-1)=4\lambda g,\quad E'(-1)=-4\lambda.
 \label{eq:markov-weight-envelope}
\end{equation}
Because $E(1)=0$,
$T(c)=E(c)/[2\lambda g(1-c)]$ is a polynomial of degree at most
$2k-1$, with its removable value at 1. It obeys
\begin{equation}
 0\leq T\leq1,\qquad T(-1)=1,\qquad T'(-1)=\frac12-\frac1g.
\end{equation}
Markov's inequality applied to $T-1/2$ gives
$1/g-1/2\leq(2k-1)^2/2$, proving the lower bound in
\cref{eq:sharp-weight-bounds}. In particular $k=2$ requires
$\lambda^2\geq5/4$. This argument concerns a common level for an
interval of tilts; exactness at the critical tilt alone does not
impose the limiting derivative conditions.

\emph{An explicit quartic proves the sharp level-two threshold.}
For $0\leq\alpha\leq2/\lambda$, put
\begin{equation}
 t=\lambda^2,\quad \rho=\frac{\alpha}{2\lambda a},\quad
 z=\frac ca,\quad r=\frac{\rho^2}{4(1-\rho)},\quad
 P(c)=\lambda a\bigl[1+z^2+\rho(3z-z^3)+r(z^2-1)^2\bigr].
 \label{eq:asymmetric-quartic}
\end{equation}
Here $0\leq\rho\leq1/t<1$ and
$a^{-2}=(t+1-t^2\rho^2)/t$. Denote the expression in brackets
by $p(z)$. The first two positivity conditions factor as
\begin{equation}
 \begin{split}
 p(z)-2(\rho+z)&=(z-1)^2[1-\rho(z+2)+r(z+1)^2],\\
 p(z)+2(\rho+z)&=(z+1)^2[1-\rho(z-2)+r(z-1)^2].
 \end{split}
 \label{eq:quartic-scalar-factors}
\end{equation}
For $\rho>0$, the two quadratic brackets have minima 0 and $2\rho$,
respectively. For $\rho=0$ both equal 1. The determinant, divided by
$\lambda^2a^2(z^2-1)^2$, is the polynomial
\begin{equation}
 C_2 z^4+C_1z^2+C_0,\qquad
 \begin{aligned}
 C_2&=r^2,\quad C_1=2r-2r^2-\rho^2,\\
 C_0&=r^2-2r-4r/t+4\rho^2+1.
 \end{aligned}
\end{equation}
Its discriminant margin is exactly
\begin{equation}
 4C_2C_0-C_1^2
 =\frac{\rho^5[t\rho^2-(5t+1)\rho+4t]}{4t(1-\rho)^3}.
 \label{eq:quartic-discriminant}
\end{equation}
For $t\geq5/4$ and $0\leq\rho\leq1/t$, the bracket divided by $t$
is decreasing in $\rho$ and at $\rho=1/t$ is $4-5/t\geq0$.
Thus the determinant is nonnegative for all real $z$ (complete the
square when $\rho>0$; for $\rho=0$ it is 1).
Also $r\leq4/5$ and
$Q(0)=\lambda a(1+2/t-r)>0$. The determinant condition makes
$Q(c)Q(-c)>0$ for $|c|<1$, so continuity gives $Q\geq0$
on the whole interval. This proves the quartet of positive Bob
polynomials of degree at most four. The classical extension above
$2/\lambda$ uses \cref{eq:asymmetric-classical-extension}.
Level-one impossibility in the preceding proof and the Markov lower
bound establish \cref{eq:sharp-level-two-threshold}.

The same quartic works whenever $t\geq1$ and $\rho\leq3/4$:
the bracket divided by $t$ in \cref{eq:quartic-discriminant} is then
at least $\rho^2-6\rho+4\geq1/16$. This covers a common compact
tilt region, leaving only $\sigma:=1-\rho\in(0,1/4]$.

\emph{A uniform upper bound from an exterior pole.}
In the remaining region retain $a,z$ and set
\begin{equation}
 \gamma=\frac{a(1+\sigma)}{1-\sigma},\quad d=\gamma-1,\quad
 v(z)=(1-\sigma)z+1+\sigma,
 \quad P_{\rm r}(c)=\lambda a\left[2(1-\sigma+z)
                   +\frac{4\sigma^3(z-1)^2}{v(z)^2}\right].
 \label{eq:weighted-rational-separator}
\end{equation}
Since $a^{-2}=1+1/t-t(1-\sigma)^2\leq1+2\sigma-\sigma^2$
and $a\leq1$, the estimates from \cref{sec:matching} give
\begin{equation}
 1-\sigma\leq a\leq1,\qquad
 \sigma\leq d\leq\frac{2\sigma}{1-\sigma}
          \leq\frac83\sigma\leq\frac23.
 \label{eq:weighted-pole-distance}
\end{equation}
The scalar branch factors in \cref{eq:rational-branch-positivity}
are simply multiplied by $\lambda$, so remain nonnegative:
$a^{-1}\leq\sqrt{1+2\sigma-\sigma^2}\leq1+\sigma$ gives
the same spectral lower bound $z\geq-(1+\sigma)$ used there.
For $D=(c^2-a^2)^2$, let $J_{\rm r}=Q_{\rm r}(c)Q_{\rm r}(-c)-4(1-c^2)$.
Its exact normalized factorization is
\begin{equation}
 \begin{gathered}
 \frac{J_{\rm r}v(z)^2v(-z)^2}{\lambda^2a^2(z^2-1)^2}
 =8\sigma^2 L+\frac{4(t-1)}t(1-\sigma)^2
       [1+4\sigma+5\sigma^2+2\sigma^3-(1-\sigma)^2z^2],
 \\ L=1+2\sigma-2\sigma^3+\sigma^4.
 \end{gathered}
 \label{eq:weighted-rational-determinant}
\end{equation}
This is a polynomial identity after clearing denominators, including
the contacts. On the spectral interval $z^2\leq a^{-2}$, the square
bracket is at least $4\sigma+9\sigma^2-2\sigma^3+\sigma^4>0$.
Consequently the previous determinant lower bound holds with an
additional factor $\lambda^2$. At $c=0$,
$Q_{\rm r}(0)/(\lambda a)=2/t+2\sigma-4\sigma^3/(1+\sigma)^2>0$.
The positive product in the open interval again implies $Q_{\rm r}\geq0$,
including the critical-tilt endpoint where a diagonal vanishes.

For the Hermite polynomial $H$ in \cref{eq:asymmetric-H}, write
\begin{equation}
 \frac{P_{\rm r}-H}{D}=\frac A{c+\gamma}+\frac B{(c+\gamma)^2},
 \qquad A=\frac{\lambda(1-\sigma)}{a^2},\quad
 B=\frac{\lambda\sigma}{a}.
\end{equation}
Here $A\leq2\lambda$, $B/d\leq2\lambda$.
Moreover $Q_{\rm r}(-c)\leq q-\alpha+2\lambda c\leq2\lambda(c+\gamma)$.
The same determinant-loss argument as in
\cref{eq:rational-error-budget} therefore shows that the sufficient
one-sided approximation budget is
\begin{equation}
 0\leq\Pi(c)-\frac A{c+\gamma}-\frac B{(c+\gamma)^2}
       \leq\frac{\lambda\sigma^2}{2(c+\gamma)^3}.
 \label{eq:weighted-error-budget}
\end{equation}
At a vanishing endpoint diagonal, the error $D(\Pi-R)$ vanishes
exactly; the inequalities extend by continuity, without division by
that zero diagonal.
More explicitly, let $E=D(\Pi-R)$. On the open interval the budget
implies $0\leq E(c)\leq J_{\rm r}(c)/(2Q_{\rm r}(-c))
\leq Q_{\rm r}(c)/2$. Hence the new diagonals are nonnegative and
\begin{equation}
 \begin{split}
 J(c)&=J_{\rm r}(c)-E(c)Q_{\rm r}(-c)-E(-c)Q_{\rm r}(c)
       +E(c)E(-c)\\
 &\geq E(c)E(-c)\geq0.
 \end{split}
\end{equation}
All quantities here are continuous on the closed interval since
$\gamma>1$, so these undivided inequalities also cover its endpoints.

Use the kernel and double-pole cancellation of
\cref{eq:matching-kernel,eq:matching-polynomial}, now with
$N=\lceil5/\sqrt d\rceil$, coefficients $A,B$, and
$C=A-B\mathcal K'(-\gamma)$, $M=B/d^2+C/d$. Since $d\leq2/3$,
$h=\operatorname{arcosh}(1+d)\geq(4/3)\sqrt d$ and
$5\leq N\sqrt d<6$. The exact estimates are
\begin{equation}
 \begin{gathered}
 f_N(-\gamma)\geq\sum_{j=1}^5
    \frac{(4/3)^{2j}5^{2j-2}}{(2j)!}
    =\frac{496881496}{33480783}>14,\\
 0\leq\mathcal K(c)\leq\left(\frac{27d}{350(c+\gamma)}\right)^3,
 \qquad 0<-d\mathcal K'(-\gamma)<11,\qquad dM<26\lambda.
 \end{gathered}
\end{equation}
For the derivative bound use
$3[6\cdot33/(32\sqrt2)-1]<11$ in the logarithmic derivative formula.
The numerator defining $\Pi$ has a double zero at $-\gamma$, so
$\deg\Pi\leq3(N-1)$. Its nonnegative approximation error is at most
\begin{equation}
 \frac{26\lambda d^2}{(c+\gamma)^3}\left(\frac{27}{350}\right)^3
 \leq\frac{454896}{5359375}
             \frac{\lambda\sigma^2}{(c+\gamma)^3}
 <\frac{\lambda\sigma^2}{2(c+\gamma)^3}.
\end{equation}
Thus $P=H+D\Pi$ proves exactness at
$k\leq15/(2\sqrt d)+3\leq15/(2\sqrt\sigma)+3$.
Finally $\sigma=1-\rho\geq1-\lambda^{-2}=g$ throughout the tilt
interval. Together with the quartic on $\rho\leq3/4$, this proves
the uniform upper bound. Since $g\sim2(\lambda-1)$, the upper
and lower bounds have the same exponent.
\end{proof}

The transition is governed by a dimensionless contact separation:
$g$ is exactly the distance between the kink and the negative contact
at the critical tilt. The pole construction costs $O(g^{-1/2})$;
the limiting contacts require $\Omega(g^{-1/2})$ by Markov's
inequality. At $\lambda=1$, the same approximation mechanism resolves
a separation proportional to $2-\alpha$. This unifies the two
limits within this family without asserting a classification of
arbitrary Bell functionals.
The exact lower bound implies $K_*(1.02)\geq5$,
$K_*(1.01)\geq6$, and $K_*(1.005)\geq8$. Apparent closure at
smaller levels on a coarse tilt grid and a fixed value tolerance
does not contradict these statements. The checker
\codepath{scripts/verify_asymmetric_transition.py} audits the
identities and constants. Independent exact Sturm checks of the four
Bob blocks cover the level-two threshold and complete pole-cancelled
polynomials at both interior and critical tilts. These finite audits
supplement, and do not replace, the uniform proof above.

% Intended for possible insertion into the Supplemental Material after review.
% This file is not part of the published v0.2.0 or the prepared arXiv package.
% It uses the manuscript's amsthm environments and existing labels.

\section{Operational consequence for device-independent randomness}
\label{sec:randomness}

The quantum CHSH guessing-probability tradeoff used here is known
\cite{MasanesPironioAcin2011}; the new conclusion is its exact native-level
certification cost. We specify the side-information model and prove
the transfer from Bell-bound separation to randomness-bound separation.
The latter needs a convexity argument: a witness exceeding a tilted
bound need not itself have the required CHSH expectation.
The target throughout this section is Alice's output at $x=0$.
For Bob's output the same conditioned hierarchy is already exact at
level one; see \cref{prop:opposite-tilt}.

\subsection{Quantum side information and finite-level relaxations}
Write $S(p)=\langle F_0\rangle_p$ and $b(p)=\langle A_0\rangle_p$.
For a subnormalized behavior, let $n(p)$ denote its normalization.
Let $\mathcal K$ be a compact convex normalized moment body, projected
onto behavior coordinates, and let $\operatorname{cone}(\mathcal K)$
be its subnormalized cone, including zero. Define
\begin{equation}
 \begin{split}
 G_{\mathcal K}(s)=\max_{p^+,p^-\in\operatorname{cone}(\mathcal K)}\quad&
 \frac{n(p^+)+b(p^+)+n(p^-)-b(p^-)}{2},\\
 \text{subject to}\quad&n(p^+)+n(p^-)=1,\\
 &S(p^+)+S(p^-)=s.
 \end{split}
 \label{eq:guess-conic}
\end{equation}
For $\mathcal K=\mathcal Q$, this is the device-independent probability
of guessing Alice's outcome at $x=0$ from knowledge of $S=s$ alone.
Indeed, any measurement of Eve's quantum system followed by a binary
guess prepares two subnormalized quantum behaviors on Alice and Bob.
Conversely, any such pair is realizable by a direct-sum construction
with a classical flag available to Eve. This is the usual conic
description of guessing probability \cite{NietoSillerasPironioSilman2014},
here with one observed statistic instead of a fixed complete behavior.
Replacing $\mathcal Q$ by either native-level moment body gives a sound
upper bound against arbitrary quantum side information. No assumption
that Eve's initial information is classical is being made.

The familiar quantum value is
\begin{equation}
 G_Q(s)=\frac{1+h_Q(s)}2,\qquad
 h_Q(s)=\sqrt{2-s^2/4},\qquad 2\leq s\leq2\sqrt2.
 \label{eq:guess-quantum}
\end{equation}
One can also derive this value directly from the standard tilted SOS
already checked in \cref{lem:standard-tilted}. For $e=\pm1$, relabeling
outcomes gives $S(p^e)+e\alpha b(p^e)\leq q(\alpha)n(p^e)$.
Summing gives
$s+\alpha(2G-1)\leq q(\alpha)$. Minimizing this upper bound at
\begin{equation}
 s_\alpha=\frac8{q(\alpha)},\qquad
 h_Q(s_\alpha)=\frac{2\alpha}{q(\alpha)}
 \label{eq:guess-param}
\end{equation}
gives \cref{eq:guess-quantum}. The attaining two-qubit strategy with
trivial Eve attains equality. This recovers a known tradeoff, rather
than establishing a new entropy formula.

\subsection{A contact criterion for exact randomness certification}
Let $R$ simultaneously reverse both outcomes for both questions on
both parties. Thus $S(Rp)=S(p)$ and $b(Rp)=-b(p)$.
Both finite hierarchies are invariant under $R$: a complemented
projector is $I-P$, so its word expansion does not increase the native
level, and the moment transformation is a congruence. In the conditioned
case Alice's blocks are also permuted. This argument applies to the raw
quotient without adding constraints.

\begin{lemma}[From a supporting Bell bound to guessing probability]
\label{lem:guess-contact}
Suppose $\mathcal K$ is compact, convex, contains the quantum behavior
body, and is invariant under $R$. Set
$h_{\mathcal K}(s)=\max\{b(p):p\in\mathcal K,\ S(p)=s\}$.
Then $G_{\mathcal K}(s)=(1+h_{\mathcal K}(s))/2$. Moreover, for
$0<\alpha<2$,
\begin{equation}
 G_{\mathcal K}(s_\alpha)=G_Q(s_\alpha)
 \quad\Longleftrightarrow\quad
 \max_{p\in\mathcal K}\{S(p)+\alpha b(p)\}=q(\alpha).
 \label{eq:guess-iff}
\end{equation}
\end{lemma}
\begin{proof}
For any pair in \cref{eq:guess-conic},
$r=p^++Rp^-\in\mathcal K$ has normalization one, CHSH expectation $s$,
and bias $b(p^+)-b(p^-)$. Conversely, $p^+=r$, $p^-=0$ realizes any
normalized feasible $r$. This proves the first identity.

The function $h_{\mathcal K}$ is finite and concave on its CHSH domain,
and $h_{\mathcal K}\geq h_Q$ in a neighborhood of $s_\alpha$.
If the Bell bound is exact, it gives
$h_{\mathcal K}(s_\alpha)\leq[q(\alpha)-s_\alpha]/\alpha
=h_Q(s_\alpha)$, proving contact.
Conversely, suppose contact holds. Since $s_\alpha$ is an interior
point of the domain, a concave supergradient exists there. For any
such supergradient $m$, the inequalities
\begin{equation}
 h_Q(s)\leq h_{\mathcal K}(s)
 \leq h_Q(s_\alpha)+m(s-s_\alpha)
\end{equation}
near $s_\alpha$ force $m=h_Q'(s_\alpha)=-1/\alpha$, by taking
difference quotients from both sides. Its supporting line bounds
$h_{\mathcal K}$ on the entire domain. Consequently,
$s+\alpha h_{\mathcal K}(s)\leq q(\alpha)$ everywhere. The quantum
strategy attains the bound, proving the converse.
\end{proof}

This is an application of concave duality, with a symmetry that includes
the guessing branches. A failure of one chosen dual certificate alone
would not establish the converse or the strict loss.

\begin{theorem}[Sharp hierarchy cost for randomness]
\label{thm:guess-cost}
Let $G_k^{\rm std}$ and $G_k^{\rm os}$ be \cref{eq:guess-conic}
for the two native-level moment bodies. For
$s\in I_S=[8/q(3/2),8/q(13/10)]$,
\begin{equation}
 G_2^{\rm std}(s)=G_3^{\rm os}(s)=G_Q(s)<G_2^{\rm os}(s).
 \label{eq:guess-sharp}
\end{equation}
The minimum exact levels are two and three, respectively. The same
equalities and strict inequality hold if the observed constraint is
replaced by $S\geq s$.
\end{theorem}
\begin{proof}
Apply \cref{lem:guess-contact} to the continuous exact-closure and
nonclosure results in \cref{thm:L2-threshold,sec:interval}.
Strictness follows by contraposition of the contact criterion, since
every relaxed value is at least the quantum value. For minimality,
$\mathsf D_1\subseteq\mathsf O_1$ and nesting imply that the standard
level-one moment body contains the conditioned level-two body.
Its guessing bound is therefore at least $G_2^{\rm os}>G_Q$;
conditioned level one also cannot be exact. For the variant $S\geq s$,
the same tangent bound gives an upper bound no larger than $G_Q(s)$
for standard level two and conditioned level three. The attaining
quantum strategy has $S=s$. The level-two separating point at exact
$S=s$ remains feasible for the weaker observed constraint.
\end{proof}

\subsection{An exact quantitative witness}
At $s_0=23/10$, $h_Q(s_0)=\sqrt{271}/20$ and
$G_Q(s_0)=(20+\sqrt{271})/40$.
Choose the rational level-two witness $W$ at tilt $11/8$ from
\datapath{data/contact_geometry/tangent_fan.json}, and write $g=S(W)$, $h=b(W)$.
Choose the physical quantum strategy from \cref{sec:interval} at
$u=111/25$, with CHSH value $s_q=8/q(u)$ and bias
$h_q=2\alpha(u)/q(u)$. These are rational. Its level-two moment blocks
$Q$ are also rational, as verified by constructing them from the
physical two-qubit density matrix and Bob projectors. Set
\begin{equation}
 \lambda=\frac{s_q-s_0}{s_q-g},\qquad
 M=\lambda W+(1-\lambda)Q.
 \label{eq:guess-mixture}
\end{equation}
The exact weight obeys $0<\lambda<1$ (approximately $0.678913$), so
all four blocks of $M$ are feasible. They give $S(M)=s_0$ and
\begin{equation}
 b(M)=
 \frac{1602301208219020675967821740853154900000130939255720797}
 {1943603588252961268912416887709717200792379154154268740}.
 \label{eq:guess-rational-bias}
\end{equation}
The independent verifier checks normalization, every moment equality,
outcome-sum consistency and exact PSD of $W,Q,M$, and of the relabeled
$M$. It checks both the scalar mixture relation and its full matrix
entries. The rational inequalities
\begin{equation}
 b(M)-\frac{3}{2500}>0,\qquad
 \left(b(M)-\frac{3}{2500}\right)^2>\frac{271}{400}
\end{equation}
prove
$G_2^{\rm os}(s_0)-G_Q(s_0)>3/5000$ without a floating-point comparison.
For the entropy deficit, put $r=10007/10000$. Exact squaring checks
$[1+b(M)]/[1+\sqrt{271}/20]>r$, and the integer-power comparison
$10007^{1000}>2\cdot10000^{1000}$ gives
\begin{equation}
 \log_2\frac{G_2^{\rm os}(s_0)}{G_Q(s_0)}>
 \log_2 r>\frac1{1000}\ \text{bits}.
\end{equation}
The explicit witness yields approximately $0.00102303$ bits of deficit;
this decimal is illustrative, and is not the optimum of the SDP.

The artifact is \datapath{data/analytic/randomness/cost.json}. Its verifier
\codepath{scripts/randomness_cost.py} runs under the common solver-free
entry point and rejects corrupted moments, a changed CHSH statistic,
and an inflated gap. The optional numerical illustration
\codepath{scripts/run_randomness_curve.py} solves the two subnormalized
guessing branches directly and cross-checks the symmetry-reduced
one-branch program; it is not used in the proof. Only the CHSH
statistic is constrained in this task. Complete observed behaviors
can give stronger bounds and are not covered by the separation claim.
No finite-key, multi-round, or compiled-protocol security conclusion
is inferred from this single-round certification deficit.

\begin{figure}[htbp]
 \centering
 \includegraphics[width=0.66\textwidth]{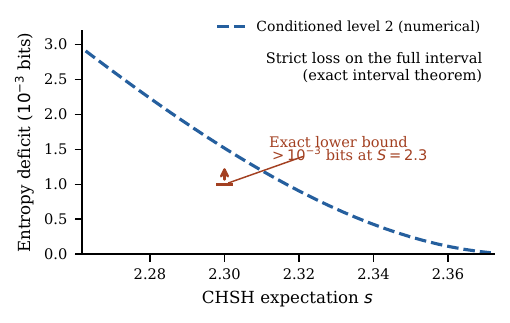}
 \caption{Finite-level min-entropy certification deficit
 $\log_2[G_2^{\rm os}(s)/G_Q(s)]$. The dashed curve is numerical;
 strict positivity throughout the interval follows from
 \cref{thm:guess-cost}. The arrow at $s=2.3$ marks the exactly proved
 lower bound $10^{-3}$ bits, not the optimum of the relaxation.
 Standard level two and conditioned level three certify the exact
 quantum tradeoff throughout this interval.}
 \label{fig:randomness-deficit}
\end{figure}

\section{Mixed-direction certificates}
\label{sec:mixed-group}
These results delimit the directional obstruction. A common truncated
functional is needed to exclude both orientations in the same proof;
the exact finite separation and its extension then follow.
\subsection{Compatibility of the two conditioning directions}
\label{sec:mixed-directions-research}

Write $\mathsf M_k=\mathsf O^A_k+\mathsf O^B_k$ for sums of nice-SOS
polynomials that may use either conditioning direction, within one fixed
Bell scenario. This differs from choosing one direction for the entire
certificate. In particular, a gap outside
$\mathsf O^A_k\cup\mathsf O^B_k$ need not be outside $\mathsf M_k$.
The weakly coupled four-question construction illustrates this distinction:
its gap lies in $\mathsf M_1$ even though its individual conditioned
degrees diverge.

\begin{lemma}[Compatible truncated functionals test the mixed cone]
\label{lem:mixed-gluing}
Suppose a normalized Alice-conditioned level-$k$ assemblage and a
normalized Bob-conditioned level-$k$ assemblage induce the same
marginals and cross-party correlators. Their common Bell behavior
$\ell$ is nonnegative on every Bell polynomial in $\mathsf M_k$.
Consequently, $\ell(F)>\eta$ proves
$\eta I-F\notin\mathsf M_k$.
\end{lemma}
\begin{proof}
In the reduced reflection-word algebra with only cross-party commutation,
the support space of an Alice-conditioned square has at most one Alice
letter and at most $2k$ Bob letters. The Bob-conditioned support space
has the reverse restrictions. Their intersection is precisely the
Bell subspace spanned by $I$, $A_x$, $B_y$, and $A_xB_y$.
The two normalized functionals agree on this intersection, and therefore
extend to a well-defined linear functional on the sum of the two support
spaces. Each assemblage is nonnegative on its own nice-SOS cone, so the
glued functional is nonnegative on $\mathsf M_k$. Its value on
$\eta I-F$ is $\eta-\ell(F)$, proving the assertion by weak duality.
The construction uses all outcome blocks and all native words; no
compatibility of unmeasured higher moments across the two orientations
is imposed or needed, since those words are outside the intersection.
\end{proof}

\subsection{Exact standard degree three and mixed degree eight}
\label{sec:exact-mixed-eight}
Define
\begin{align}
 J_0={}&a(A_0+B_0)-b(A_1+B_1)+c A_0B_0
       +d(A_0B_1+A_1B_0)-e A_1B_1,\\
 (a,b,c,d,e)={}&\frac1{10^8}
 (31096805,31141733,33401776,35456380,4405679),\\
 \beta={}&2a+c+e=\frac{20000213}{20000000}.
 \label{eq:mixed-eight-base}
\end{align}
For $0\le t\le 10^{-5}$, set
\begin{equation}
 J_t=J_0+t(A_0+B_0),\qquad \beta_t=\beta+2t.
 \label{eq:mixed-eight-interval}
\end{equation}

\begin{theorem}[Exact finite separation]
\label{thm:mixed-eight}
For every $t\in[0,10^{-5}]$,
\begin{equation}
 \omega_Q(J_t)=\omega_L(J_t)=\beta_t,
 \qquad \beta_t I-J_t\in\mathsf D_3\setminus\mathsf M_7.
 \label{eq:mixed-eight-membership}
\end{equation}
More precisely, throughout this interval,
\begin{equation}
 d_{\rm std}(J_t)=3,\qquad d_{\rm mix}(J_t)=8.
 \label{eq:mixed-eight-degrees}
\end{equation}
The Bell coefficients are exactly invariant under exchange of Alice
and Bob. For every $t\ge0$, the same ray has an exact attained quantum
target $\beta_t$ and satisfies $d_{\rm mix}(J_t)\le8$.
\end{theorem}

\begin{proof}
The finite rational data are in
\datapath{research/mixed_divergence/D3_not_M7_exact.json}. We describe the complete algebraic
verification, including the degree convention and the exact target.

Let $W_3$ be all 25 distinct reduced words of total length at most three
in the four reflections. Evaluate these words at the two characters
\begin{equation}
 (A_0,A_1,B_0,B_1)=(1,-1,1,1),\qquad(1,1,1,-1).
 \label{eq:mixed-eight-characters}
\end{equation}
The resulting $2\times25$ rational matrix has a 23-dimensional kernel.
The data specify 23 polynomial factors $T_i$ spanning that kernel and a
rational symmetric $23\times23$ matrix $Q$ such that
\begin{equation}
 \beta I-J_0=T^*QT,\qquad Q\succ0.
 \label{eq:mixed-eight-standard}
\end{equation}
Every coefficient in \eqref{eq:mixed-eight-standard} is matched in the full
reduced noncommutative algebra, without identifying a word with its
adjoint during verification. All 23 exact LDL pivots of $Q$ are positive.
The factors have degree at most three, proving membership in $\mathsf D_3$.
Both assignments in \eqref{eq:mixed-eight-characters} attain $\beta$; enumeration
of all 16 deterministic assignments also gives local maximum $\beta$.
Consequently this is the exact quantum target, not an SDP upper bound
used as a substitute for the quantum value.

For the lower bound the data supply two full assemblages, one for each
orientation, with identical eight Bell moments. Each has four outcome
blocks on the complete 15-word ball in two reflections of length at most
seven. Its Gram entries are values of $u^*v$, using every required
moment through length fourteen. In both orientations the two outcome
sums agree at every moment and have total mass one. Each of the eight
$15\times15$ Grams is positive definite: all 120 exact LDL pivots are
positive.

Let $\ell$ be the resulting common Bell functional. These checks imply
that it is nonnegative on $\mathsf M_7$. Indeed, the support space of an
Alice-conditioned square consists of words with Alice length at most
one and Bob length at most fourteen; the reversed support space
intersects it exactly in
\[
 \mathcal B=\operatorname{span}_{\mathbb R}
 \{I,A_0,A_1,B_0,B_1,A_xB_y\}.
\]
Agreement on these moments therefore defines a linear functional on
the sum of the two support spaces, positive on both conditioned cones.
The raw PVM basis is related to the reflection basis by an invertible
triangular matrix with diagonal entries $2^{-|w|}$, so no projection
or change of filtration enters the positivity argument.

The exact witness excess and its response to \eqref{eq:mixed-eight-interval} are
\begin{align}
 \Delta&=\ell(J_0)-\beta
 =\frac{593429101169281}{50000000000000000000}>10^{-5},\\
 s&=2-\ell(A_0)-\ell(B_0)=\frac{20024920147}{62500000000}>0,\\
 \ell(J_t)-\beta_t&=\Delta-st
 \ge\frac{433229739993281}{50000000000000000000}
 >8\times10^{-6}.
 \label{eq:mixed-eight-margin}
\end{align}
Hence $\beta_t I-J_t\notin\mathsf M_7$. Finally,
\[
 \beta_t I-J_t=(\beta I-J_0)
 +t[(I-A_0)+(I-B_0)]
\]
is a standard degree-three certificate for every $t\ge0$; its added
terms equal $(t/2)(I-A_0)^2$ and $(t/2)(I-B_0)^2$.
Both strategies \eqref{eq:mixed-eight-characters} continue to attain the target.
This proves \eqref{eq:mixed-eight-membership} on the entire stated interval.

To prove the matching mixed upper bound, the additional rational data
\datapath{research/mixed_divergence/M8_upper_exact.json} specify four Bob-polynomial vectors
$v_{a|x}$ of degree at most eight and four positive definite Gram matrices
$H_{a|x}$ of sizes $15,17,16,16$. Define
\begin{equation}
 R_A=\sum_{x,a}P^A_{a|x}\,v_{a|x}(B)^*H_{a|x}v_{a|x}(B).
\end{equation}
Writing $\sigma$ for the exchange of parties, the verified identity is
\begin{equation}
 \beta I-J_0=R_A+\sigma(R_A)\in\mathsf M_8.
 \label{eq:mixed-eight-upper}
\end{equation}
The factor spaces impose the deterministic zeros in
\eqref{eq:mixed-eight-characters}: the $A_0=+1$ block vanishes at both Bob
characters $(+1,+1)$ and $(+1,-1)$, the $A_0=-1$ block is unrestricted,
and the two $A_1$ blocks vanish at their respective attaining characters.
The verifier checks the complete coefficient identity in the reduced
algebra and all 64 positive rational LDL pivots. Party exchange supplies
the other four outcome blocks; it introduces no additional Gram data.
Adding $t[(I-A_0)+(I-B_0)]\in\mathsf M_1$ extends
\eqref{eq:mixed-eight-upper} to every $t\ge0$. Together with the strict
level-seven witness, this proves $d_{\rm mix}=8$ on the stated interval.

Finally, \datapath{research/mixed_divergence/D2_lower_exact.json} gives a normalized standard
level-two moment functional $m$ on all 13 reduced words of total degree
at most two. Its $13\times13$ Gram has 13 strictly positive rational
LDL pivots. Direct evaluation gives
\begin{align}
 m(J_0)-\beta&=\frac{306975015009709929}{200000000000000000000},\\
 m(J_t)-\beta_t&\ge
 \frac{305924556676806129}{200000000000000000000}>0.0015
 \quad(0\le t\le10^{-5}).
\end{align}
This excludes $\mathsf D_2$ throughout the interval. The already verified
degree-three certificate proves $d_{\rm std}=3$ and completes
\eqref{eq:mixed-eight-degrees}.
\end{proof}

The independent verifier
\codepath{research/mixed_divergence/verify_degree_pair.py}
checks the degree-three identity and the common level-seven witness in
standard-library rational arithmetic. Its mutation checks reject
corrupted Bell coefficients, negative Gram pivots and inconsistent
shared moments. The checkers
\codepath{research/mixed_divergence/verify_mixed_upper.py} and
\codepath{research/mixed_divergence/verify_standard_lower.py} establish
the remaining two bounds. Numerical optimisation discovers candidates;
the rational identities and positivity checks establish the result.

The degree-eight upper bound holds on the entire nonnegative parameter
ray. Increasing $t$ outside the certified lower-bound interval therefore
cannot produce mixed-degree divergence. The finite separation does not
establish an open region of coefficient perturbations.

\begin{figure}[t]
\centering
\includegraphics[width=\textwidth]{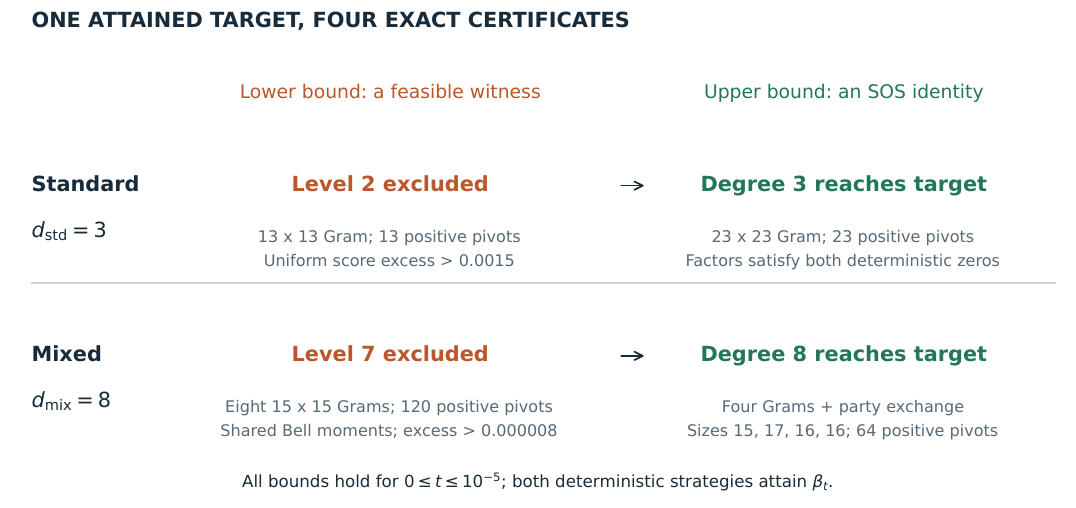}
\caption{Four exact certificates determining $(d_{\rm std},d_{\rm mix})=(3,8)$.
Lower witnesses exclude standard level two and mixed level seven;
upper identities attain the same target at degrees three and eight.
Every stated interval margin and Gram pivot is checked in rational
arithmetic. The result is a finite comparison of proof systems at a
locally attained quantum target.}
\label{fig:mixed-certificates}
\end{figure}

\subsection{A nonlocal extension that preserves certificate degrees}
\label{sec:degree-preserving-extension}
The equality $\omega_Q=\omega_L$ in Theorem~\ref{thm:mixed-eight} is a real
limitation if it is presented as a quantum advantage experiment. It
does not obstruct studying the certificate separation. The following
reduction supplies a nonlocal counterpart without changing the degree
problem. The nonlocality is supplied by an independent CHSH summand;
this is not a new irreducible physical mechanism.

\begin{proposition}[Degree-preserving addition of CHSH]
\label{prop:degree-preserving-chsh}
Let $G$ use a finite set of binary questions, have an attained quantum
value $\beta$, and let $C=A_2(B_2+B_3)+A_3(B_2-B_3)$ use disjoint
question labels. For any $\eta>0$, set
\[
 H=G+\eta C,\qquad \gamma=\beta+2\sqrt2\,\eta.
\]
Then $\omega_Q(H)=\gamma$ and $\omega_L(H)=\omega_L(G)+2\eta$.
For every $k\ge1$,
\begin{align}
 \gamma I-H\in\mathsf M_k&\ \Longleftrightarrow\ \beta I-G\in\mathsf M_k,\\
 \gamma I-H\in\mathsf D_k&\ \Longleftrightarrow\ \beta I-G\in\mathsf D_k.
 \label{eq:degree-preserving-profile}
\end{align}
The same equivalence holds separately for $\mathsf O^A_k$ and $\mathsf O^B_k$.
For a two-question $G$, the enlarged scenario has four questions per
party and remains fixed as parameters vary.
\end{proposition}

\begin{proof}
The standard CHSH identity is
\[
 2\sqrt2I-C=\frac1{\sqrt2}\left[
 \left(A_2-\frac{B_2+B_3}{\sqrt2}\right)^2+
 \left(A_3-\frac{B_2-B_3}{\sqrt2}\right)^2\right].
\]
Each factor uses only one Alice question, with Bob degree one. Resolving
that Alice reflection into its two projectors puts the gap in $\mathsf O^A_1$;
the transposed identity puts it in $\mathsf O^B_1$. It is also in $\mathsf D_1$.
Adding certificates proves the forward constructions from a certificate
of $G$. The sum of quantum upper bounds is attained by tensoring an
attaining strategy of $G$ with the optimal CHSH strategy. Independent
deterministic assignments similarly give the stated local value.

For the converse, map the enlarged universal algebra into the old
algebra tensored with two auxiliary qubit matrix algebras. Keep old
generators on the old factor and substitute optimal CHSH matrices for
the new generators on the auxiliary factors. Take expectation in the
fixed CHSH Bell state. Denote this map by $E$; it maps
$\gamma I-H$ exactly to $\beta I-G$.

To verify its degree property, expand any substituted Bob polynomial
of length at most $k$ as $\sum_w w(B)\otimes M_w$, where all old words
have length at most $k$. An old Alice projector is retained, and the
auxiliary expectation of the square is an old Bob SOS of degree at most
$k$: its coefficient matrix is a positive Gram matrix obtained from
the fixed auxiliary state. A new Alice projector instead gives another
positive auxiliary Gram matrix and a Bob-only SOS. Such a Bob-only SOS
belongs to $\mathsf O^A_k$ by multiplying by the sum of the two outcomes of
any old Alice question. No same-party commutation is assumed in the
enlarged universal algebra; the commuting tensor substitution is a
valid homomorphism used only for this implication. The identical
argument with parties reversed maps $\mathsf O^B_k$ into $\mathsf O^B_k$. Expanding a
standard degree-$k$ square in auxiliary matrix coefficients likewise
produces a standard SOS of degree at most $k$. Thus $E$ preserves all
four cones and proves \eqref{eq:degree-preserving-profile}.
\end{proof}

\begin{corollary}
For $J_t$ in \eqref{eq:mixed-eight-interval}, $\widetilde J_t=J_t+C$ satisfies
\[
 \omega_Q(\widetilde J_t)-\omega_L(\widetilde J_t)=2\sqrt2-2>0,
 \qquad d_{\rm std}(\widetilde J_t)=3,\quad d_{\rm mix}(\widetilde J_t)=8.
\]
The quantum target is $\beta_t+2\sqrt2$ and is attained by a known
finite-dimensional strategy. Both parties have four binary questions.
This remains a finite-level statement.
\end{corollary}

\subsection{Bell splitting and an extreme-ray reduction}
\label{sec:bell-splitting}
The support argument has a useful consequence that is independent of
the number of questions. Let $\mathcal C$ be the cone of real Bell
polynomials nonnegative in every quantum model, and let $\mathcal B$
be the Bell subspace.

\begin{proposition}[Bell splitting and extreme rays]
\label{prop:bell-splitting}
If a Bell polynomial $P$ has a decomposition $P=P_A+P_B$ with
$P_A\in\mathsf O^A_k$ and $P_B\in\mathsf O^B_k$, then $P_A,P_B$ are themselves Bell
polynomials. At every quantum strategy attaining $P=0$, both summands
have expectation zero. If, additionally, $P\ne0$ spans an extreme ray
of $\mathcal C$, then
\[
 P\in\mathsf M_k\quad\Longleftrightarrow\quad
 P\in\mathsf O^A_k\ \hbox{or}\ P\in\mathsf O^B_k.
\]
In this extreme-ray case $d_{\rm mix}(P)=\min(d_A(P),d_B(P))$,
including infinite degrees.
\end{proposition}

\begin{proof}
In the reduced reflection basis the two conditioned support spaces
intersect exactly in $\mathcal B$. Since $P=P_A+P_B$ belongs to this
intersection, neither summand can retain a coefficient outside it:
there is no word in the other summand that could cancel that coefficient.
Both summands are nonnegative in every quantum model, so their zero
expectations follow from the zero expectation of their sum. Extremality
then forces $P_A=c_A P$, $P_B=c_B P$ with $c_A,c_B\ge0$ and
$c_A+c_B=1$. A nonzero scalar multiple in either conditioned cone places
$P$ itself in that cone. The converse is immediate.
\end{proof}

This proposition is a conditional route, not an assertion that the
present example or the coupled family of \cref{thm:bridge-ratio} spans such an extreme
ray. It explains precisely why adding easy certificates of opposite
orientations can destroy the intended mixed obstruction. A proof of
the all-level mixed-divergence statement must supply an all-level common witness, a verified
appropriate extreme-ray family, or another argument that controls
every possible mixed splitting.

\section{Additional finite-level constructions}
\label{sec:additional}
This section collects auxiliary certificates and structural examples.
It includes the standard tilted-CHSH identity used in the coupled
construction, low-level comparisons, and exact optimal-face certificates
that complement the analytic closure mechanism. These examples support
and delimit the main results; they are not separate asymptotic claims.
\subsection{An auxiliary degree-two versus mixed-three separation}
\label{sec:mixed-prototype}
This auxiliary example provides a low-degree certificate and a common
mixed witness that are also used in the exact-target construction below.

\begin{proposition}[Exact degree-two separation from the mixed cone]
\label{prop:mixed-Gstar}
For the archived $G^\ast$ of \cref{thm:L2-sep}, namely
\begin{align*}
 F={}&\frac{256}{965}A_0+\frac{266}{719}A_1+\frac{62}{183}B_0
       +A_0B_0+\frac{355}{647}A_0B_1\\
    &-\frac{767}{846}A_1B_0+\frac{857}{980}A_1B_1,
 \qquad \eta=\frac{501}{200},
\end{align*}
one has
\begin{equation}
 \eta I-F\in\mathsf D_2\cap
                 (\mathsf M_3\setminus\mathsf M_2).
 \label{eq:mixed-Gstar}
\end{equation}
Thus the standard degree is two and the mixed-direction degree is three
at this specified rational target.
\end{proposition}
\begin{proof}
The existing exact $\mathsf D_2$ and $\mathsf O^A_3$ certificates prove
the two memberships; both identities and all their Gram positivity
conditions have been rechecked from the archived rational artifacts.
For nonmembership, the new rational artifact
\datapath{research/bidirectional_coupling/mixed_Gstar_level2_exact.json}
supplies eight positive definite $5\times5$ Grams, four for each
orientation, on the full word basis
$(I,X_0,X_1,X_0X_1,X_1X_0)$ of the unconditioned party.
All moments through degree four are supplied, with reversal symmetry.
In each orientation, the two outcome sums agree on every moment and
have identity value one. The eight Bell moments shared by the two
orientations are exactly equal.

For every Gram, the five leading principal minors are strictly positive
in rational arithmetic. This verifies positive definiteness by
Sylvester's criterion. An invertible triangular change to the full
reduced PVM word basis preserves the degree and positivity; its forty
leading principal minors are checked separately as well.
The common exact Bell score is
\begin{equation}
 w=\frac{8127072587938224040860007}{3243311919325800000000000},
 \qquad
 w-\eta=\frac{2576230027095040860007}{3243311919325800000000000}
 >\frac{79}{100000}.
 \label{eq:mixed-Gstar-witness}
\end{equation}
Apply \cref{lem:mixed-gluing}. Nonmembership in $\mathsf D_1$ follows
from $\mathsf D_1\subseteq\mathsf O^A_1\subseteq\mathsf M_2$.
The checker
\codepath{research/bidirectional_coupling/verify_mixed_separation.py}
performs all these verifications with calls to SDP solvers disabled.
Numerical optimization was used only to discover the rational witness.
\end{proof}

\paragraph{A coefficient neighborhood.}
The separation is not isolated. For every coefficient vector
$\theta$ within $\ell_1$ distance $\rho=1/10000$ of the displayed
eight coefficients, the target $\eta'=\eta+\rho$ satisfies
$\eta'I-F_\theta\in\mathsf D_2\cap
(\mathsf M_3\setminus\mathsf M_2)$.
Indeed the change in the common witness value is at most $\rho$,
because all eight moments lie in $[-1,1]$. Therefore its margin over
$\eta'$ is at least $w-\eta-2\rho>0$.
The standard and mixed level-three upper certificates transport by
adding nonnegative multiples of $I\pm A_x$, $I\pm B_y$, and
$I\pm A_xB_y$, together with any remaining scalar slack.
These elementary gaps are in both conditioned level-one cones and in
$\mathsf D_2$.

\paragraph{Research boundary.}
This is a finite-level separation at a rational upper target; it is not
an exact-quantum-target degree-divergence theorem. In particular, it does
not settle existence of a Bell family with uniformly bounded standard
degree and unbounded mixed-direction degree at its exact quantum value.
It shows that mixing the directions is not a universal level-two repair,
and supplies an exactly verified starting point for that stronger
question. The algebraic mixed cone does not automatically coincide with
the proof system of any particular compiled-game protocol.

\subsection{An exact quantum-target interval beyond mixed degree two}
\label{sec:mixed-quantum-interval}

The fixed upper target in \cref{prop:mixed-Gstar} can be replaced by
an exactly attained quantum target on a parameter interval. This is
a further finite-level result; it does not give degree divergence.
Define
\begin{align}
 L&=\frac{256}{965}A_0+\frac{266}{719}A_1+\frac{62}{183}B_0,\notag\\
 C&=A_0B_0+\frac{355}{647}A_0B_1
       -\frac{767}{846}A_1B_0+\frac{857}{980}A_1B_1,\notag\\
 H_t&=tL+C,\qquad
 \beta_t=\frac{123675752}{126971805}t+
                  \frac{406751687}{268207380}.
 \label{eq:mixed-quantum-path}
\end{align}
Thus $H_1=G^\ast$ and $\beta_t$ is the payoff of the deterministic
all-plus strategy. Let $\mathsf M_k=\mathsf O^A_k+\mathsf O^B_k$.

\begin{theorem}[Exact quantum-target mixed-degree interval]
\label{thm:mixed-quantum-interval}
For every $t\in[101/100,51/50]$,
\begin{equation}
 \omega_Q(H_t)=\omega_L(H_t)=\beta_t,\qquad
 \beta_t I-H_t\in\mathsf D_2\cap
                (\mathsf M_3\setminus\mathsf M_2).
 \label{eq:mixed-quantum-interval}
\end{equation}
The standard exact degree is two and the mixed-direction exact
degree is three throughout this interval. One shared rational
level-two behavior separates uniformly with excess greater than
$59/100000$.
Moreover $\beta_t I-H_t\in\mathsf D_2\cap\mathsf O^A_3$ for every
$t\geq101/100$, so this upper parameter ray has a uniform degree-three
repair.
\end{theorem}

\begin{proof}
Put $t_0=101/100$. The exact standard certificate at $t_0$ has
the form
\begin{equation}
 \beta_{t_0}I-H_{t_0}
 =\sum_{i,j=1}^{12}Q_{ij}(W_i-I)^*(W_j-I),\qquad Q\succ0,
 \label{eq:mixed-quantum-standard}
\end{equation}
where $W_1,\ldots,W_{12}$ are all nonidentity reduced bipartite
reflection words of total length at most two. All entries of $Q$
are rational and are supplied in
\datapath{research/bidirectional_coupling/quantum_target_family_exact.json}.
The checker reconstructs the identity in the full free quotient
algebra and verifies its twelve positive leading principal minors.
The factors vanish at the all-plus strategy. This factor basis
removes the known optimal kernel before the numerical search;
the resulting strictly positive Gram is then reconstructed rationally.

The exact conditioned upper certificate at $t_0$ is
\begin{equation}
 \beta_{t_0}I-H_{t_0}
 =\sum_{x=0}^1\sum_{a=0}^1
 P_{a|x}\, f_{a|x}(B)^*Q_{a|x} f_{a|x}(B),
 \qquad Q_{a|x}\succ0.
 \label{eq:mixed-quantum-conditioned}
\end{equation}
Here $f_{0|x}$ lists $W-I$ for the six nonidentity Bob words of
length at most three, and $f_{1|x}$ lists all seven Bob words
through degree three. The four rational Grams, of sizes $6,7,6,7$,
are in
\datapath{research/bidirectional_coupling/quantum_conditioned_upper_exact.json}.
Every coefficient identity and all 26 leading principal minors are
checked exactly. Each square is in the native Alice-conditioned
degree-three cone; no relation between different Bob questions is added.

These endpoint certificates extend to the whole ray by the exact identity
\begin{align}
 \beta_t I-H_t={}&\beta_{t_0}I-H_{t_0}\notag\\
 &+(t-t_0)\left[\frac{256}{965}(I-A_0)
       +\frac{266}{719}(I-A_1)+\frac{62}{183}(I-B_0)\right].
 \label{eq:mixed-quantum-transport}
\end{align}
For $t\geq t_0$ the additional terms are nonnegative elementary
gaps belonging to $\mathsf D_1$ and both conditioned level-one cones.
The deterministic strategy attains $\beta_t$, proving the asserted
local and quantum equalities without a qubit optimization argument.

For nonmembership, use the same eight rational level-two Grams as in
\cref{prop:mixed-Gstar}. Their common Bell behavior is $\ell$.
The excess $\ell(H_t)-\beta_t$ is affine in $t$, with derivative
$\ell(L)-L(+,+,+,+)<0$. The exact minimum on the interval is
\begin{equation}
 \ell(H_{51/50})-\beta_{51/50}
 =\frac{111872951957204276161427}{189193195294005000000000000}
 >\frac{59}{100000}.
 \label{eq:mixed-quantum-gap}
\end{equation}
The gluing lemma \cref{lem:mixed-gluing} therefore excludes
$\mathsf M_2$ throughout the interval. Membership in
$\mathsf O^A_3\subseteq\mathsf M_3$ proves the exact mixed degree.
Since $\mathsf D_1\subseteq\mathsf M_2$, the standard degree is
exactly two as well.
\end{proof}

\paragraph{Verification and interpretation.}
Run
\begin{verbatim}
python research/bidirectional_coupling/verify_quantum_target_family.py
\end{verbatim}
The default checker reads only the exact artifacts and uses rational
arithmetic and SymPy. It does not import a numerical SDP solver.
Certificates at two additional rational parameters are independent
cross-checks; the continuum claim follows from
\eqref{eq:mixed-quantum-transport} and the affine witness excess,
not from sampling the interval.
The equality $\omega_Q=\omega_L$ is essential context: this result
concerns the cost of certifying an exact Bell bound even in the
absence of a quantum violation. It is not a new nonlocal advantage.
Nor can this certified parameter ray produce the desired unbounded
mixed-direction degree, because its degree-three upper bound is uniform.

\subsection{Level one: no equality --- the tilted family separates exactly}
\label{sec:L1}

Write the tilted-CHSH functional as
\begin{equation*}
F_\alpha=\alpha\,\bigl(P_{A_{0|0}}-P_{A_{1|0}}\bigr)
 +\sum_{x,y}c_{xy}\,E_{xy},\qquad
(c_{00},c_{01},c_{10},c_{11})=(1,1,1,-1),
\end{equation*}
with $E_{xy}=(2P_{A_{0|x}}-I)(2P_{B_{0|y}}-I)$; its quantum value is
$\sqrt{8+2\alpha^2}$ \cite{Acin2012Tilted,BampsPironio2015}.  Throughout we
write $P_{A,x}:=P_{A_{0|x}}$ for the retained (outcome-$0$) projector of
Alice's question $x$, and $P_{B,y}$ likewise for Bob; the eliminated
projector is $I-P_{A,x}$.

\begin{lemma}[Alice-conditioned level-one behaviors are nonsignaling]
\label{lem:os1-ns}
Let $\{\varphi_{a|x}\}$ be feasible for the Alice-conditioned level-one problem
(Definition~4.1 of \cite{CuiFalorNatarajanZhang2025}).  Then
$p(ab|xy):=\varphi_{a|x}(B_{b|y})$ is a nonsignaling behavior.
\end{lemma}

\begin{proof}
Positivity: $B_{b|y}=B_{b|y}^\dagger B_{b|y}$, so
$p(ab|xy)=\varphi_{a|x}(B_{b|y}^\dagger B_{b|y})\geq0$ since
$\varphi_{a|x}\succeq0$.  Bob nonsignaling: consistency makes
$\sum_a\varphi_{a|x}$ independent of $x$, so $\sum_a p(ab|xy)=p(b|y)$.
Alice nonsignaling: $\sum_b p(ab|xy)=\varphi_{a|x}(\sum_b B_{b|y})
=\varphi_{a|x}(I)$, independent of $y$.  Normalization: consistency plus
$\sum_a\varphi_{a|0}(I)=1$ gives $\sum_a p(a|x)=1$ for every $x$.
\end{proof}

\begin{corollary}[Value bound through the nonsignaling polytope]
\label{cor:ns-bound}
For every Bell functional $G$, $\omega_1^{\rm os}(G)\leq\omega_{\rm ns}(G)$.
\end{corollary}

\begin{proposition}[The sandwich]
\label{prop:sandwich}
$\widetilde{\mathcal Q}^{\rm beh}\subseteq\mathsf O_1^{\rm beh}\subseteq
\mathrm{NS}$, and both inclusions are strict.
\end{proposition}

\begin{proof}
The right inclusion is \cref{lem:os1-ns}.  For the left inclusion, take a
$1{+}AB$ (almost-quantum) moment matrix \cite{Navascues2015AlmostQuantum}
with all outcomes present. For \emph{each} $a\in\{0,1\}$, its principal
submatrix on $\{A_{a|x},A_{a|x}P_{B,0},A_{a|x}P_{B,1}\}$ is PSD.
Its $(u,v)$ entry is $m(A_{a|x}u^\dagger v)$, using idempotency and
commutation. In the reduced basis the same construction is a congruence
$T_{a|x}\Gamma T_{a|x}^T$, including $A_{1|x}=I-P_{A,x}$; PSD of that
block does not follow merely by subtracting the $a=0$ block from the
outcome sum. Summing the two constructed blocks gives
$m(u^\dagger v)$, so consistency and normalization hold. The four maps
are checked symbolically in \codepath{scripts/standard_tilted_sos.py}.
Right
strictness: a PR box is a nonsignaling behavior with CHSH value $4$, while
every Alice-conditioned level-one behavior has CHSH value at most
$\omega_1^{\rm os}(F_0)=2\sqrt2<4$ (the $\alpha=0$ closure of
\cref{thm:L1-sep}); hence the PR box lies in $\mathrm{NS}\setminus\mathsf
O_1^{\rm beh}$.  Left
strictness is exact, by an explicit rational Alice-conditioned level-one primal
witness at $\alpha=1/4$ with objective value above the almost-quantum value
$q=\sqrt{8+2\alpha^2}$ (proved in \cref{lem:standard-tilted}; the witness is verified without solvers;
\datapath{data/contact_geometry/sandwich.json}).  The numerical fit of the inset of
Figure~\ref{fig:L1} ($\omega_1^{\rm os}\approx2\sqrt2+\alpha^2/\sqrt2$
versus $q\approx2\sqrt2+\frac{\sqrt2}{4}\alpha^2$) displays the same
strictness continuously.
\end{proof}

\begin{theorem}[Standard level-one value, exact]
\label{thm:L1-std}
For all $0\leq\alpha<2$,
\begin{equation*}
\omega_1^{\rm std}(F_\alpha)=2\sqrt2+\alpha .
\end{equation*}
\end{theorem}

\begin{proof}[Proof (certificate; verified without solvers)]
Lower bound. Use the identity followed by Alice's two retained projectors
and Bob's two retained projectors as the word order. The
$\alpha$-independent moment matrix
\begin{equation*}
\Gamma=\begin{pmatrix}
1 & 1 & \tfrac12 & s & s\\
1 & 1 & \tfrac12 & s & s\\
\tfrac12 & \tfrac12 & \tfrac12 & u & \tfrac14\\
s & s & u & s & u\\
s & s & \tfrac14 & u & s
\end{pmatrix},
\qquad s=\frac{2+\sqrt2}{4},\quad u=\frac{1+\sqrt2}{4},
\end{equation*}
has exact eigenvalues $0$ (multiplicity $3$) and
$\frac{7+\sqrt2\pm\sqrt{28+10\sqrt2}}{4}>0$, satisfies moment consistency
and normalization, and attains objective $2\sqrt2+\alpha$ for every
$\alpha$; the exact verification is performed entry-wise over
$\QQ(\sqrt2)$ by the script
\codepath{scripts/run_level1_family.py}.

Upper bound.  The degree-one SOS identity
\begin{align*}
\bigl(2\sqrt2+\alpha\bigr)I-F_\alpha
={}&\tfrac{1}{\sqrt2}\Bigl(A_0-\tfrac{B_0+B_1}{\sqrt2}\Bigr)^2
 +\tfrac{1}{\sqrt2}\Bigl(A_1-\tfrac{B_0-B_1}{\sqrt2}\Bigr)^2
 +2\alpha\,(I-P_{A_{0|0}})^2,
\end{align*}
with $A_x=2P_{A_{0|x}}-I$ and $B_y=2P_{B_{0|y}}-I$, holds word-by-word in
the quotient algebra with $\alpha$ symbolic; all square coefficients are
nonnegative for $\alpha\geq0$.  The identity is expanded and checked
coefficient-wise by \texttt{npa2.exact.verify\_sos\_identity}.
\end{proof}

\begin{lemma}[Nonsignaling ceiling of the tilted family]
\label{lem:ns-ceiling}
For $0\leq\alpha<2$, $\omega_{\rm ns}(F_\alpha)=4$.
\end{lemma}

\begin{proof}
The nonsignaling maximum of $F_\alpha$ is attained at a vertex of the
polytope: the PR boxes give $4$ (all marginals zero), and every local
deterministic vertex gives at most $2+\alpha<4$ for $\alpha<2$.
Note that the standard level-one body is \emph{not} nonsignaling: the
witness of \cref{thm:L1-std} has $p(1,1|1,1)=(1-\sqrt2)/4<0$.
\end{proof}

\begin{remark}
\label{rem:chsh-nice}
Each square in the level-one CHSH part uses only one Alice question: the
standard degree-one CHSH certificate is already nice.  This is why
$\alpha=0$ is the equality point of \cref{thm:L1-sep} below; the tilt term
is where the two cones diverge.
\end{remark}

\begin{theorem}[Strict level-one separation on the whole family, with an
explicit linear margin]
\label{thm:L1-sep}
For all $0<\alpha<2$,
\begin{equation*}
\omega_1^{\rm std}(F_\alpha)-\omega_1^{\rm os}(F_\alpha)
\ \geq\ (\sqrt2-1)\,\alpha\ >\ 0,
\end{equation*}
equivalently $\omega_1^{\rm os}(F_\alpha)\leq
2\sqrt2+(2-\sqrt2)\alpha<2\sqrt2+\alpha=\omega_1^{\rm std}(F_\alpha)$.
\end{theorem}

\begin{proof}[Proof by endpoint convexity]
Write $f(\alpha):=\omega_1^{\rm os}(F_\alpha)$.  The Alice-conditioned level-one
feasible body does not depend on $\alpha$ --- the tilt enters only the
objective, which is affine in $\alpha$ --- so $f$ is the support function
of a fixed body composed with an affine map, hence convex.

Two endpoint facts are certified exactly.  Left: $f(0)=2\sqrt2$.  The
degree-one CHSH certificate of \cref{thm:L1-std} is already nice
(\cref{rem:chsh-nice}), giving the upper bound, and the quantum value
$2\sqrt2$ gives the lower bound since the relaxation contains the quantum
set.  Right: $f(\alpha)\leq\omega_{\rm ns}(F_\alpha)=4$ for every
$\alpha<2$ (\cref{cor:ns-bound,lem:ns-ceiling}).

For any $\beta\in(\alpha,2)$, convexity across $0<\alpha<\beta$ gives
\begin{equation*}
f(\alpha)\ \leq\ \Bigl(1-\frac{\alpha}{\beta}\Bigr)\,f(0)
+\frac{\alpha}{\beta}\,f(\beta)
\ \leq\ \Bigl(1-\frac{\alpha}{\beta}\Bigr)\,2\sqrt2
+\frac{\alpha}{\beta}\cdot4,
\end{equation*}
and $\beta\to2^{-}$ yields $f(\alpha)\leq2\sqrt2+(2-\sqrt2)\alpha$.
Subtracting from $\omega_1^{\rm std}(F_\alpha)=2\sqrt2+\alpha$
(\cref{thm:L1-std}) leaves the margin
$\alpha-(2-\sqrt2)\alpha=(\sqrt2-1)\alpha>0$.
\end{proof}

\begin{remark}[Exact anchor certificates: cross-checked upper bounds]
\label{rem:anchors}
The chord $2\sqrt2+(2-\sqrt2)\alpha$ is certified but not tight:
numerically, the Alice-conditioned value leaves $2\sqrt2$ with vanishing right
derivative (fit $2\sqrt2+\alpha^2/\sqrt2$, a numerical identification;
\cref{sec:L1-form}).  Certified upper bounds come from exact rational dual
certificates: the affine
identity
\begin{equation*}
t-F(m)=\sum_{x,a}\operatorname{Tr}\!\bigl(S_{x,a}\Gamma_{a|x}(m)\bigr)
 +\operatorname{Tr}\!\bigl(\Lambda\,{\textstyle\sum_a}
 (\Gamma_{a|1}-\Gamma_{a|0})\bigr)
 +\nu\bigl({\textstyle\sum_a} m_{0,a}(I)-1\bigr)
\end{equation*}
holds coefficient-wise in the free block moments $m$, with each $S_{x,a}$
exactly PSD (rational entries, exact $LDL^\ast$ test).  At ten rational
anchors $\alpha_i\in\{1/16,1/8,1/4,3/8,1/2,3/4,1,5/4,3/2,7/4\}$ these give
$\omega_1^{\rm os}(F_{\alpha_i})\leq U_i$, with $U_i<2\sqrt2+\alpha_i$
compared exactly over $\QQ(\sqrt2)$ (for rational $r>0$, $r<2\sqrt2$ iff
$r^2<8$).  The anchors are certified upper bounds consistent with the
numerical curve (Figure~\ref{fig:L1}) and solver-independent cross-checks
of the SDP implementation; the separation of \cref{thm:L1-sep} itself uses
only the two endpoints.  Every step is exact rational arithmetic in
\codepath{scripts/run_level1_anchors.py}.
\end{remark}

\begin{corollary}[Strictness of the conversion; failure of the level-one
value-equality claim]
\label{cor:cfnz}
The inclusion $\mathsf D_1\subseteq\mathsf O_1$ of
\cref{prop:d1} is strict, and the level-one value-equality statement of
\cite[Section~5.2]{CuiFalorNatarajanZhang2025} (v1, 2025) fails as stated
under its raw PVM convention (audited row by row in
Table~\ref{tab:audit}); the one-direction conversion of \cref{prop:d1} is
unaffected.  For the sequential POVM/localizer convention of
\cite{KlepEtAl2025} we record a caution only, not a correction
(\cref{rem:audit}).
\end{corollary}

\begin{proof}
By \cref{thm:L1-sep}, $\omega_1^{\rm os}(F_{1/2})<\omega_1^{\rm
std}(F_{1/2})$.  If the standard and Alice-conditioned level-one primal values
coincided for every functional this would be a contradiction.  Strictness:
for any $\beta\in[\,\omega_1^{\rm os}(F_{1/2}),\,\omega_1^{\rm
std}(F_{1/2}))$, the polynomial $\beta I-F_{1/2}$ lies in $\mathsf O_1$ but
not in $\mathsf D_1$.
\end{proof}

\begin{proposition}[Level-one value comparison]
\label{prop:criterion}
For every Bell functional $G$,
\begin{equation*}
\omega_1^{\rm os}(G)\ \leq\ \omega_1^{\rm std}(G).
\end{equation*}
Moreover, if the two values are equal then the common value
$\beta=\omega_1^{\rm std}(G)$ is certified by a nice degree-one SOS
decomposition; and a sufficient condition for equality is that
$\omega_1^{\rm std}(G)=\omega_q(G)$ and that value admits a nice degree-one
certificate (then $\omega_q(G)\leq\omega_1^{\rm os}(G)\leq\omega_1^{\rm
std}(G)=\omega_q(G)$).
\end{proposition}

\begin{proof}
The inequality is a value-level consequence of the certificate inclusion
$\mathsf D_1\subseteq\mathsf O_1$ (\cref{prop:d1}): if
$\beta I-\Bell_G\in\mathsf D_1$ then the converted nice certificate gives
$\beta I-\Bell_G\in\mathsf O_1$, so the Alice-conditioned infimum is no larger.
(The same conclusion is obtained behaviorally by the Gram-vector
construction of \cite[Section~5.2]{CuiFalorNatarajanZhang2025}; we
implemented that construction and confirmed that the Alice-conditioned level-one
optimal behaviors of our controls lift to standard level-one moment
matrices.)  Suppose now that equality holds with common value
$\beta=\omega_1^{\rm std}(G)=\omega_1^{\rm os}(G)$.  By
\cref{lem:attainment} the Alice-conditioned optimum is attained by an explicit
degree-one nice SOS decomposition at $\beta$, not merely approached in the
closure.  The sufficient condition then follows from
$\omega_q(G)\leq\omega_1^{\rm os}(G)$, which holds because the Alice-conditioned
body contains the quantum set.
\end{proof}

\begin{remark}
\label{rem:criterion-scope}
Note what is \emph{not} claimed: the mere existence of \emph{some} nice
certificate at $\beta=\omega_1^{\rm std}(G)$ does \emph{not} imply
equality, since a nicer certificate may sit lower.  CHSH is the equality
case through the sufficient condition (its standard degree-one certificate
is already nice, \cref{rem:chsh-nice}, and its level-one value is the
quantum value); tilted CHSH with $\alpha>0$ is strict by
\cref{thm:L1-sep}.
\end{remark}

\begin{figure}[t]
\centering
\includegraphics[width=0.92\textwidth]{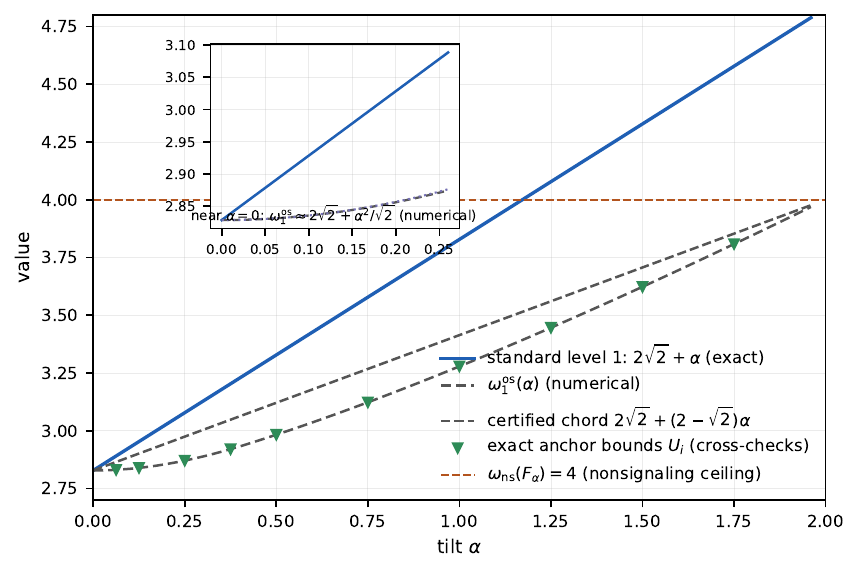}
\caption{Level-one separation of tilted CHSH (\cref{thm:L1-std,thm:L1-sep}).
The standard level-one value is exactly $2\sqrt2+\alpha$
(\cref{thm:L1-std}); the dashed curve is the numerical Alice-conditioned level-one
value.  The gray dashed line is the certified chord
$2\sqrt2+(2-\sqrt2)\alpha$ of \cref{thm:L1-sep}, obtained from the exact
endpoint $(0,2\sqrt2)$ and the nonsignaling ceiling
$\omega_{\rm ns}(F_\alpha)=4$ (orange dashed line, \cref{lem:ns-ceiling}),
below which the Alice-conditioned body sits
(Lemma~\ref{lem:os1-ns} and Proposition~\ref{prop:sandwich}); the strip
between chord and standard line is the certified margin
$(\sqrt2-1)\alpha$.  Green marks are ten exact rational upper bounds
$U_i$ (\cref{rem:anchors}) --- cross-checks, not load-bearing.  Inset:
near $\alpha=0$ the Alice-conditioned value leaves $2\sqrt2$ with vanishing right
derivative (numerical fit $2\sqrt2+\alpha^2/\sqrt2$), the standard value
linearly.  Certified: the endpoint, the ceiling, the chord, and the
anchor points; the drawn curves are numerical.
}
\label{fig:L1}
\end{figure}

\subsubsection{The Alice-conditioned level-one value as a function of alpha}
\label{sec:L1-form}

High-precision KKT analysis (150-digit homotopy refinement plus algebraic
identification; numerical, not yet a formal proof) indicates that
$\omega_1^{\rm os}(F_\alpha)$ is a two-regime algebraic function with
critical point $\alpha_c=1-1/\sqrt3$: for $0\leq\alpha\leq\alpha_c$,
\begin{equation*}
\omega_1^{\rm os}(F_\alpha)^2
 =\frac{2\alpha^4+10\alpha^2-1+(1+4\alpha^2)^{3/2}}{2\alpha^2}
 \qquad\Bigl(=8\text{ at }\alpha=0\Bigr),
\end{equation*}
with $\omega_1^{\rm os}(F_\alpha)=2\sqrt2+\frac{\sqrt2}{2}\alpha^2
-\frac{5\sqrt2}{16}\alpha^4+\cdots$ as $\alpha\to0$; in particular the
right derivative at $0$ vanishes while the standard value leaves
$2\sqrt2$ linearly.  We use this only as guidance; no theorem in this
paper depends on it.

\subsection{Level two: the separation reverses}
\label{sec:L2}

The search of \cref{sec:search} produced the following functional.  The
lead candidate $G^\ast$ (rationalized; this is the functional used
below) has correlator matrix and marginals
\begin{equation*}
C=\begin{pmatrix}1 & \tfrac{355}{647}\\[2pt] -\tfrac{767}{846} & \tfrac{857}{980}\end{pmatrix},
\qquad
a=\bigl(\tfrac{256}{965},\tfrac{266}{719}\bigr),\qquad
b=\bigl(\tfrac{62}{183},0\bigr),
\end{equation*}
in the convention $F=\sum_{xy}c_{xy}E_{xy}+\sum_xa_xE_{A_x}+\sum_yb_yE_{B_y}$
with $E_{A_x}=P_{A_{0|x}}-P_{A_{1|x}}$ and likewise for Bob.  Its
$|c_{xy}|$-multiset has four distinct values and it has three nonzero
marginal components, so it lies outside the equivalence closure of the
catalogued families.

\begin{theorem}[Exact level-two cone separation]
\label{thm:L2-sep}
Let $\eta=501/200$.  Then
\begin{equation*}
\eta I-\Bell_{G^\ast}\ \in\ \mathsf D_2\ \text{ but }\ \notin\ \mathsf O_2 .
\end{equation*}
\end{theorem}

\begin{proof}[Proof (two exact certificates; verified without solvers)]
Membership in $\mathsf D_2$: an exact rational standard level-two dual
certificate
\begin{equation*}
\eta I-\Bell_{G^\ast}=\operatorname{Tr}(S\,\Gamma(m))+\nu\,(m(I)-1)
\end{equation*}
holds coefficient-wise in the free moments, with the $13\times13$ Gram
matrix $S$ exactly PSD (rational $LDL^\ast$ test).  Nonmembership in
$\mathsf O_2$: an exact rational Alice-conditioned level-two primal witness
(blocks $\Gamma_{a|x}$, exact PSD, exact consistency and normalization,
with the last-outcome block carrying the consistency residual and a
$10^{-3}$ admixture of an exact feasible reference point lifting the
rationalized blocks off the PSD boundary; the final block PSD is verified
exactly a posteriori) attains the exact rational
objective value $w>\eta$, with $w\approx2.51030724$.
The full rational value and all witness entries are stored in
\datapath{data/finite_separations/exact/search_s31337_308.json}; the verifier
reconstructs $w$ from those entries and checks $w>\eta$ exactly. Hence
$\omega_2^{\rm std}(G^\ast)\leq\eta<\omega_2^{\rm os}(G^\ast)$.
Verification: \codepath{scripts/run_m3_exactify.py} and its solver-free
\texttt{{-}{-}verify-only} mode.
\end{proof}

\begin{corollary}[One extra level for the bound-$\eta$ polynomial]
\label{cor:dos3}
The Bell-gap polynomial $\eta I-\Bell_{G^\ast}$ belongs to $\mathsf O_3$:
an exact rational Alice-conditioned level-three dual certificate with bound
$\eta=501/200$ holds coefficient-wise, with all four Gram blocks exactly
PSD (\datapath{data/finite_separations/exact/o3_308.json}, verified without solvers).
Hence, at the target bound $\eta$ of \cref{thm:L2-sep},
$d_{\rm os}(G^\ast,\eta)=3$ while $d_{\rm std}(G^\ast,\eta)=2$
(\cref{def:convdeg}): the standard level-two certificate converts to the
Alice-conditioned cone at the price of exactly one additional level,
$\eta I-\Bell_{G^\ast}\in\mathsf D_2\cap(\mathsf O_3\setminus\mathsf O_2)$.
The standard degree cannot be one, since
$\mathsf D_1\subseteq\mathsf O_1\subseteq\mathsf O_2$ would contradict
the separating witness. The target $\eta=2.505$ is our certified bound, not the quantum value
(numerically $\omega_q\approx2.49728$); $d_{\rm os}(G^\ast,\omega_q)$ is
not claimed.
\end{corollary}

\begin{theorem}[Uniform separation on a coefficient neighbourhood]
\label{thm:robust}
Write
\begin{equation*}
\theta^\ast=(a_0,a_1,b_0,b_1,c_{00},c_{01},c_{10},c_{11})\in\QQ^8
\end{equation*}
for the $\pm1$-observable coefficients of $G^\ast$ displayed above, and
let $r=1/400$ and $\eta'=\eta+r=1003/400$.  Then for every
$\theta\in\mathbb{R}^8$ with $\|\theta-\theta^\ast\|_1\leq r$, the associated
functional $G_\theta$ satisfies
\begin{equation*}
\omega_2^{\rm std}(G_\theta)\ \leq\ \eta'\ <\ \omega_2^{\rm os}(G_\theta),
\qquad
\eta' I-\Bell_{G_\theta}\ \in\ \mathsf D_2\cap(\mathsf O_3\setminus
\mathsf O_2).
\end{equation*}
\end{theorem}

\begin{proof}[Proof (exact transport; verified without solvers)]
Write $X_j$ for the eight $\pm1$ monomials
\begin{equation*}
X_j\ \in\ \{E_{A_0},E_{A_1},E_{B_0},E_{B_1},E_{00},E_{01},E_{10},E_{11}\};
\end{equation*}
each
satisfies $X_j^2=I$, hence $I\pm X_j=\tfrac12(I\pm X_j)^2$.  With
$\delta=\theta-\theta^\ast$ one has
$\Bell_{G_\theta}-\Bell_{G^\ast}=\sum_j\delta_jX_j$ (the constant term
cancels by the outcome-sign symmetry) and
\begin{multline*}
\eta' I-\Bell_{G_\theta}
=\bigl(\eta I-\Bell_{G^\ast}\bigr)\\
+\tfrac12\sum_j\Bigl(\delta_j^{+}(I-X_j)^2+\delta_j^{-}(I+X_j)^2\Bigr)
+\bigl(r-\|\delta\|_1\bigr)I ,
\end{multline*}
with $\delta_j^+=\max(\delta_j,0)$ and
$\delta_j^-=\max(-\delta_j,0)$.
The first bracket is the certificate of \cref{thm:L2-sep}
(resp.\ \cref{cor:dos3}); each summand is a nonnegative multiple of a
square of a word polynomial of length $\leq2$ using at most one Alice
question, hence lies in $\mathsf D_2$ and in $\mathsf O_3$.  This proves
the membership and $\omega_2^{\rm std}(G_\theta)\leq\eta'$.  For
nonmembership, the exact Alice-conditioned level-two witness of
\cref{thm:L2-sep} evaluates every $X_j$ in $[-1,1]$ (its induced behavior
is nonsignaling, \cref{lem:os1-ns}), so its objective value on
$G_\theta$ is at least $w-\|\delta\|_1\geq w-r$, where $w$ is its exact
rational value on $G^\ast$.  The comparisons $w-\eta>1/200$ and
$w-r>1003/400$ are exact rational arithmetic
(\datapath{data/finite_separations/robust_family.json}, which also re-derives the
eight expectation values and the value identity from the stored block
entries).  Hence $\eta'I-\Bell_{G_\theta}\notin\mathsf O_2$.
\end{proof}

\begin{remark}
\label{rem:robust-scope}
The radius is limited by the exact margin $w-\eta\approx5.3\times10^{-3}$:
any $r<(w-\eta)/2$ works.  The statement certifies the conversion degree
of the explicit upper-bound polynomial $\eta'I-\Bell_{G_\theta}$:
$d_{\rm os}(G_\theta,\eta')=3$ while $d_{\rm std}(G_\theta,\eta')=2$
(\cref{def:convdeg}); it does not claim that the perturbed games'
quantum values close at level three.
\end{remark}

\begin{remark}[Replication and controls]
\label{rem:replication}
A second, inequivalent game (candidate \texttt{s777\#298}) passes the same
pipeline with $\eta=11197/5000$, so the phenomenon is not a single-instance
accident.  A transport control from the same search admits an exact
Alice-conditioned level-two dual certificate with bound
$\omega_2^{\rm std}+1.2\times10^{-6}$, so the contrast is not a solver
artifact.  For both separated candidates, an explicit two-qubit
tensor-product strategy attains $\omega_2^{\rm std}$ to within
$2.2\times10^{-9}$ numerically, i.e. the standard level-two value closes at
the quantum value; we record this as numerical evidence, not as part of the
exact statement.
\end{remark}

\subsubsection{Search protocol}
\label{sec:search}

We sampled $2700$ general $(2,2,2,2)$ functionals
$F=\sum_{xy}c_{xy}E_{xy}+\sum_x a_xE_{A_x}+\sum_y b_yE_{B_y}$ (normalized
correlator part, random marginal structure, fixed seeds) and computed
$\omega_2^{\rm std},\omega_3^{\rm std},\omega_2^{\rm os},\omega_3^{\rm os}$
on the raw cones with a two-solver fallback chain (CLARABEL, then SCS on
failure), plus the classical value by exhaustive
enumeration; the two candidates carried through the full pipeline below
were re-solved independently under both solvers
(\datapath{data/finite_separations/candidates/}).  Candidates are flagged by machine-checked novelty: the
ratio multiset $\{|c_{xy}|/\max_{x'y'}|c_{x'y'}|\}$ is invariant under question permutations, outcome
flips, party swap, and rescaling; CHSH-type correlators have all
$|c_{xy}|$ equal, XOR games have no marginals, and $B_3$ lives in
$(2,2,3,3)$.  This protocol produced $171$ level-two transport examples
and $12$ separation candidates with
$\omega_2^{\rm std}=\omega_3^{\rm std}=\omega_3^{\rm os}$ but
$\omega_2^{\rm os}-\omega_2^{\rm std}\in[2.5\times10^{-4},1.55\times10^{-2}]$,
of which $G^\ast$ and the replicated candidate of
\cref{rem:replication} were carried through dual-solver, outcome-complete,
and explicit-strategy checks before exactification.

\subsection{Tilted CHSH: exact points and a separating interval}
\label{sec:L2t}

\begin{lemma}[Standard and almost-quantum exactness]
\label{lem:standard-tilted}
For every $0\leq\alpha<2$,
\begin{equation*}
\omega^{1+AB}(F_\alpha)=\omega^{\rm std}_2(F_\alpha)
=\omega_q(F_\alpha)=q(\alpha)=\sqrt{8+2\alpha^2}.
\end{equation*}
For $0<\alpha<2$, $d_{\rm std}(F_\alpha,q(\alpha))=2$.
\end{lemma}
\begin{proof}
Set $R=qI-F_\alpha$ and
$S'=A_0(B_0-B_1)+A_1(B_0+B_1)$. The identity of
\cite[Eq.~(27)]{BampsPironio2015} is
\begin{equation}
R=\frac{R^2+(\alpha A_1-S')^2}{2q}.
\label{eq:bp15}
\end{equation}
Both squared Hermitian polynomials, $R$ and $\alpha A_1-S'$, lie in
$\operatorname{span}\{I,A_0,A_1,B_0,B_1,A_xB_y\}$, the $1{+}AB$ word
space, and $1/(2q)>0$. Our verifier expands \eqref{eq:bp15} word by word
with $\alpha$ symbolic. The quantum strategy of \cref{sec:face} attains
$q$ and is feasible in both relaxations. Since $1{+}AB$ is contained in
the standard level-two word space,
$q\leq\omega^{\rm std}_2\leq\omega^{1+AB}\leq q$.
Finally \cref{thm:L1-std} gives $\omega^{\rm std}_1>q$ for $0<\alpha<2$.
The symbolic SOS, attaining-state identities, and the four outcome
congruences are stored in \datapath{data/analytic/standard_tilted/family.json}
and rechecked by the common verification entry point.
\end{proof}

Thus the standard side is exact on the entire tilted family. The
Alice-conditioned level-two relaxation behaves differently, as certified
below; its complete exactness region remains open.

\begin{figure}[t]
\centering
\includegraphics[width=0.68\textwidth]{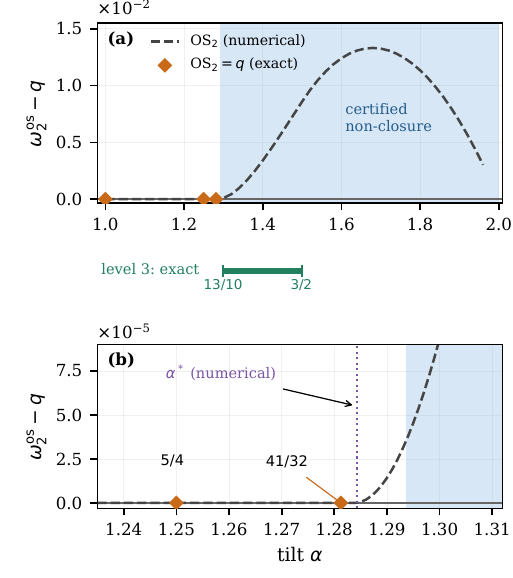}
\caption{Exact quantum-bound certification. Orange diamonds are exact
level-two closure points. The separate green strip marks an interval of
exact level-three closure; it is not plotted on the level-two gap axis.
The blue open interval shows the finite witness fan at level two;
\cref{sec:unbounded} additionally extends non-closure to every
$\alpha\in(\alpha_-,2)$.
Dashed curves and the candidate boundary near $1.28427$ are numerical;
the existence of a unique transition is not proved.
}
\label{fig:L2}
\end{figure}

\begin{theorem}[Exact quantum-bound certificates and separation]
\label{thm:L2-threshold}
For tilted CHSH at level two:
\begin{enumerate}[label=(\roman*)]
\item (non-closure, exact on an open interval) for all
$\alpha\in(\alpha_-,2)$, where $\alpha_-\approx1.293774$ is an exact
algebraic coverage boundary (strict separation at $\alpha_-$ is not
asserted, and at $\alpha=2$ the quantum value is attained),
\begin{equation*}
\omega_2^{\rm os}(F_\alpha)\ >\ \sqrt{8+2\alpha^2};
\end{equation*}
in particular, at the quantum-value target the conversion degrees satisfy
$d_{\rm os}(F_\alpha,q(\alpha))\geq3$ while $d_{\rm
std}(F_\alpha,q(\alpha))\leq2$ on this interval (\cref{def:convdeg}; the
standard side closes at level two);
\item (exact level-two closure) at
$\alpha\in\{1,5/4,41/32\}$, explicit algebraic dual certificates prove
$\omega_2^{\rm os}(F_\alpha)=q(\alpha)$ exactly;
\item (exact level-three closure on an interval) for every
$\alpha\in[13/10,3/2]$, a rational-function dual certificate proves
$\omega_3^{\rm os}(F_\alpha)=q(\alpha)$ exactly. This entire closed interval
lies inside the open interval in (i), so $d_{\rm os}(F_\alpha,q)=3$ and
$d_{\rm std}(F_\alpha,q)=2$.
\end{enumerate}
\end{theorem}

\begin{proof}[Proof sketch (exact verification per step)]
(i) Each anchor's exact primal witness $m_i$ (rational, verified feasible)
gives an affine lower bound $\omega_2^{\rm os}(F_\alpha)\geq
g_i+\alpha h_i$ valid for \emph{all} $\alpha$: the Alice-conditioned level-two
feasible body does not depend on $\alpha$ at all --- the tilt enters only
the objective $F_\alpha$, which is affine in $\alpha$ --- so evaluating the
support function at a fixed feasible point is legal.  Non-closure at $\alpha$ follows when
$g_i+\alpha h_i>\sqrt{8+2\alpha^2}$; since both sides are positive on our
range this is the quadratic inequality
$(g_i+\alpha h_i)^2>8+2\alpha^2$ with rational coefficients, whose roots
are computed by exact root isolation over $\QQ$.  Ten anchors
($13/10,21/16,11/8,3/2,13/8,7/4,15/8,19/10,197/100,199/100$) give overlapping
intervals merging to the open interval $(\alpha_-,\alpha_+)$.  The band
$(41/32,\alpha_-)$ lies between the largest exact level-two closure
point and the lower coverage boundary. It has no complete classification
from the present certificates: the
generator only attempts anchors whose numerical gap exceeds a
robust-witness cutoff of $2\times10^{-5}$, so a finer fan there is an
unresolved certification problem. The number $\alpha^\ast$ is only a
numerical candidate, not an endpoint of a certified interval.
The analytic level-two witness in \cref{sec:unbounded} separates on
$(254/161,2)$; this overlaps the fan at $\alpha=8/5$ and extends the
certified union to $(\alpha_-,2)$. The upper fan boundary $\alpha_+$
is retained only to describe that finite collection's coverage.
(ii) follows from the algebraic point certificates in \cref{sec:face}.
(iii) follows from the rational-function certificate and its exact
Bernstein positivity proof in \cref{sec:interval}. The independent
verifiers check the original free-moment identities and all PSD blocks.
For the standard degree, \cref{thm:L1-std} gives
$2\sqrt2+\alpha>q(\alpha)$, while the standard degree-two certificate is
the symbolic tilted-CHSH SOS in \cref{lem:standard-tilted}.
\end{proof}

\begin{remark}[The critical tilt as a numerical threshold]
\label{rem:threshold-numeric}
Numerically, closure transitions at $\alpha^\ast\approx1.28427$ (bisection
with two independent solvers): below it the Alice-conditioned level-two value is
not distinguishable from the quantum value at our numerical precision,
above it a strictly positive gap opens continuously.  We do not prove the
existence of a single transition point: both $\omega_2^{\rm os}(F_\alpha)$
and $\sqrt{8+2\alpha^2}$ are convex in $\alpha$, so in principle their
difference could oscillate; the certified content of this section is
exactly items (i)--(iii) of \cref{thm:L2-threshold}, and the transition
picture should be read as a numerical observation.  Note also that the
band $(41/32,\alpha_-)$ has no complete
classification from our certificates. This band is specified by an
exact closure point and an exact coverage boundary, without using the
numerical $\alpha^\ast$ as a certified endpoint.
\end{remark}

\subsection{Exact optimal-face certificates}
\label{sec:face}
Let $q=\sqrt{8+2\alpha^2}$, $c=q/4$, $d=\sqrt{1-c^2}$,
$z=\alpha/(2c)$, and $t=\sqrt{1-z^2}$. The optimal strategy uses
$A_0=Z$, $A_1=X$, $B_0=cZ+dX$, $B_1=cZ-dX$ and
\begin{equation*}
\rho_\alpha=\frac12\begin{pmatrix}
1+z&0&0&t\\0&0&0&0\\0&0&0&0\\t&0&0&1-z
\end{pmatrix}.
\end{equation*}
The verifier checks $\rho_\alpha^2=\rho_\alpha$,
$\operatorname{tr}\rho_\alpha=1$, and
$\operatorname{tr}(\rho_\alpha F_\alpha)=q$ exactly.
All strategy and certificate matrices here are real. Taking real parts
of a Hermitian feasible dual preserves PSD and the real affine identity,
so there is no loss in using real matrices. We use $T$ for their transpose;
$\dagger$ continues to denote the adjoint of an operator word.

Put $r=(1-z)/t$. Up to a positive normalization of their outer products,
the conditional Bob vectors are $(1,0)^T,(0,1)^T,(1,r)^T,(1,-r)^T$.
For each vector $v$, form $V$ with columns $wv$ over the retained
Bob-projector word basis. Its Gram matrix is proportional to the quantum
block. At all five points $V$ has rank two. Compute a basis $K$ of
$\ker V$ by exact algebraic elimination.

The dual identity has the form
\begin{equation*}
q-F(m)=\sum_{a,x}\operatorname{tr}(S_{a|x}\Phi_{a|x}(m))
+\text{consistency and normalization terms}.
\end{equation*}
Evaluating an optimal dual on the quantum strategy gives
$\sum_{a,x}\operatorname{tr}(S_{a|x}\Phi^q_{a|x})=0$.
Each summand is nonnegative. Therefore
$S_{a|x}\Phi^q_{a|x}=0$, and all attaining duals have
$S_{a|x}=K_{a|x}H_{a|x}K_{a|x}^{T}$ with $H_{a|x}\succeq0$.
Conversely this factorization, together with the original identity,
proves the bound $q$. Alice-conditioned dual attainment in \cref{lem:attainment} makes this an if-and-only-if criterion.

The new certificates have the following dimensions and exact targets:
\begin{center}
\begin{tabular}{ccccc}
\toprule
$\alpha$ & level & $q$ & Gram size & reduced size\\
\midrule
$1$ & $2$ & $\sqrt{10}$ & $5$ & $3$\\
$5/4$ & $2$ & $\sqrt{178}/4$ & $5$ & $3$\\
$41/32$ & $2$ & $\sqrt{11554}/32$ & $5$ & $3$\\
$13/10$ & $3$ & $\sqrt{1138}/10$ & $7$ & $5$\\
$3/2$ & $3$ & $5\sqrt{2}/2$ & $7$ & $5$\\
\bottomrule
\end{tabular}
\end{center}
The two level-three rows are independent algebraic cross-checks of the
endpoints. The continuous certificate in \cref{sec:interval} supersedes
them as the proof of level-three closure; they are retained to check the
rational-function construction against a separate number-field route.
Exact elimination leaves nine free coordinates at each level-two point
and 36 at each level-three point. Only these free coordinates are
rationalized after numerical search. The dependent coordinates are
computed in the exact number field, so the affine identity is preserved.

Each JSON file in \datapath{data/analytic/quantum_face/} stores the exact
$K,H$, Gram entries, and all nonzero consistency and normalization
multipliers. The independent verifier does not trust the generator's
reduced identity, rank, numerical eigenvalues, or status flag. It checks
the original free-moment identity coefficient-wise, verifies $VK=0$ and
$S=KHK^T$, and tests the reduced matrices by exact $LDL^T$ arithmetic.
The common solver-free driver is \codepath{scripts/run_reproduction.py}:
\begin{verbatim}
python scripts/run_reproduction.py
\end{verbatim}
It invokes \codepath{scripts/verify_certificates.py} as its core
certificate checker, together with the auxiliary exact checks described
in \cref{sec:repro-commands}. To run only the core checker, use
\begin{verbatim}
python scripts/verify_certificates.py
\end{verbatim}
For the five optimal-face certificates, use \codepath{scripts/quantum_face.py}:
\begin{verbatim}
python scripts/quantum_face.py
\end{verbatim}
For the generator \codepath{scripts/run_quantum_face.py}, use
\begin{verbatim}
python scripts/run_quantum_face.py --alpha 1 --level 2
\end{verbatim}
Use the other four table entries to regenerate their certificates.
Numerical optimization is used only in generation. During verification,
\texttt{cvxpy.Problem.solve} is patched to raise an exception.

\subsubsection{An exact continuous certificate family}
\label{sec:interval}
The interval proof uses a rational parameter for the quantum-bound curve:
\begin{equation}
 c(u)=\frac{u^2-4u+8}{u^2-8},\qquad
 \alpha(u)=\frac{-2(u^2-8u+8)}{u^2-8},\qquad q(u)=4c(u).
 \label{eq:interval-param}
\end{equation}
On $J=[43/10,447/100]$, $0<c<1$, $0<\alpha<2$, and
$q^2=8+2\alpha^2$. Since
$\alpha'(u)=-16(u^2-4u+8)/(u^2-8)^2<0$, its image is
\begin{equation}
 \alpha(J)=\left[\frac{155582}{119809},\frac{1582}{1049}\right]
 \supset\left[\frac{13}{10},\frac32\right].
 \label{eq:interval-image}
\end{equation}

To remove square roots from the kernel calculation, conjugate the physical
Bob projectors by $\operatorname{diag}(1,d)$, where $d^2=1-c^2>0$:
\begin{equation*}
 \widetilde P_0=\frac12\begin{pmatrix}1+c&1\\1-c^2&1-c\end{pmatrix},\qquad
 \widetilde P_1=\frac12\begin{pmatrix}1+c&-1\\-(1-c^2)&1-c\end{pmatrix}.
\end{equation*}
The transformed conditional vectors can be scaled to
$(1,0)^T,(0,1)^T,(1,c-\alpha/2)^T,(1,-c+\alpha/2)^T$.
This invertible change affects neither the word-map kernels nor their
dimension. Form $V_{a|x}(u)$ from the seven level-three Bob words and these
vectors. Write $V_{a|x}=(V^{(0)}_{a|x}\ V^{(1)}_{a|x})$, where
$V^{(0)}_{a|x}$ contains the first two columns. They give the rational basis
\begin{equation*}
 K_{a|x}(u)=\begin{pmatrix}
 -(V^{(0)}_{a|x})^{-1}V^{(1)}_{a|x}\\ I_5
\end{pmatrix}.
\end{equation*}
In the displayed vector order, the four determinants of $V^{(0)}$ are
\begin{equation}
 \frac{4u(u-4)(u-2)}{(u^2-8)^2},\quad -\frac12,\quad
 -\frac{4(u-4)^2(u-2)^2}{(u^2-8)^2},\quad
 \frac{8u(u-4)(u-2)}{(u^2-8)^2}.
 \label{eq:interval-minors}
\end{equation}
All are nonzero on $J\subset(4,\infty)$. Thus every $V$ has rank two
and every kernel has dimension five throughout $J$, without a rank-jump
assumption. The verifier reconstructs the four word maps, checks $VK=0$
as rational-function identities, and excludes zeros of these minors on
the entire interval.
Exact elimination of the original dual
identity with $S_{a|x}=K_{a|x}H_{a|x}K_{a|x}^T$ gives a rational affine
solution with 36 free coordinates. We choose each free coordinate as a
quadratic polynomial in $(u-u_0)/r$, where $u_0=877/200$ and $r=17/200$.
Numerical optimization selects their coefficients on nine search nodes;
rounding to a common rational grid and exact substitution give a candidate
identity at every $u$. Search nodes are not used as a positivity proof.
They are not used to verify the identity either: after expansion in the
independent free-moment coordinates, each of its 41 coefficients is
checked as an identity in $\QQ(u)$ by reduction to a zero numerator
polynomial. This includes all Gram and consistency/normalization
multiplier contributions and does not substitute sample values for $u$.

The resulting four matrices have the form $H_{a|x}(u)=N_{a|x}(u)/D(u)$,
where each $N_{a|x}$ is a symmetric rational-coefficient matrix polynomial
of degree at most 20 and
\begin{equation}
 D(u)=u(u-4)(u-2)(u^2-8)^4(u^2-4u+8)^3>0\qquad(u\in J).
 \label{eq:positive-denominator}
\end{equation}
Write $s=(100u-430)/17\in[0,1]$. The degree-20 Bernstein expansions are
\begin{equation}
 N_{a|x}(u)=\sum_{j=0}^{20} C_{a|x,j}\binom{20}{j}s^j(1-s)^{20-j}.
 \label{eq:bernstein-interval}
\end{equation}
Every one of the $4\times21=84$ rational matrices $C_{a|x,j}$ is PSD,
as verified by exact $LDL^T$ arithmetic. The weights are nonnegative and
sum to one, hence $N_{a|x}(u)\succeq0$ and $H_{a|x}(u)\succeq0$ on all of
$J$. Together with the exact original identity this proves
$\omega_3^{\rm os}(F_{\alpha(u)})\leq q(u)$. The attaining quantum strategy
provides the reverse bound. Equation~\eqref{eq:interval-image} then proves
the certified full closed interval; combining it with the standard SOS and
the level-two tangent fan gives exactly $d_{\rm std}=2$, $d_{\rm os}=3$
throughout $[13/10,3/2]$.
In fact all 84 matrices are nonsingular and hence positive definite;
their determinants are checked exactly. The reduced family therefore
has a uniform strictly positive margin on $J$.

The artifact \datapath{data/analytic/interval/level3_interval.json} stores
$K$, the numerator polynomials, and the original Gram and multiplier
functions. The independent checker reconstructs the full moment identity
from the raw algebra, checks $S=KNK^T/D$, and derives all 84 Bernstein
matrices itself. It also verifies the word-map kernels and their constant
rank, and certifies that every reduced rational-function
denominator is nonzero on $J$. Neither a sampled eigenvalue nor a stored
success flag is accepted as proof. Commands are
\begin{verbatim}
python scripts/quantum_interval.py
python scripts/run_quantum_interval.py
\end{verbatim}
The first runs the solver-free checker \codepath{scripts/quantum_interval.py};
the second uses \codepath{scripts/run_quantum_interval.py} to regenerate a
candidate and verify the entire interval before writing it.

\subsubsection{When the interval procedure applies}
\label{sec:interval-general}
The kernel reduction and the parameter-interval certificate serve
different purposes: the first identifies a necessary face of the dual
cone, while the second gives a finite positivity proof for a whole
family on that face. Polynomial selections for parameter-dependent
linear matrix inequalities are established by Bliman \cite{Bliman2004},
and Bernstein positivity has a classical certificate literature
\cite{BoudaoudCarusoRoy2008}. The following specialization combines those
principles with exact rational affine elimination. It is a sufficient
condition for this certificate format, not a claim that these component
methods originate here.

\begin{proposition}[Finite rational interval certificates]
\label{prop:interval-general}
Let $J=[a,b]$ have rational endpoints $a<b$, and fix a native hierarchy
level. Suppose fixed full-column-rank rational kernel bases $K_b(u)$ have
no poles on $J$ and the resulting reduced dual affine identity has a parameterization
\begin{equation}
 z(u,t)=z_0(u)+Z(u)t,\qquad z_0,Z\in\QQ(u),
 \label{eq:affine-parameterization}
\end{equation}
with no poles on $J$, where $z$ collects the reduced Gram entries and
equality multipliers. For each $u\in J$, assume some $t\in\mathbb R^r$
makes every reduced Gram block $H_b(u,t)$ positive definite.
Then there exist $t(u)\in\QQ[u]^r$, a polynomial $D(u)>0$ on $J$, and
a finite degree $n$ such that the identity holds exactly in $\QQ(u)$,
all $N_b(u)=D(u)H_b(u,t(u))$ are polynomial matrices, and their degree-$n$
Bernstein coefficient matrices on $J$ are positive definite.
\end{proposition}
\begin{proof}
At a fixed $u_i$, choose strictly feasible free coordinates $t_i$.
Continuity of the pole-free affine maps makes $t_i$ feasible on a relative
open neighborhood of $u_i$. A finite subcover of $J$ and a subordinate
continuous partition of unity $\lambda_i(u)$ give
$t_*(u)=\sum_i\lambda_i(u)t_i$. Because the Gram maps are affine in $t$,
$H_b(u,t_*(u))=\sum_i\lambda_i(u)H_b(u,t_i)\succ0$.
There are finitely many blocks, so compactness yields a uniform margin
$\delta>0$ for their smallest eigenvalues.

Approximate the continuous vector $t_*$ uniformly by a rational-coefficient
polynomial vector $p$. For example, its Bernstein approximants converge
uniformly, and their finitely many real coefficients can be approximated
by rationals. The linear part of each Gram map is uniformly bounded on
$J$, so a sufficiently accurate approximation leaves
$H_b(u,p(u))\succeq(\delta/2)I$. The affine identity remains exact by
\eqref{eq:affine-parameterization}. Multiplying by the product of the
squares of its nonvanishing denominator polynomials gives a common
rational-coefficient $D>0$ with polynomial numerators, still uniformly
positive definite.

Rescale to $s=(u-a)/(b-a)\in[0,1]$ and write
$N(s)=\sum_{\ell=0}^{d} A_\ell s^\ell$. Its degree-$n$ Bernstein
coefficients, for $n\geq d$, are
\begin{equation}
 C_{j,n}=\sum_{\ell=0}^{\min(j,d)}
 A_\ell\frac{\binom{j}{\ell}}{\binom{n}{\ell}},
 \qquad 0\leq j\leq n.
 \label{eq:bernstein-elevation}
\end{equation}
For fixed $d$, the scalar ratio in this formula differs from
$(j/n)^\ell$ by $O(1/n)$ uniformly over $j$, taking the ratio to be zero
when $j<\ell$. Hence
$\max_j\|C_{j,n}-N(j/n)\|\to0$. The uniform positive margin makes every
$C_{j,n}$ positive definite once $n$ is sufficiently large. Choose a
common such $n$ for the finitely many blocks.
\end{proof}

The native hierarchy level is fixed in this proposition. Increasing the
degree of the parameter polynomial or its Bernstein representation does
not introduce longer operator words. The result ensures existence of a
finite certificate under the stated hypotheses; it guarantees neither
the success of our quadratic search nor a degree bound. Known attainment
and an algebraic quantum value alone are insufficient: algebraicity does
not ensure a rational parameterization, and the strategy kernel need not
be the smallest feasible face. If reduced Grams remain singular, further
face reduction or another positivity argument may be necessary.

For a simple illustration of degree elevation, the polynomial matrix
$N(s)=\operatorname{diag}((s-1/2)^2+1/16,1)$ is positive definite on
$[0,1]$. Its degree-two middle coefficient has first entry $-3/16$,
whereas all six degree-five coefficients are positive definite (their
smallest diagonal entry is $1/80$). Thus failure at a chosen Bernstein
degree is not a proof of infeasibility. The corresponding statement with
only semidefinite feasibility would be false: $(s-1/2)^2$ vanishes at an
interior point, where every Bernstein weight is positive, so it cannot
have all nonnegative coefficients at any finite degree.

\subsection{A necessary quantum bound on a critical path}
\label{sec:critical-mixed-path}
To investigate a genuinely changing gap, consider
\begin{align}
 \widehat G_{\delta,h}={}&(1-\delta)(A_0+B_0)
 -(\tfrac12+h)(A_1+B_1)\\
 &+A_0B_0+A_0B_1+A_1B_0-\tfrac12 A_1B_1,
 \qquad \widehat\beta_\delta=\tfrac72-2\delta.\nonumber
\end{align}
The two deterministic assignments in \eqref{eq:mixed-eight-characters} attain
$\widehat\beta_\delta$, but this alone does not establish a quantum
upper bound. Put $X=(I-A_0)/2$, $Y=(I-B_0)/2$,
$U=(I-A_1)/2$, and $V=(I-B_1)/2$. Direct expansion gives
\begin{align}
 P_{\delta,h}:=\frac{\widehat\beta_\delta I-\widehat G_{\delta,h}}{2}
 &=H_0-\delta(X+Y)+h(I-U-V),\\
 H_0&=(X-Y)^2+2X(I-V)+2Y(I-U)+UV.
 \label{eq:mixed-critical-path}
\end{align}
At $\delta=h=0$, this is a sum of projection squares. Approaching
this singular boundary requires care: numerical feasibility can fail
to resolve the terms that determine the target's validity.

\begin{proposition}[Necessary cubic quantum bound]
\label{prop:mixed-cubic-necessary}
If $0<\delta<3/4$ and $P_{\delta,h}$ is nonnegative in every quantum
strategy, then
\begin{equation}
 h\ge \frac{2\delta^3}{3-4\delta}.
 \label{eq:mixed-cubic-necessary}
\end{equation}
In particular, a path $h=C\delta^3$ with $C\le2/3$ cannot have
$\widehat\beta_\delta$ as its quantum value for any such $\delta$.
\end{proposition}

\begin{proof}
For $0<t<3/4$, let $s=\sqrt{t(1-t)}$ and use the local projectors
\[
 X=Y=\begin{pmatrix}0&0\\0&1\end{pmatrix},\qquad
 U=V=\begin{pmatrix}t&s\\s&1-t\end{pmatrix},
\]
where each is understood to act on its own party's qubit. In the
product computational basis choose the nonzero, unnormalized state
\[
 \psi_t=\bigl(-(1-t)(3+2t),\ 2ts,\ 2ts,\ t(3-2t)\bigr)^{\mathsf T}.
\]
One checks $UV\psi_t=0$ and the exact identities
\begin{align}
 \langle H_0\rangle_{\psi_t}&=4t^3(3-2t),\\
 \langle X+Y\rangle_{\psi_t}&=2t^2(9-8t),\\
 \langle I-U-V\rangle_{\psi_t}&=(3-4t)^2.
\end{align}
Setting $t=\delta$ therefore gives
\[
 \langle P_{\delta,h}\rangle_{\psi_\delta}
 =(3-4\delta)\bigl[h(3-4\delta)-2\delta^3\bigr].
\]
Nonnegativity of this expectation proves \eqref{eq:mixed-cubic-necessary}.
Normalizing the state does not change its sign. The symbolic checker
\codepath{research/mixed_divergence/verify_corner_obstruction.py} verifies these identities
coefficient by coefficient modulo $s^2=t(1-t)$, without numerical
optimization.
\end{proof}

The inequality is a necessary condition, not a sufficient quantum bound
or a fixed-degree SOS construction. It supplies an exact preliminary
test for candidate critical paths: a reported feasible point below
this bound is not quantum-positive, regardless of solver status.
For example, at $\delta=1/160$ the necessary lower bound is
$h\geq1/6092800>0$.

The remaining proof obligations are distinct. One must first establish
an attained quantum target and a uniform standard degree along a
critical path. One must then construct, for every fixed $k$, a
\emph{common} functional positive on both conditioned cones and negative
on that same gap at a suitable parameter. Two separate one-direction
witnesses do not meet the second obligation. Neither obligation is
settled for \eqref{eq:mixed-critical-path} by the cubic necessary bound.

\subsection{How the finite-level limitations differ}
\label{sec:qualitative}
At level one, Alice-conditioned blocks enforce positivity of joint probabilities
and carry total-degree-three moments absent from the standard raw
level-one truncation. At native level two the blocks even carry selected
total-degree-five moments $A_{a|x}u^\dagger v$, with Bob words
$|u|,|v|\leq2$, whereas standard level two stops at total degree four.
For instance $u=P_{B,0}P_{B,1}$ and $v=P_{B,1}P_{B,0}$ give the reduced
word $A_{a|x}P_{B,1}P_{B,0}P_{B,1}P_{B,0}$. This maximum-degree comparison
does not imply containment of truncated moment sets: the blocks lack
mixed-Alice-question words. The exact separating witnesses show
that a standard quantum-bound certificate cannot always be represented
by a single-question SOS at that level. The optimal-face criterion gives
a complete feasibility test for the latter question. The occurrence of
cross-question terms in a particular numerical standard Gram matrix is
not an invariant obstruction, since Gram representations are nonunique.
\Cref{sec:structure} supplies an approximation theorem and a scalar geometric test for uniform closure
and a fixed-size repair for tilted CHSH. A necessary and sufficient
classification and a minimal repair remain open.

\section{Reproducibility and independent verification}
\label{sec:repro}
\subsection{Checks accompanying the article}
The supplied certificate package contains the coupled witnesses, the
four certificates for the mixed degree pair, the relative-accuracy
audit and the figure-generation script. From its root, run
\begin{verbatim}
python verify_revision_results.py
python scripts/verify_relative_precision.py
python scripts/make_story_figures.py
\end{verbatim}
The first command runs eight exact checks. The second verifies the
identities, local bound, word counts and constants used in
\cref{cor:relative-precision}. The third draws analytic quantities and
checks all plotted relative-accuracy thresholds at 80-digit precision;
no SDP solver is used to produce the three main figures.
\begin{center}
\small
\begin{tabular}{@{}p{.36\textwidth}p{.59\textwidth}@{}}
\toprule
Result & Checker, relative to the package root\\
\midrule
Full coupled lift, quantum target and standard minimal degree &
\codepath{research/bidirectional_coupling/verify_bidirectional_coupling.py}\\[4pt]
Tunable bridge and independent-resource comparison &
\codepath{research/bidirectional_coupling/verify_scope_extensions.py}\\[4pt]
Relative-accuracy corollary &
\codepath{scripts/verify_relative_precision.py}\\[4pt]
Standard degree-three upper and mixed level-seven lower certificates &
\codepath{research/mixed_divergence/verify_degree_pair.py}\\[4pt]
Standard level-two lower certificate &
\codepath{research/mixed_divergence/verify_standard_lower.py}\\[4pt]
Mixed degree-eight upper certificate &
\codepath{research/mixed_divergence/verify_mixed_upper.py}\\
\bottomrule
\end{tabular}
\end{center}
The mixed checkers use Python's standard library and exact fractions.
Run \nolinkurl{verify_degree_pair.py --mutation-checks} in its directory
to test rejection of deliberately corrupted certificates. The coupled
and symbolic checks use SymPy; other dependencies are pinned in
\nolinkurl{requirements-lock.txt}. Numerical search is separate from
acceptance of each rational identity or positive Gram matrix.
The SHA-256 manifest identifies all supplied files.

\subsection{Analytic and additional finite-level data}
\label{sec:repro-commands}
Code version \texttt{v0.3.0}~\cite{Wang2026Code} and the separate data
archive version 0.3.0~\cite{Wang2026Data} support all results in this article.
Extract the data ZIP into the matching code root, so that
\nolinkurl{data/}, \nolinkurl{research/} and \nolinkurl{figures/} sit alongside
\nolinkurl{scripts/} and \nolinkurl{src/}. The JSON data and Python checkers
share directories only after this extraction; the public code and data
archives are separate. Use the public directories below throughout.
\begin{center}
\small
\begin{tabular}{@{}p{.43\textwidth}p{.52\textwidth}@{}}
\toprule
Material & Public data directory\\
\midrule
Contact geometry and tangent fan & \datapath{data/contact_geometry/}\\[3pt]
Finite separations & \datapath{data/finite_separations/}\\[3pt]
Analytic examples and randomness & \datapath{data/analytic/}\\[3pt]
Level-one certificates and families & \datapath{data/level_one_exact/}, \datapath{data/level_one_family/}\\[3pt]
Convention controls & \datapath{data/convention_checks/}\\[3pt]
Conditioned closure certificates & \datapath{data/conditioned_closure/}\\
\bottomrule
\end{tabular}
\end{center}
Before running any script that regenerates a report, verify the pristine
archive and then run the analytic and additional finite-level checks:
\begin{verbatim}
python scripts/check_data.py
python scripts/run_reproduction.py
\end{verbatim}
The latter runs six stages, including all 28 groups in the core
\codepath{scripts/verify_certificates.py}, without SDP solving in that checker.
The current coupled and mixed checks are listed above. The relative-accuracy
script verifies the quantitative consequence of the same coupled
construction, not an additional independent witness. The README records
dependencies and individual analytic-check commands. All code links point
to the fixed GitHub tag; data links point to the matching Zenodo record.

\subsection{Convention and numerical checks}
\label{sec:full}
Two implementations compare an outcome-eliminated PVM quotient with
an outcome-complete representation that retains all completeness
relations. Their Bell-coordinate probes agree at levels one and two
on the control instances (worst deviation $2.5\times10^{-7}$).
This is a numerical convention audit on those instances. The exact
constructions use the explicitly defined quotient and full outcome
blocks; their universal proofs do not follow from the numerical probes.
The control data and reproduction commands are in the public archives.

\paragraph{Finite-tolerance observations.}
\label{sec:degree-survey}
Numerical saturation at a chosen tolerance should not be interpreted as
an exact closing level. The optional archived
\codepath{scripts/survey_conditioned_degree.py} records solver status,
residuals and adjacent-level comparisons. The asymptotic degree and
accuracy statements in this article are established by the analytic
upper constructions and feasible lower witnesses, not by a fitted
convergence curve.

\section{Relation to other hierarchies and open questions}
\label{sec:open}

\paragraph{Boundary.}
Finite-level SDP attainment in Lemma~\ref{lem:attainment} is distinct from
finite convergence to the commuting-operator value. Fanizza
\emph{et al.}\ \cite{FanizzaEtAl2025} construct games for which no finite
NPA level reaches that limiting value. Their result does not assert a
duality gap or failure of an optimizer to exist for each finite SDP.
The degree-one conversion $\mathsf D_1\subseteq\mathsf O_1$ is
\cite[Theorem~5.3]{CuiFalorNatarajanZhang2025} (v1, 2025); the same
preprint's Section~5.2 level-one value-equality statement is refuted as
stated, under its raw v1 convention, by
\cref{cor:cfnz} of this supplement (not a corollary of that preprint).
The conversion theorem itself stands and is used here as
\cref{prop:d1}.  The sequential-NPA finite-level convergence and flatness
results are \cite{KlepEtAl2025} (v2, 2026); we keep its POVM/localizer
filtration distinct from the PVM quotient used here and make no claim
about its convention (\cref{rem:audit}).  In July 2026 three
parallel works established, from different directions, that the standard
NPA hierarchy is not exact at finite level near the critical doubly-tilted
family: Pakhunov's companion
pair \cite{Pakhunov2026,Pakhunov2026Companion} (no finite level is exact
near the critical tilt; a phase transition in exactness) and Chaturvedi's
\cite{Chaturvedi2026} (no finite level characterizes the complete quantum
set in the simplest Bell scenario). Pakhunov constructs feasible standard
NPA moment curves with a quadratic excess above the local value where
the quantum advantage is cubic. Chaturvedi obtains unbounded error
relative to that cubic advantage using a Motzkin-polynomial obstruction
and deduces strict containment of the quantum set by every fixed-word
standard relaxation. These are obstructions within the standard hierarchy.
Our
\cref{thm:unbounded-sm} concerns a distinct obstruction: standard level
two is exact throughout the family, while no fixed finite conditioned
level is exact throughout it. This does not assert that a single
fixed tilted-CHSH functional remains nonexact at every finite level.
The
quantum values of tilted and doubly-tilted families are from
\cite{Acin2012Tilted,GigenaEtAl2025}; the level-$1{+}AB$ SOS decomposition
is from \cite{BampsPironio2015}; the almost-quantum comparison level is
from \cite{Navascues2015AlmostQuantum}.

\paragraph{Positive-definite extension literature.}
Bakonyi and Timotin \cite[Proposition~3.4]{BakonyiTimotin2011} prove
positive-definite extension from symmetric word balls on the infinite
dihedral group. This concerns one partially defined positive-definite
function. It does not impose common outcome sums on four simultaneous
extensions. In our path setting, every set of vertices of diameter at
most $2k$ lies in a translate of $W_k$, so each individual truncated
block meets their positivity hypothesis and has an extension.
The Bell witness instead obstructs choosing four extensions with equal
outcome sums on the full algebra: such a choice would be an all-level
feasible assemblage and obey the quantum bound by hierarchy convergence.
Thus our contribution is the explicit compatible truncated assemblage,
its separation from a standard degree-two exact target, and its degree
scaling. Neither the general extension question nor the elementary
rank-one PSD criterion is claimed as a new result.

\paragraph{Fixed-party comparison with compiled tilted expressions.}
The single-question SOS of Mehta, Paddock, and Wooltorton
\cite[Eqs.~(3.15)--(3.17)]{MehtaPaddockWooltorton2025} certifies an
extended expression $S_{\theta,\phi}$ with a Bob marginal proportional
to $I\otimes(B_0+B_1)$. It should not be identified with our
$F_\alpha$, whose marginal is $\alpha A_0\otimes I$, from their shared
attaining strategy alone. For example, write $s=\sin(2\theta)$,
$c=\cos(2\theta)$, and $\tan\mu=s$. Direct substitution in that
reference's expression gives
\begin{equation}
 s^2\cos\mu\,S_{\theta,\mu}
 =s^2A_0(B_0+B_1)+A_1(B_0-B_1)+cI(B_0+B_1).
\end{equation}
For $0<\theta<\pi/4$ these Bell coefficients differ from those of
$F_\alpha$ with the parties fixed. Thus the displayed SOS for $S$
does not itself provide an $\mathsf O_k$ certificate for $qI-F_\alpha$;
an explicit conversion identity would be needed. Our degree obstruction
is not a claim of failure of the compiled soundness theorem for $S$.
More directly, \cref{prop:opposite-tilt} supplies a level-one certificate
for the party-exchanged singly tilted expression itself. These facts limit the original singly tilted obstruction to a specified
conditioning direction. The coupled construction in
\cref{thm:bridge-ratio} subsequently excludes a global party swap,
but its mixed-direction decomposition still prevents a claim about
every proof method or compilation order.

\paragraph{Open questions.}
The central unresolved extension is double directional divergence with
a bridge coefficient bounded away from zero. The present relative-accuracy
law keeps the bridge proportional to the endpoint distance; it does not
establish a fixed-noise separation. A second question is whether mixed
degree can diverge while standard degree remains bounded. The coupled
family has mixed degree one, and the finite examples have uniform mixed
upper bounds on their positive parameter rays. The extreme-ray reduction
in \cref{prop:bell-splitting} describes a potential route, without
supplying such a family. A third task is to relate certificate degree
and coefficient norms to a specified protocol's quantitative soundness
estimate. The algebraic approximation cost is identified here, while
that protocol-dependent conversion requires additional control.

For the contact mechanism itself, further questions concern optimal
leading constants, wider coefficient perturbations and extensions of
the uniform-closure construction. These are distinct from the matched
exponents and exact finite degree minima proved in this supplement.

\bibliographystyle{alpha}
\bibliography{references}